\documentclass[a4paper,11pt,DIV=calc]{scrartcl}
\usepackage{hyperref}
\usepackage[utf8]{inputenc}
\usepackage{amsmath,amsfonts,amssymb,amsthm,mathtools}
\usepackage{amsmath}
\usepackage{amsfonts}
\usepackage{amssymb}
\usepackage{dsfont}
\usepackage{mathrsfs}
\usepackage[T1]{fontenc}
\usepackage[ngerman,english]{babel}
\usepackage[inline]{enumitem}
\usepackage{latexsym,graphicx}
\usepackage[all]{xy}
\usepackage{mathtools}
\usepackage{color}
\usepackage{authblk}
\usepackage{amscd}
\usepackage{microtype}
\usepackage{mathrsfs}
\numberwithin{equation}{section}

\usepackage{txfonts}

\DeclareMathOperator{\tr}{Tr}
\DeclareMathOperator{\trs}{tr}
\DeclareMathOperator{\var}{\mathbf{Var}}

\newtheorem{theorem}{Theorem}[section]

\newtheorem{lemma}{Lemma}[section]

\theoremstyle{definition}
\newtheorem{remark}{Remark}[section]

\newcommand{\dda}{\mathrm{d}}
\newcommand{\de}{\,\dda}

\theoremstyle{definition}

\newcommand{\beq}{\begin{equation}}
\newcommand{\eeq}{\end{equation}}
\usepackage{pdfsync}

\begin{document}

\title{Gross-Pitaevskii Limit of a Homogeneous Bose Gas at Positive Temperature}

\author{Andreas Deuchert, Robert Seiringer}

\date{}

\maketitle

\begin{abstract} 
We consider a dilute, homogeneous Bose gas at positive temperature. The system is investigated in the Gross-Pitaevskii limit, where the scattering length $a$ is so small that the interaction energy is of the same order of magnitude as the spectral gap of the Laplacian, and for temperatures that are comparable to the critical temperature of the ideal gas. We show that the difference between the specific free energy of the interacting system and the one of the ideal gas is to leading order given by $4 \pi a \left( 2 \varrho^2 - \varrho_0^2 \right)$. Here $\varrho$ denotes the density of the system and $\varrho_0$ is the expected condensate density of the ideal gas. Additionally, we show that the one-particle density matrix of any approximate minimizer of the Gibbs free energy functional is to leading order given by the one of the ideal gas. This in particular proves Bose-Einstein condensation with critical temperature given by the one of the ideal gas to leading order. One key ingredient of our proof is a novel use of the Gibbs variational principle that goes hand in hand with the c-number substitution.
\end{abstract}

\setcounter{tocdepth}{2}
\tableofcontents

\section{Introduction and main results}
\label{sec:introductionandmainresults}
\subsection{Background and summary}
\label{sec:backgroundandsummary}
The experimental realization of the first Bose-Einstein condensate (BEC) in an alkali gas in 1995 \cite{WieCor1995, Kett1995} triggered numerous mathematical investigations on the properties of dilute Bose gases. The starting point was a work by Lieb and Yngvason \cite{LiYng1998} who proved a lower bound for the ground state energy of a dilute Bose gas in the thermodynamic limit. Together with the upper bound given in \cite{RobertGPderivation}, it rigorously establishes its leading order behavior. In the case of hard-core bosons, the correct upper bound had already been proven in 1957 by Dyson \cite{Dyson}. Also the next-to-leading order correction to the ground state energy predicted by Lee, Huang and Yang in 1957 \cite{LeeHuangYang1957} could recently be proven, see \cite{YauYin2009} for the upper bound and \cite{FournSol2019} for the lower bound.

Bose gases in experiments are usually prepared in a trapping potential and such a set-up is well-described by the Gross-Pitaevskii (GP) limit. As shown in \cite{RobertGPderivation,LiSei2002,LiSei2006,NamRouSei2016}, the ground state energy of a Bose gas in this limit is to leading order given by the minimum of the GP energy functional. Additionally, a convex combination of projections onto the minimizers of this functional approximates the one-particle density matrix of the gas to leading order. Also in the GP limit the next to leading order correction to the ground state energy predicted by Bogoliubov in 1947 could be justified \cite{BogGP}. The accuracy reached in this work allows for an approximate computation of the ground state wave function and for a characterization of the low lying excitation spectrum. The dynamics of a system in the GP limit, on the other hand, can be described by the time-dependent GP equation, which was established in \cite{ErdSchlYau2009,ErdSchlYau2010,BenOlivSchl2015,Pickl2015}. For a more extensive list of references we refer to \cite{Themathematicsofthebosegas,Rou2015,BenPorSchl2015}.

While ground states provide a good description of quantum gases at very low temperatures, positive temperature effects are crucial for a complete understanding of modern experiments. In this case one is interested in the free energy and the Gibbs state of the system rather than in its ground state energy and in the ground state wave function. For the dilute Bose gas in the thermodynamic limit, the leading order behavior of its free energy per unit volume has been established, see \cite{Yin2010} for the upper bound and \cite{Sei2008} for the lower bound. The techniques developed in \cite{LiYng1998,RobertGPderivation} have also been extended to treat fermions, both for the ground state energy \cite{LiSeiSol2005} and for the free energy at positive temperature \cite{RobertFermigas}. We mention also the papers \cite{LewinNamRougerie2015,LewinNamRougerie2017,LewinNamRougerie2018,FKSS2017} and \cite{FKSS2020} where Gibbs states of Bose gases with mean-field interactions are studied.

In a  recent work \cite{me}, the trapped Bose gas at positive temperature is studied in a combination of thermodynamic limit in the trap and GP limit. It was shown that the difference between the  free energy of the interacting system and the one of the ideal gas is to leading order given by the minimum of the GP energy functional. Additionally, the one-particle density matrix of any approximate minimizer of the Gibbs free energy functional is to leading order given by the one of the ideal Bose gas, but with the condensate wave function replaced by the minimizer of the GP functional. This in particular proves the existence of a BEC phase transition in the system. The proof of these statements relies heavily on the fact that particles in the thermal cloud have a much larger energy per particle, and therefore live on a much larger length scale than particles in the condensate. As a consequence, the interaction can be seen to leading order only in the condensate. The case of the homogeneous gas in a box, where the condensate and the thermal cloud live on the same length scale, was left as an open problem. 

In the present work we consider this case, that is, we consider a homogeneous Bose gas (a gas in a box) at positive temperature in the GP limit. In this system the condensate and the thermal cloud necessarily live on the same length scale and interactions between them are relevant. We prove similar statements as in the case of the trapped gas in \cite{me}, in particular, we show the existence of a BEC phase transition with critical temperature given by the one of the ideal gas to leading order.
\subsection{Notation}
\label{sec:notation}
For functions $a$ and $b$ of the particle number and other parameters of the system, we use the notation $a \lesssim b$ to say that there exists a constant $C>0$ independent of the parameters such that $a \leq C b$. If $a \lesssim b$ and $b \lesssim a$ we write $a \sim b$, and $a \simeq b$ means that $a$ and $b$ are equal to leading order in the limit considered. 
\subsection{The model}
\label{sec:model}
We consider a system of $N$ bosons confined to a three-dimensional flat torus $\Lambda$ of side length $L$ (we could set $L=1$ but we prefer to keep a length scale to explicitly display units in formulas). The one-particle Hilbert space is thus $\mathcal{H} = L^2(\Lambda,\de x)$, with $\de x$ denoting Lebesgue measure, and the Hilbert space of the $N$-particle system is the $N$-fold symmetric tensor product $\mathcal{H}_N = L^2_{\mathrm{sym}}(\Lambda^N, \de x)$. That is, $\mathcal{H}_N$ is the space of square integrable functions of $N$ variables that are invariant under exchange of any pair of variables. On $\mathcal{H}_N$ we define the Hamiltonian of the system by\footnote{In our units the mass is given by $m = 1/2$ and $\hbar = 1 = k_{\mathrm{B}}$, where $k_{\mathrm{B}}$ is Boltzmann's constant.}
\begin{equation}
	H_N = \sum_{i=1}^N -\Delta_i + \sum_{1 \leq i < j \leq N} v_N(d(x_i,x_j)).
	\label{eq:HN}
\end{equation}
Here $\Delta$ denotes the Laplacian on the torus and $d(x,y)$ is the distance between two points $x,y \in \Lambda$. The interaction potential is of the form
\begin{equation}
	v_N(d(x,y)) = \left( N/L \right)^2 v(N d(x,y) /L)
	\label{eq:vN}
\end{equation}
with a nonnegative, measurable function $v: [0,\infty) \mapsto [0,\infty]$, independent of $N$. A simple scaling argument shows that if $a_v$ is the scattering length of $v$, then the scattering length $a_N$ of $v_N$ is given by
\begin{equation}
	a_N = a_v L/ N.
	\label{eq:aN}
\end{equation}
The scattering length is a combined measure for the range and the strength of a potential and its definition is recalled in \cite[Appendix~C]{Themathematicsofthebosegas}. We are interested in the choice $a_{v} \sim 1$, i.e. $a_N/L \sim N^{-1}$. By definition, $v$ is allowed to take the value $+\infty$ on a set of positive measure which corresponds to hard core interactions. We will assume that $v$ vanishes outside the ball with radius $R_0$, that is, it is of finite range. 

In the concrete realization of $\Lambda$ as the set $[0,L]^3 \subset \mathbb{R}^3$, $\Delta$ is the usual Laplacian with periodic boundary conditions and the distance function $d(x,y)$ is given by $d(x,y) = \min_{k \in \mathbb{Z}^3} | x - y - k L |$. We also note that $v_N(d(x,y)) = \sum_{k \in \mathbb{Z}^3} v_N(| x - y - k L |)$ if $R_0 < N/2$.

The canonical free energy related to the Hamiltonian $H_N$ at inverse temperature $\beta$ is defined by
\begin{equation}
	F(\beta,N,L) = -\tfrac{1}{\beta} \ln\left( \tr_{\mathcal{H}_N} e^{-\beta H_N} \right).
	\label{eq:F}
\end{equation}
The trace in Eq.~\eqref{eq:F} is taken over $\mathcal{H}_N$. In the following we will drop the subscript $\mathcal{H}_N$ and write $\tr$ for this trace. By $F_0(\beta,N,L)$ we will denote the free energy of the ideal Bose gas, that is, the free energy for the Hamiltonian $H_N$ with $v = 0$. A useful characterization of the free energy \eqref{eq:F} is via the Gibbs variational principle. Let us denote by $\mathcal{S}_N$ the set of $N$-particle states on $\mathcal{H}_N$ with finite energy, that is, the set of trace-class operators $\Gamma_N$ on $\mathcal{H}_N$ with $0 \leq \Gamma_N \leq 1$, $\tr \Gamma_N = 1$ and $\tr [H_N \Gamma_N] < + \infty$\footnote{Here and in the following, we interpret $\tr[H \Gamma]$ for positive operators $H$ and $\Gamma$ as $\tr[H^{1/2} \Gamma H^{1/2}]$. This expression is always well-defined if one allows the value $+\infty$. In particular, finiteness of $\tr[H \Gamma]$ does not require the operator $H \Gamma$ to be trace-class, only that $H^{1/2} \Gamma H^{1/2}$ is.}. We also denote by $S(\Gamma_N) = - \tr[\Gamma_N \ln(\Gamma_N)]$ the von-Neumann entropy of a state $\Gamma_N \in \mathcal{S}_N$. The Gibbs free energy functional is given by
\begin{equation}
	\mathcal{F}(\Gamma_N) = \tr\left[ H_N \Gamma_N \right] - \tfrac{1}{\beta} S\left( \Gamma_N \right).
	\label{eq:Gibbsfreeenergyfunctional}
\end{equation}
Minimization of $\mathcal{F}$ over all states yields the free energy \eqref{eq:F}, i.e.,
\begin{equation}
	F(\beta,N,L) = \inf_{\Gamma_N \in \mathcal{S}_N} \mathcal{F}(\Gamma_N).
	\label{eq:Gibbsvariationalprinciple}
\end{equation}
 The unique minimizer of the Gibbs free energy functional is given by the canonical Gibbs state
\begin{equation}
	\Gamma_{N}^{\mathrm{G}} = \frac{e^{-\beta H_N}}{\tr\left[ e^{-\beta H_N} \right]}.
	\label{eq:Gibbsstate}
\end{equation}
Apart from the particle number $N$ and the scattering length $a_{v}$ of the unscaled potential $v$, our system depends on the density $\varrho=N/|\Lambda|$ and on the inverse temperature $\beta$. We are interested in the free energy of the system as $N$ tends to infinity and for temperatures that are comparable to or smaller than the critical temperature of the ideal Bose gas, or equivalently, such that $\beta \gtrsim \beta_{\mathrm{c}}$. Here $\beta_{\mathrm{c}} = (4 \pi)^{-1} (\zeta(3/2)/\varrho)^{2/3}$ denotes the inverse critical temperature of the ideal Bose gas in a box of side length $L$ and $\zeta$ denotes the Riemann zeta function. Since $a_N \sim L/N$ this limit can  be interpreted as a combined thermodynamic and GP limit, see also Remark~11 in Section~\ref{sec:maintheorem} below.

For a given state $\Gamma_N$ we denote by $\gamma_N$ its one-particle density matrix (1-pdm) which we define via its integral kernel 
\begin{equation}
	\gamma_N(x,y) = \tr \left[ a_y^* a_x \Gamma_N \right].
	\label{eq:1pdm}
\end{equation}
In the above equation $a_x^*$ and $a_x$ denote the usual creation and annihilation operators (actually operator-valued distributions) of a particle at point $x$, fulfilling the canonical commutation relations $\left[ a_x, a_y^* \right] = \delta(x-y)$. Here $\delta$ denotes the Dirac delta distribution. Equivalently, the integral kernel of $\gamma_N$ can be defined via the integral kernel $\Gamma_N(x_1,...,x_N;y_1,...,y_N)$ of the state $\Gamma_N$ by integrating out all but one coordinates and multiplying the result with $N$: 
\begin{equation}
	\gamma_N(x,y) = N \int_{\mathbb{R}^{3(N-1)}} \Gamma_N(x,q_1,...,q_{N-1};y,q_1,...,q_{N-1}) \de (q_1,...,q_{N-1}).
\end{equation}
By definition, a sequence of states $\Gamma_N$ shows BEC if the related 1-pdms $\gamma_N$ have at least one eigenvalue of order $N$, i.e.,
\begin{equation}
	\liminf_{N \to \infty} \frac{ \sup_{\Vert \phi \Vert = 1} \left\langle \phi, \gamma_N \phi \right\rangle }{N} > 0.
	\label{eq:DefBEC}
\end{equation}
\subsection{The ideal Bose gas on the torus}
\label{sec:supplementary}
In this section we recall some basic facts and formulas concerning the ideal Bose gas on a flat torus $\Lambda$.  Since no explicit formulas are available for the canonical ensemble we state all results for the grand canonical ensemble. This is justified because the discussion in the Appendix shows that the free energy and the expected number of particles in the condensate, when computed in the two ensembles, agree with a precision that is sufficient for our purposes. 

The expected number of particles in the condensate and in the thermal cloud (all particles outside the condensate) are given by
\begin{equation}
N_0^{\mathrm{gc}} = \frac{1}{e^{-\beta \mu_0} - 1} \quad \text{ and } \quad N_{\mathrm{th}}^{\mathrm{gc}} = \sum_{p \in \frac{2 \pi}{L} \mathbb{Z}^3 \backslash \{ 0 \} } \frac{1}{e^{\beta \left( p^2 - \mu_0 \right)} - 1},
\label{eq:gcparticlenumbers}
\end{equation}
respectively.  The chemical potential $\mu_0$ and the expected number of particles $\overline{N}$ in the system are related via the equation $\overline{N} = N_0^{\mathrm{gc}} + N_{\mathrm{th}}^{\mathrm{gc}}$. The relevant densities are denoted by $\overline{\varrho} = |\Lambda|^{-1} \overline{N}$, $\varrho_0^{\mathrm{gc}} = | \Lambda |^{-1} N_0^{\mathrm{gc}}$ and $\varrho_{\mathrm{th}}^{\mathrm{gc}} = | \Lambda |^{-1} N_{\mathrm{th}}^{\mathrm{gc}}$. The ideal Bose gas shows a BEC phase transition in the limit of large particle number. More precisely, for $\overline{N} \to \infty$ one has
\begin{equation}
\frac{N_0^{\mathrm{gc}}}{N} \simeq \left[ 1- \left( \frac{\beta_{\mathrm{c}}}{\beta} \right)^{3/2} \right]_+ \quad \text{ with } \quad \beta_{\mathrm{c}} = \frac{1}{4 \pi} \left( \frac{\overline{\varrho}}{\zeta(3/2)} \right)^{-2/3}
\label{eq:criticaltemp}
\end{equation}
and $\zeta$ the Riemann zeta function. Here $[x]_+ = \max\{ x,0 \}$ denotes the positive part. For inverse temperatures such that $\beta > \beta_{\mathrm{c}} (1 + \epsilon)$ with $\epsilon > 0$, the chemical potential is to leading order given by $\mu_0 \simeq - (\beta N_0^{\mathrm{gc}})^{-1} \sim - (\beta N)^{-1}$, while for $\beta < \beta_{\mathrm{c}} (1 - \epsilon)$ it scales as $\mu_0 \sim - \beta^{-1}$. Finally, the free energy of the system is given by
\begin{equation}
F^{\mathrm{gc}}_0(\beta,\mu_0,L) = \mu_0 \overline{ N } + \tfrac{1}{\beta} \sum_{p \in \frac{2 \pi}{L} \mathbb{Z}^3} \ln\left( 1 - e^{-\beta\left( p^2 - \mu_0 \right)} \right).
\label{eq:gcfreeenergy}
\end{equation}
\subsection{The main theorem}
\label{sec:maintheorem}
Our main result is the following statement:
\begin{theorem}
\label{thm:periodic}
	Assume that $v : [0,\infty) \to [0,\infty]$ is a nonnegative and measurable function with compact support. Denote by $\varrho_0(\beta,N,L)$ the expected condensate density in the canonical Gibbs state of the ideal Bose. In the combined limit $N \to \infty$, $\beta \varrho^{2/3} \sim 1$ and $a_N$ given by \eqref{eq:aN} with $a_{v} > 0$ fixed, we have
	\begin{equation}
		F(\beta,N,L) = F_0(\beta,N,L) + 4 \pi a_N | \Lambda | \left( 2 \varrho^2 - \varrho_0^2\left( \beta, N, L \right) \right) + O\left(L^{-2} N^{1-\alpha}\right) 
		\label{eq:main1}
	\end{equation}
	for some $\alpha>0$. Moreover, for any sequence of states $\Gamma_N \in \mathcal{S}_N$ with 1-pdms $\gamma_N$ and 
	\begin{equation}
		\mathcal{F}\left( \Gamma_N \right) = F_0(\beta,N,L) + 4 \pi a_N | \Lambda | \left( 2 \varrho^2 - \varrho_0^2\left( \beta, N, L \right) \right) \left( 1 + \delta(N) \right),
		\label{eq:main2}
	\end{equation}
	we have for some $\sigma > 0$
	\begin{equation}
		\left\| \gamma_N - \gamma_{N,0} \right\|_1 \leq O\left( N \left( N^{-\sigma} + \delta(N)^{1/8} \right) \right).
		\label{eq:main3}
	\end{equation}
	Here $\gamma_{N,0}$ denotes the 1-pdm of the canonical Gibbs state of the ideal gas in $\Lambda$ and $\Vert \cdot \Vert_1$ is the trace norm.
\end{theorem}
The fact that the difference between the specific free energy of the interacting system and the one of the ideal gas is given by $4 \pi a ( 2 \varrho^2 - \varrho_0^2\left( \beta, N, L \right) )$ also holds for the dilute Bose gas in the thermodynamic limit, see \cite{Yin2010,Sei2008}. The formulas look the same because, as already mentioned above, our limit can be interpreted as a combined thermodynamic and dilute limit. We emphasize that \eqref{eq:main3} holds for any approximate minimizer of the Gibbs free energy functional in the sense of \eqref{eq:main2}, and not only for the interacting Gibbs state \eqref{eq:Gibbsstate} of the system. This in particular proves BEC for this class of states (see also Remark~9 below).

\subsubsection*{Remarks:}
\begin{enumerate}
	\item The constants in the error terms in \eqref{eq:main1} and \eqref{eq:main3} are uniform in the inverse temperature as long as $\beta \gtrsim \beta_{\mathrm{c}}$. The exponents $\alpha$ and $\sigma$ can be chosen as $\alpha = 4/6885 - \eta$ and $\sigma = 1/6885 - \eta$  for any  $\eta > 0$. (For $\eta \to 0$ the constants in the error terms in \eqref{eq:main1} and \eqref{eq:main3} blow up.) Our rate in Eq.~\eqref{eq:main1} is the same as the one for the lower bound for the dilute Bose gas in the thermodynamic limit \cite{Sei2008}. The known result for the ground state is implied by our result in the limit $\beta \to \infty$. The error term is worse, however. 
	\item The result in \eqref{eq:main3} does not assume translation invariance of the states $\Gamma_N$. If one assumes that $\Gamma_N$ is translation invariant the rate of convergence can be improved. In particular, one finds that the error term is bounded in terms of $\delta(N)^{1/2}$ instead of on $\delta(N)^{1/8}$ in this case.
	\item Our result is uniform in the unscaled scattering length $a_v$ as long as $a_{v} \in (0,d]$ with $0<d<\infty$. 
	\item For $\beta \sim \beta_{\mathrm{c}}$ and $a_N \sim L/N$, we have $F_0 \sim | \Lambda | \varrho^{5/3} = L^{-2} N^{5/3}$ for the free energy of the ideal gas, whereas the interaction energy is given by $| \Lambda | a_N \varrho^2 \sim L^{-2} N$. Up to this scale we control the free energy of the interacting gas.
	\item The interaction energy is for $\beta \leq \beta_{\mathrm{c}}$ given by $8 \pi | \Lambda | a_N \varrho^2$ to leading order, which has to be compared to $4 \pi | \Lambda | a_N \varrho^2$, its value at zero temperature. The additional factor of two is an exchange effect due to the symmetrization of the wave function, which only plays a role if the particles occupy two different one-particle orbitals. Above the critical temperature this is essentially always the case but particles inside a condensate do not experience this effect. This leads to the dependence of the interaction energy in Eq.~\eqref{eq:main1} on $\varrho_0(\beta,N,L)$.
	\item The free energy $F_0(\beta,N,L)$, the condensate density $\varrho_0(\beta,N,L)$ and the 1-pdm $\gamma_{N,0}$ are the ones of the ideal gas in the canonical ensemble for which no explicit formulas are available. Our results are still valid if these three quantities are replaced by their corresponding grand canonical versions, as can be seen from the discussion in the Appendix.
	\item Our bounds depend, apart from the scattering length $a_{v}$, only on the range $R_0$ of the interaction potential. This dependence could be displayed explicitly. By cutting $v$ in a suitable way one can extend the result to infinite range potentials which are integrable outside some ball with finite radius, that is, to all nonnegative potentials with a finite scattering length.
	\item Our proof allows for the incorporation of internal degrees of freedom such as spin. For simplicity we only treat the case of spinless particles here.
	\item Eq.~\eqref{eq:main3} implies BEC into the constant function $| \Lambda |^{-1/2}$ on the torus with condensate fraction given by the one of the ideal Bose gas to leading order. The statement follows from the fact that the trace norm $\Vert \cdot \Vert_1$ bounds the operator norm $\Vert \cdot \Vert$, and hence \eqref{eq:main3} implies
		\begin{equation}
		\left\| \gamma_N - \gamma_{N,0} \right\| \leq O\left( N \left( N^{-\sigma} + \delta(N)^{1/8} \right) \right).
		\label{eq:main4}
		\end{equation}
	The critical temperature does not depend on the interaction in the dilute GP limit considered here. Deviations from this limiting value can be observed in experiments \cite{Tamm_etal 2011}, however. 
	\item In the initial experiments BECs could  only be prepared in harmonic traps. More recent set-ups also allow for the preparation of such systems in a box type potential with approximate hard wall boundary conditions, see \cite{Gauntetal2013}. 
	The inclusion of these boundary conditions into our setting will be discussed in the next section. 
	\item Let us compare our setting to the one of the dilute Bose gas in the thermodynamic limit, which was considered in \cite{Sei2008,Yin2010}. In the latter case one first takes the thermodynamic limit $N,L \to \infty$ with $\varrho = N/L^3$ and $\beta$ fixed, and afterwards considers the dilute limit where $a \varrho^{1/3} \ll 1$ and $\beta \varrho^{2/3} \gtrsim 1$. In our case 
	we take the limit $N \to \infty$ with $a_N \sim L/N$ and $\beta \varrho^{2/3} \gtrsim 1$. That is, we take a combined thermodynamic and dilute limit, where $a_N \varrho^{1/3} \sim N^{-2/3}$. 
	In this limit the spectral gap of the Laplacian is of the same order of magnitude as the interaction energy per particle, $\sim L^{-2}$. This allows for a proof of BEC based on the coercivity of the relevant (free) energy functional, see \eqref{eq:main3}, \cite{LiSei2002} and \cite{me}. In contrast, proving BEC for the dilute Bose gas in the thermodynamic limit has been a major open problem in mathematical physics for almost a century. The spectral gap of the Laplacian closes in the thermodynamic limit and the system is expected to have Goldstone modes (sound waves in the case of the dilute Bose gas), that is, excitations with arbitrarily small energy. Accordingly, the relevant coercivity of the (free) energy is lost. We refer to \cite{BECprogram} for an overview of an ambitious long-term project aimed at proving BEC in the thermodynamic limit with renormalization group techniques. Although it has so far not been  possible to prove BEC for the dilute Bose gas in the thermodynamic limit, it is possible with our methods to obtain information on the 1-pdm when it is projected to suitably chosen high momentum modes. This should be compared to the case of the dilute Fermi gas in the thermodynamic limit, where comparable bounds have been proven in \cite{RobertFermigas}.
\end{enumerate}
\subsection{Extension to the case of Dirichlet boundary conditions}
\label{sec:dirichletboundary}
The methods  developed for the proof of Theorem~\ref{thm:periodic} also allow for the proof of a similar statement when the periodic boundary conditions are replaced by Dirichlet boundary conditions. In this case the system does not condense into the constant function but into the minimizer of the GP energy functional. To state this result, we first need to introduce some notation.

For functions $\phi \in H_0^1([0,L]^3)$, we introduce the GP energy functional
\begin{equation}
	\mathcal{E}^{\mathrm{GP}}\left( \phi \right) = \int_{[0,L]^3} \left( \left| \nabla \phi(x) \right|^2 + 4 \pi a \left| \phi(x) \right|^4 \right) \text{d}x.
	\label{eq:GPfunctional}
\end{equation}  
Here $H_0^1([0,L]^3)$ denotes the usual Sobolev space of functions with zero boundary conditions. We denote by
\begin{equation}
	E^{\mathrm{GP}}(N,a,L) = \inf_{\phi \in H_0^1([0,L]^3) \text{ with } \Vert \phi \Vert = \sqrt{N}} \mathcal{E}^{\mathrm{GP}}\left( \phi \right)
	\label{eq:GPenergy}
\end{equation}
its ground state energy and by $\phi^{\mathrm{GP}}_{N,a}$ its minimizer, which is unique up to a phase. One readily checks the  scaling relations $E^{\mathrm{GP}}(N,a,L) = L^{-2} N E^{\mathrm{GP}}(1,Na/L,1)$ and $\phi^{\mathrm{GP}}_{N,a} = N^{1/2} \phi^{\mathrm{GP}}_{1,Na}$. Additionally, let $H_N^{\mathrm{D}}$ be the Hamiltonian \eqref{eq:HN}, where the Laplacian on $\Lambda$ is replaced by $\Delta^{\mathrm{D}}$, the Dirichlet Laplacian on $[0,L]^3$. Similarly let $F^{\mathrm{D}}(\beta,N,L)$ be the canonical free energy for  $H_N^{\mathrm{D}}$ and let $F_0^{\mathrm{D}}(\beta,N,L)$ be the same quantity in the case $v=0$. By $\varrho^{\mathrm{D}}(x)$ we denote the density of the canonical Gibbs state of the ideal gas, and $\varphi_0$ is the ground state of $-\Delta^{\mathrm{D}}$. Finally, we introduce with $N^{\mathrm{D}}_0$ the expected number of particles in the condensate of the Gibbs state of the ideal gas and denote by $\varrho_{\mathrm{th}}^{\mathrm{D}}(x) = \varrho^{\mathrm{D}}(x) - N_0^{\mathrm{D}} | \varphi_0(x) |^2$ the density of its thermal cloud. For simplicity we suppress the dependence of the densities on $\beta$, $N$ and $L$.

The analogue of Theorem~\ref{thm:periodic} in the case of Dirichlet boundary conditions is the following statement:
\begin{theorem}
	\label{thm:Dirichlet}
	Assume that $v : [0,\infty) \to [0,\infty]$ is a nonnegative and measurable function with finite scattering length $a_v$. In the combined limit $N \to \infty$, $\beta \varrho^{2/3} \sim 1$ and $a_N$ given by \eqref{eq:aN} with $a_{v} > 0$ fixed, we have
	\begin{equation}
	F^{\mathrm{D}}(\beta,N,L) - F^{\mathrm{D}}_0(\beta,N,L)\simeq  \left[ E^{\mathrm{GP}}\left(N_0^{\mathrm{D}},a_N,L \right) + 8 \pi a_N \int_{[0,L]^3} \left( \varrho_{\mathrm{th}}^{\mathrm{D}}(x)^2 + 2 \varrho_{\mathrm{th}}^{\mathrm{D}}(x) \left| \phi^{\mathrm{GP}}_{N_0^{\mathrm{D}},a_N}(x) \right|^2 \right) \de x \right].
	\label{eq:main1b}
	\end{equation}
	Moreover, for any sequence of states $\Gamma_N \in \mathcal{S}_N$ with 1-pdms $\gamma_N$ and 
	\begin{equation}
	\mathcal{F}\left( \Gamma_N \right) - F^{\mathrm{D}}_0(\beta,N,L) \simeq \left[ E^{\mathrm{GP}}\left( N_0^{\mathrm{D}},a_N,L \right) + 8 \pi a_N \int_{[0,L]^3} \left( \varrho_{\mathrm{th}}^{\mathrm{D}}(x)^2 + 2 \varrho_{\mathrm{th}}^{\mathrm{D}}(x) \left| \phi^{\mathrm{GP}}_{N_0^{\mathrm{D}},a_N}(x) \right|^2 \right) \de x \right],
	\label{eq:main2b}
	\end{equation}
	we have
	\begin{equation}
	\lim \frac{1}{N} \left\| \gamma_N - \left( \gamma_{N,0} - N_0^{\mathrm{D}} | \varphi_0 \rangle\langle \varphi_0 | + \left| \phi^{\mathrm{GP}}_{N_0^{\mathrm{D}},a_N} \right\rangle \left\langle \phi^{\mathrm{GP}}_{N_0^{\mathrm{D}},a_N} \right| \right) \right\|_1 = 0,
	\label{eq:main3b}
	\end{equation}
	where $\gamma_{N,0}$ denotes the 1-pdm of the canonical Gibbs state of the ideal gas and $\lim$ is our combined limit. Finally,
	\begin{equation}
	\lim \frac{1}{N} \left\| \gamma_N - \left| \phi^{\mathrm{GP}}_{N_0^{\mathrm{D}},a_N} \right\rangle \left\langle \phi^{\mathrm{GP}}_{N_0^{\mathrm{D}},a_N} \right| \ \right\| = 0,
	\label{eq:main3c}
	\end{equation}
	where $\Vert \cdot \Vert$ denotes the operator norm. In particular, we have BEC with the same condensate fraction and the same critical temperature as in the case of the ideal Bose gas to leading order.
\end{theorem}

\subsubsection*{Remarks:}
\begin{enumerate}
	\item In the case of periodic boundary conditions, the condensate wave function is given by a constant and therefore minimizes the GP energy functional on the torus, that is, \eqref{eq:GPfunctional} for functions $\phi \in H_{\mathrm{per}}^1(\Lambda)$ (the usual Sobolev space of functions with periodic boundary conditions). For Dirichlet boundary conditions this picture changes because they force the minimizer of the GP functional to vary on the length scale of the box, that is, on $L$. This results in a macroscopic change of the energy of the condensate compared to the case with periodic boundary conditions. Although the Dirichlet boundary conditions do also change the free energy of the ideal gas compared to the case of periodic boundary conditions (not to leading order but on the scale we are interested in), they do not affect the density of the thermal cloud to leading order. This is because the energy per particle inside the thermal cloud is for $\beta \sim \beta_{\mathrm{c}}$ given by $\varrho^{2/3}$, where $\varrho=N/L^3$. Its density therefore varies on a length scale of order $\varrho^{-1/3}$ which is much smaller than the length scale of the box: $\varrho^{-1/3}/ L \sim N^{-1/3}$. Hence, the density of the thermal cloud is essentially a constant until close to the boundary. Since the expected number of particles in the condensate does not depend on the boundary conditions to leading order this, in particular, implies that the second term in the bracket on the right-hand side of \eqref{eq:main1b} satisfies
	\begin{equation}
		8 \pi a_N \int_{[0,L]^3} \left( \varrho_{\mathrm{th}}^{\mathrm{D}}(x)^2 + 2 \varrho_{\mathrm{th}}^{\mathrm{D}}(x) \left| \phi^{\mathrm{GP}}_{N_0^{\mathrm{D}},a_N}(x) \right|^2 \right) \de x \simeq 8 \pi a_N L^3 \left( \varrho_{\mathrm{th}}^2 + 2 \varrho_{\mathrm{th}} \varrho_0 \right).
	\end{equation}
	The term on the right-hand side depends on the expected condensate density of the Gibbs state of the ideal gas in the case of periodic boundary conditions $\varrho_0$ and on $\varrho_{\mathrm{th}} = \varrho - \varrho_0$. This should be compared to \eqref{eq:main1}, where the same terms appear.
	\item In the remaining part of the paper we will prove Theorem~\ref{thm:periodic} but we will not prove Theorem~\ref{thm:Dirichlet}. The methods developed to prove Theorem~\ref{thm:periodic} can, however, be adjusted to also obtain a proof for Theorem~\ref{thm:Dirichlet}. Let us mention the main points to consider. Concerning the lower bound, the main point is that the technique from \cite{Sei2008}, that we use for the proof of the lower bound, naturally translates to the case of Dirichlet boundary conditions. This is because the c-number substitution is done with a sufficient number of momentum modes such that the GP minimizer, which varies on the length scale of the box, can be efficiently approximated with them. Additionally, as explained in the previous remark, the density of the thermal cloud of the ideal gas is constant to leading order. This allows one to use essentially the same technique to compute the free energy related to the modes that are not affected by the c-number substitution as in the case of periodic boundary conditions. To extend the proof of the upper bound, one has to cut the Fock space into high and low momentum modes, as it has been done in the proof of the lower bound. In the Fock space related to the low momentum modes one chooses the trial state to be a product wave function with $N_0(\beta,N,L)$ particles sitting in an approximate version of the GP minimizer. As above, $N_0(\beta,N,L)$ denotes the expected number of particles in the condensate of the ideal Bose gas. The overall trial state is then given by the symmetric tensor product of this function and a non-interacting canonical Gibbs state acting on the Fock space related to the high momentum modes (at the correct temperature and with $N-N_0(\beta,N,L)$ particles). In order to obtain the leading order behavior of the interaction energy, which depends on the scattering length, one has to, as in the case of periodic boundary conditions, add a correlation structure. The proof of the asymptotics of the 1-pdm remains up to minor adjustments unchanged. Here the main point is that the Griffith argument has to be done with an approximate version of the GP minimizer, which depends only on the low momentum modes of the c-number substitution, instead of with the constant function. Since the concrete implementation of the above strategy would considerably increase the length of the proof compared to the case of periodic boundary conditions, without adding substantial new difficulties, we only give the proof of Theorem~\ref{thm:periodic} here.
	\item We expect that the error bounds one obtains by following the strategy indicated by Remark~2 to prove Theorem~\ref{thm:Dirichlet} are not worse than those appearing in Theorem~\eqref{thm:periodic}. In particular, we expect the same uniformity of the remainder in the inverse temperature as long as $\beta \gtrsim \beta_{\mathrm{c}}$. 
	\item Apart from periodic and Dirichlet boundary conditions we could also treat Neumann boundary conditions. Since the condensate function is a constant in this case one obtains the same statement as for periodic boundary conditions.
\end{enumerate}
\subsection{The proof strategy}
\label{sec:proofstrategy}
Before we come to the proof of Theorem~\ref{thm:periodic}, we first briefly present the main steps to guide the reader.

Sec.~\ref{sec:proofupperbound} contains a proof of the upper bound for the  free energy of the interacting gas. It is based on the Gibbs variational principle and the construction of a trial state whose free energy can be bounded from above by the desired expression. As a trial state we use the canonical Gibbs state of the ideal Bose gas. In order to obtain the scattering length in the interaction energy, we have to add a correlation structure which decreases the probability of finding two particles close together. Our ansatz yields a much simpler and shorter proof of the upper bound than the related proof in case of the dilute Bose gas in the thermodynamic limit \cite{Yin2010}. This is only possible, however, because the scattering length scales as $a_N \sim L/N$,  and hence the system is extremely dilute.

For the proof of the lower bound for $F(\beta,N,L)$ in Sec.~\ref{sec:prooflowerbound}, we adjust the techniques developed for the related proof for the dilute Bose gas in the thermodynamic limit \cite{Sei2008}. One key ingredient of our approach is a novel use of the Gibbs variational principle that goes hand in hand with the c-number substitution, which is a central ingredient of the proof in \cite{Sei2008}. In comparison to \cite{Sei2008}, this allows us to work with a general state $\Gamma$ instead of with a version of the grand canonical Gibbs state. In particular, we can keep the information that $\Gamma$ has exactly $N$ particles. This adjustment is essential for the proof of the asymptotics of the 1-pdm of approximate minimizers of the Gibbs free energy functional that we give in Sec.~\ref{sec:asymptotics1pdm}. Also in view of Sec.~\ref{sec:asymptotics1pdm}, we have to prove the lower bound for a slightly generalized Hamiltonian in which the energy of the lowest eigenfunction of the Laplacian is shifted by $\lambda \geq 0$. We remark that, if one is only interested in the lower bound for the free energy, the technique from \cite{Sei2008} can be applied essentially without modifications. More precisely, one would only need to consider all those terms in \cite{Sei2008} that do not grow proportionally to the volume in the thermodynamic limit, and check that they are also of subleading order in the GP limit considered here (which is true). 

The proof of the asymptotics of the 1-pdm $\gamma_N$ of an approximate minimizer of the Gibbs free energy functional is based on the novel use of the Gibbs variational principle mentioned above and has two main ingredients. The first ingredient is an estimate showing that $\gamma_N$ is, when projected to high momentum modes, given by the 1-pdm of the ideal gas to leading order. This part of the proof is motivated by a similar proof in \cite{me} and is based on certain lower bounds for the bosonic relative entropy (the difference between two free energies) quantifying its coercivity. One main novelty in this part of our proof is a new lower bound for the bosonic relative entropy that allows us to simplify this part substantially w.r.t. the related part in \cite{me}. In particular, it allows one to obtain better rates for the trace norm convergence of the relevant 1-pdm for given bounds on the relative entropy. In order to show the same statement for $\gamma_N$ projected to the low momentum modes, which is the second main ingredient of our proof, we apply a Griffith argument. Such arguments are based on the fact that differentiation of the free energy w.r.t. a parameter in the Hamiltonian yields the quantity one is interested in. In our case the parameter is the shift of the lowest eigenvalue of the Laplacian and the quantity of interest is the expected number of particles in the constant function, that is, in the condensate. 
\section{Proof of the upper bound}
\label{sec:proofupperbound}
\subsection{The variational ansatz}
As trial state for the upper bound we choose the canonical Gibbs state of the ideal Bose gas on the torus and add a correlation structure. This is motivated by the following three observations: Firstly, the condensate wave function of the ideal gas on $\Lambda$ is given by $| \Lambda |^{-1/2}$. If we turn on a repulsive interaction this is not going to change. Secondly, the free energy of the ideal gas is for $\beta \sim \beta_{\mathrm{c}}$ much larger than the interaction energy given the second term on the right-hand side of \eqref{eq:main1}. This tells us that an approach based on first order perturbation theory should lead to the correct interaction energy. Finally, since $v$ may contain a hard core repulsion and because the scaled pair interaction $v_N$ becomes very singular for large $N$, we need to assure that the probability of finding two particle close together is reduced compared to the ideal gas. This is achieved with the correlation structure in the spirit of \cite{Jastrow}. In particular, it allows us to obtain the correct leading order of the interaction energy which is proportional to the scattering length. The idea to use a correlation structure in order to obtain the dependence of the energy of a dilute Bose gas on the scattering length has for the first time been used in \cite{Dyson} in the homogeneous case and in \cite{RobertGPderivation} in the inhomogeneous case.

Let $H_N^0$ denote the Hamiltonian $H_N$ \eqref{eq:HN} for $v = 0$. The canonical Gibbs state of the ideal Bose gas on the torus $\Lambda$ is given by
\begin{equation}
	\Gamma_{N,0}^{\mathrm{G}} = \frac{e^{-\beta H_{N,0}}}{\tr e^{-\beta H_{N,0}}}.
	\label{eq:idealGibbs}
\end{equation}
As correlation structure we choose the Jastrow-like function \cite{Jastrow,Dyson,ErdSchlYau2010}
\begin{equation}
	F(x_1,\ldots,x_{N}) = \prod_{1 \leq i < j \leq N} f_b( d(x_i, x_j ) ) \quad \text{ with } \quad f_b(r) = \begin{cases} f_0(r)/f_0(b) & \text{ for } r < b \\ 1 & \text{ for } r \geq b,  \end{cases}
	\label{eq:corrstructure}
\end{equation}
where $b > 0$ is a parameter to be determined and $f_0(|x|)$ is the unique solution of the zero-energy scattering equation
\begin{equation}
	- \Delta f(x) + \tfrac{1}{2} v_N(x) f(x) = 0 \quad \text{ with }  \lim_{|x| \to \infty} f(x) = 1,
	\label{eq:scatteringequation}
\end{equation}
see also \cite[Appendix~C]{Themathematicsofthebosegas}. We expand the canonical Gibbs state as $\Gamma_{N,0}^{\mathrm{G}} = \sum_{\alpha=1}^{\infty} \lambda_{\alpha} | \Psi_{\alpha} \rangle\langle \Psi_{\alpha} |$, where the functions $\Psi_{\alpha}$ are chosen as symmetrized products of real eigenfunctions of the Laplacian on $\Lambda$. The final trial state with correlation structure is defined by
\begin{equation}
	\widetilde{\Gamma}_{N,0}^{\mathrm{G}} = \sum_{\alpha=1}^{\infty} \lambda_{\alpha} | \Phi_{\alpha} \rangle\langle \Phi_{\alpha} | \quad \text{ with } \quad  \Phi_{\alpha}  = \frac{F \Psi_{\alpha}}{\| F \Psi_{\alpha} \| }
	\label{eq:trialstate}
\end{equation}
and its free energy equals
\begin{equation}
	\mathcal{F}\left(\widetilde{\Gamma}_{N,0}^{\mathrm{G}} \right) = \tr \left[ H_N \widetilde{\Gamma}_{N,0}^{\mathrm{G}} \right] - T S \left( \widetilde{\Gamma}_{N,0}^{\mathrm{G}} \right).
	\label{eq:freeenergytrialstate}
\end{equation}
The remainder of this section is devoted to finding an appropriate upper bound for $\mathcal{F}(\widetilde{\Gamma}_{N,0}^{\mathrm{G}})$. We start with the computation of the energy.
\subsection{The energy}
We use the definition of our trial state \eqref{eq:trialstate} and write
\begin{equation}
\tr \left[ H_N \widetilde{\Gamma}_{N,0}^{\mathrm{G}} \right] = \sum_{\alpha=1}^{\infty} \lambda_{\alpha} \frac{\langle F \Psi_{\alpha}, H_N F \Psi_{\alpha} \rangle}{\langle F \Psi_{\alpha}, F \Psi_{\alpha} \rangle}.
\label{eq:upperbound9}
\end{equation}
Bearing in mind that all eigenfunctions $\Psi_{\alpha}$ of $H_{N,0}$ are chosen to be real-valued, we integrate by parts once to rewrite the kinetic energy of the $i$-th particle as
\begin{equation}
-\int_{\Lambda^N} \overline{ F \Psi_{\alpha} } \nabla_i^2 F \Psi_{\alpha} \text{d}X = -\int_{\Lambda^N} \left[ F^2 \left( \Psi_{\alpha} \nabla_i^2 \Psi_{\alpha} \right) - \Psi_{\alpha}^2 \left( \nabla_i F \right)^2 \right] \text{d}X\,,
\label{eq:upperbound10}
\end{equation}
where $\text{d}X$ is short for $\text{d}(x_1,\ldots,x_N)$. For the energy of a single function $\Psi_{\alpha}$ this implies
\begin{align}
&\left\langle F \Psi_{\alpha}, \left[ \sum_{i=1}^{N} -\Delta_i + \sum_{1 \leq i < j \leq N} v_N(d(x_i,x_j)) \right] F \Psi_{\alpha} \right\rangle =
\label{eq:upperbound11} \\
&\hspace{3cm} \int_{\Lambda^N} \left\lbrace F^2 \Psi_{\alpha} \underbrace{\left[ \sum_{i=1}^{N} -\Delta_i \right] \Psi_{\alpha}}_{=E_{\alpha} \Psi_{\alpha}} +  \left[ \sum_{i=1}^{N} (\nabla_i F)^2 + \sum_{1 \leq i < j \leq N} v_N(d(x_i, x_j)) F^2 \right] \Psi_{\alpha}^2 \right\rbrace \de X, \nonumber
\end{align}
where $E_{\alpha}$ denotes its energy w.r.t. $H_N^0$, and the whole energy can be written as
\begin{align}
&\tr \left[ H_N \tilde{\Gamma}_{N,0}^{\mathrm{G}} \right] = \tr \left[ H_{N,0} \Gamma_{N,0}^{\mathrm{G}} \right] + \sum_{\alpha=1}^{\infty} \lambda_{\alpha} \frac{\int_{\Lambda^N} \Psi_{\alpha}^2 \left[ \sum_{i=1}^{N} (\nabla_i F)^2 + \sum_{1 \leq i < j \leq N} v_N(d(x_i,x_j)) F^2 \right] \de X}{\left\Vert F \Psi_{\alpha} \right\Vert^2}. 
\label{eq:upperbound12}
\end{align}
The following Lemma provides a lower bound for the norm of $F \Psi_{\alpha}$, and thereby an upper bound on  the second term on the right-hand side of Eq.~\eqref{eq:upperbound12}, as long as $(4 \pi /3) | \Lambda | \varrho^2 a_N b^2 < 1$.
\begin{lemma}
	\label{lem:denominators}
	The $L^2(\Lambda^N)$-norm of $F \psi_{\alpha}$ can be bounded from below as 
	\begin{equation}
	\left\Vert F \Psi_{\alpha} \right\Vert^2 \geq 1 - \frac{4 \pi}{3} | \Lambda | \varrho^2 a_N b^2.
	\label{eq:upperbound13}
	\end{equation}
\end{lemma}
\begin{proof}
	Spelled out in more detail, the norm of $F \Psi_{\alpha}$ reads
	\begin{equation}
	\left\Vert F \Psi_{\alpha} \right\Vert^2 = \int_{\Lambda^N} | \Psi_{\alpha} |^2 \prod_{1 \leq i < j \leq N} f_b(d(x_i, x_j))^2 \text{d}X.
	\label{eq:upperbound14}
	\end{equation} 
	We define $\eta_b(r) = 1 - f_b(r)^2$ and estimate
	\begin{align}
	\left\Vert F \Psi_{\alpha} \right\Vert^2 &\geq \int_{\Lambda^N} | \Psi_{\alpha} |^2 \left( 1 -  \sum_{1 \leq i < j \leq N} \eta_b( d(x_i, x_j )) \right) \text{d}X \label{eq:upperbound15} \\
	&= 1 - \int_{\Lambda \times \Lambda} \eta_b(d(x,y)) \varrho^{(2)}_{\Psi_{\alpha}}(x,y) \text{d}(x,y), \nonumber
	\end{align}
	where $\varrho^{(2)}_{\Psi_{\alpha}}(x,y)$ denotes the two-particle density of $\Psi_{\alpha}$. Next, we use the fact that the $\Psi_{\alpha}$ are symmetrized products of one-particle orbitals to conclude that 
	\begin{equation}
	\varrho^{(2)}_{\Psi_{\alpha}}(x,y) \leq  \varrho_{\Psi_{\alpha}}(x) \varrho_{\Psi_{\alpha}}(y)
	\end{equation}
	holds. Here $\varrho_{\Psi_{\alpha}}$ is the one-particle density of $\Psi_{\alpha}$. Since the density of each $\Psi_{\alpha}$ is a constant, we have $\varrho_{\Psi_{\alpha}} = \varrho$. This allows us to bound the integral on the right-hand side of Eq.~\eqref{eq:upperbound15} in the following way:
	\begin{equation}
	\int_{\Lambda \times \Lambda} \eta_b(d(x,y)) \varrho^{(2)}_{\Psi_{\alpha}}(x,y) \text{d}(x,y) \leq | \Lambda | \varrho^2 \int_{\mathbb{R}^3} \eta_b(|x|) \text{d}x \leq \frac{4 \pi}{3} | \Lambda | \varrho^2 a_N b^2.
	\label{eq:upperbound16}
	\end{equation}
	To obtain the bound for the integral of $\eta_b$, we used its explicit form and the lower bound $f_0(|x|) \geq \left[ 1-a_N/|x| \right]_+$, see \cite[Appendix~C]{Themathematicsofthebosegas}. In combination with \eqref{eq:upperbound15}, this proves the claim.
\end{proof}
Next we analyze the numerator of the second term on the right-hand side of Eq.~\eqref{eq:upperbound12}. We compute
\begin{equation}
\nabla_i F(x_1,\ldots,x_{N}) = \sum_{\substack{l = 1 \\ l \neq i}}^{N} \frac{F(x_1,\ldots,x_{N})}{f_b(d(x_l,x_i))} \nabla f_b(d(x_l,x_i)).
\label{eq:upperbound18}
\end{equation} 
The square of this expression is given by
\begin{align}
\left( \nabla_i F \right)^2 &= \sum_{\substack{l=1 \\ l \neq i}}^N \frac{F^2}{f_b(d(x_l, x_i))^2} \left[ \nabla f_b(d(x_l, x_i)) \right]^2 \label{eq:upperbound19} \\
&\hspace{2cm} + \sum_{\substack{k,l=1 \\ l,k \neq i \\ k \neq l}}^N \frac{F^2}{f_b(d(x_l,x_i)) f_b(d(x_k,x_i))} \nabla f_b(d(x_l,x_i)) \nabla f_b(d(x_k,x_i)). \nonumber
\end{align}
These terms need to be inserted into the numerator of the second term on the right-hand side of Eq.~\eqref{eq:upperbound12} and we start with the first term on the right-hand side of the above equation. Introducing the function $\xi(|x|) = \left[ \nabla f_b(|x|) \right]^2 + \tfrac{1}{2} v_N(|x|) f_b(|x|)^2$ and noting that $0 \leq f_b \leq 1$ as well as $\sum_{\alpha=1}^{\infty} \lambda_{\alpha} \varrho^{(2)}_{\Psi_{\alpha}}(x,y) = \varrho^{(2)}_{\Gamma_{N,0}^{\mathrm{G}}}(x,y)$, we obtain
\begin{align}
&\sum_{\alpha=1}^{\infty} \lambda_{\alpha} \sum_{1 \leq i < j \leq N} \int_{\Lambda^{N}} \left\lbrace \frac{2F^2}{f_b(d(x_i,x_j))^2} \left[ \nabla f_b(d(x_i,x_j)) \right]^2 + v_N(d(x_i,x_j)) F^2 \right\rbrace \Psi_{\alpha}^2 \text{d}X \label{eq:upperbound20} \\
&\hspace{8cm} \leq 2 \int_{\Lambda^2} \xi(d(x,y)) \varrho^{(2)}_{\Gamma_{N,0}^{\mathrm{G}}}(x,y) \text{d}(x,y). \nonumber
\end{align}
The off-diagonal terms in Eq.~\eqref{eq:upperbound19} can be bounded from above in a similar way by
\begin{equation}
6 \int_{\Lambda^3} \varrho^{(3)}_{\Gamma_{N,0}^{\mathrm{G}}}(x,y,z) \vert \nabla f_b(d(x,y)) \nabla f_b(d(z,y)) \vert \text{d}(x,y,z). \label{eq:upperbound21}
\end{equation}
Combining these two bounds with \eqref{eq:upperbound12} and \eqref{eq:upperbound13}, we obtain
\begin{equation}
\tr\left[ H_N \widetilde{\Gamma}_{N,0}^{\mathrm{G}} \right] \leq \tr\left[ H_{N,0} \Gamma_{N,0}^{\mathrm{G}} \right] + \frac{A}{1 - \frac{4 \pi}{3} | \Lambda | \varrho^2 a_N b^2}  \label{eq:upperbound22} 
\end{equation}
as an upper bound for the energy, where
\begin{equation}
A = 2 \int_{\Lambda^2} \xi(d(x,y)) \varrho^{(2)}_{\Gamma_{N,0}^{\mathrm{G}}}(x,y)  \text{d}(x,y) + 6 \int_{\Lambda^3} \left\vert \nabla f_b(d(x,y)) \nabla f_b(d(z,y)) \right\vert \varrho^{(3)}_{\Gamma_{N,0}^{\mathrm{G}}}(x,y,z) \text{d}(x,y,z). \label{eq:upperbound22b} 
\end{equation}
To derive this bound we had to assume that $(4 \pi /3) | \Lambda | \varrho^2 a_N b^2 < 1$. 

Let us denote by $a_p^*$ and $a_p$ the usual creation and annihilation operators of a plane wave state $\varphi_p(x) = | \Lambda |^{-1/2} e^{ipx}$ in $\Lambda$ for $p \in \frac{2 \pi}{L} \mathbb{Z}^3$. Also let $n_p = a_p^* a_p$ be the related occupation number operator. To bound the first term on the right-hand side of \eqref{eq:upperbound22b}, we use that $\Gamma_{N,0}^{\mathrm{G}}$ has a fixed number of particles, and hence $\sum_{p \in \frac{2 \pi}{L} \mathbb{Z}^3} n_p$ can always be replaced by $N$ when acting on $\Gamma_{N,0}^{\mathrm{G}}$. This implies 
\begin{equation}
	\sum_{p,q \in \frac{2 \pi}{L} \mathbb{Z}^3} \langle n_p n_q  \rangle_{\Gamma_{N,0}^{\mathrm{G}}} = \sum_{p,q \in \frac{2 \pi}{L} \mathbb{Z}^3} \langle n_p  \rangle_{\Gamma_{N,0}^{\mathrm{G}}} \langle n_q  \rangle_{\Gamma_{N,0}^{\mathrm{G}}}.
	\label{eq:upperbound22bb} 
\end{equation}
When we use that all eigenfunctions $\varphi_p$ of $-\Delta$ are in absolute value equal to $| \Lambda |^{-1/2}$ (they are plane waves), we come to the second line in the following inequality:
\begin{align}
\varrho^{(2)}_{\Gamma_{N,0}^{\mathrm{G}}}(x,y) &= \sum_{p_1,p_2}\left| \frac{1}{2}  \sum_{\sigma \in S_2} \varphi_{\sigma(p_1)}(x) \varphi_{\sigma(p_2)}(y) \right|^2 \left\langle a_{p_1}^* a_{p_2}^* a_{p_2} a_{p_1}  \right\rangle_{\Gamma_{N,0}^{\mathrm{G}}} \label{eq:upperbound22c} \\
&\leq \frac{| \Lambda |^{-2}}{2} \left[ 2 \sum_{p \neq q} \langle n_p n_q  \rangle_{\Gamma_{N,0}^{\mathrm{G}}} + \sum_{p} \langle n_p (n_p - 1)  \rangle_{\Gamma_{N,0}^{\mathrm{G}}} \right] \nonumber \\
&\leq \frac{| \Lambda |^{-2}}{2} \left[ 2 \sum_{p,q} \langle n_p n_q  \rangle_{\Gamma_{N,0}^{\mathrm{G}}} - \sum_p \langle n_p^2 \rangle_{\Gamma_{N,0}^{\mathrm{G}}} \right] \nonumber \\
&\leq \frac{| \Lambda |^{-2}}{2} \left[ 2 N^2 - \langle n_0 \rangle_{\Gamma_{N,0}^{\mathrm{G}}}^2 \right].  \nonumber
\end{align}
By $S_2$ we denote the group of permutations of two elements. To arrive at the third line, we simply estimated $n_p (n_p -1) \leq n_p^2$ and to come to the last line we used \eqref{eq:upperbound22bb} and
\begin{equation}
\sum_p \langle n_p^2  \rangle_{\Gamma_{N,0}^{\mathrm{G}}} \geq \langle n_0^2 \rangle_{\Gamma_{N,0}^{\mathrm{G}}} \geq \langle n_0 \rangle_{\Gamma_{N,0}^{\mathrm{G}}}^2.
\end{equation}
Since $\Gamma_{N,0}^{\mathrm{G}}$ is the Gibbs state of the ideal gas we have $| \Lambda |^{-1} \langle n_0 \rangle_{\Gamma_{N,0}^{\mathrm{G}}} = \varrho_0(\beta,N,L)$ by definition.

The second term on the right-hand side of Eq.~\eqref{eq:upperbound22b} can be treated with a rough bound that we derive now. An application of the Cauchy-Schwarz inequality tells us that 
\begin{align} \nonumber
\varrho^{(3)}_{\Gamma_{N,0}^{\mathrm{G}}}(x,y,z) &=  \sum_{p_1,p_2,p_3}\left| \frac{1}{6}  \sum_{\sigma \in S_3} \varphi_{\sigma(p_1)}(x) \varphi_{\sigma(p_2)}(y) \varphi_{\sigma(p_3)}(z) \right|^2 \left\langle a_{p_1}^* a_{p_2}^* a_{p_3}^* a_{p_3} a_{p_2} a_{p_1}  \right\rangle_{\Gamma_{N,0}^{\mathrm{G}}}  \\  \label{eq:upperbound23a} 
&\leq | \Lambda |^{-3}  \sum_{p_1,p_2,p_3} \left\langle a_{p_1}^* a_{p_2}^* a_{p_3}^* a_{p_3} a_{p_2} a_{p_1}  \right\rangle_{\Gamma_{N,0}^{\mathrm{G}}} \\
&\leq \varrho^3, \nonumber
\end{align}
where $S_3$ denotes the group of permutations of three elements. We insert this bound into the second term on the right-hand side of Eq.~\eqref{eq:upperbound22b} and obtain 
\begin{equation}
\int_{ \Lambda^3} \left\vert \nabla f_b(d(x,y)) \nabla f_b(d(z,y)) \right\vert \varrho^{(3)}_{\Gamma_{N,0}^{\mathrm{G}}}(x,y,z) \text{d}(x,y,z) \leq \varrho^3 | \Lambda | \left( \int_{\mathbb{R}^3} \left| \nabla f_b(|x|) \right| \text{d}x \right)^2. \label{eq:upperbound23c}
\end{equation}
An explicit computation together with the bound $f_0(|x|) \geq \left[ 1-a_N/|x| \right]_+$, see \cite[Appendix~C]{Themathematicsofthebosegas} shows for $a_N < b \eta$ with $0 < \eta < 1$ that $\int_{\mathbb{R}^3} \left| \nabla f_b(|x|) \right| \text{d}x \lesssim a_N b$. We combine this with Eqs.~\eqref{eq:upperbound22}, \eqref{eq:upperbound22b} and \eqref{eq:upperbound22c}, and use $\int_{\Lambda} \xi(d(x,y)) \de x \leq 4 \pi a_N/(1-a_N/b)$, see again \cite[Appendix~C]{Themathematicsofthebosegas}, to finally obtain
\begin{equation}
\tr\left[ H_N \widetilde{\Gamma}_{N,0}^{\mathrm{G}} \right] \leq \tr\left[ H_{N,0} \Gamma_{N,0}^{\mathrm{G}} \right] + \frac{\widetilde{A}}{1 - \frac{4 \pi}{3} | \Lambda | \varrho^2 a_N b^2} \label{eq:upperbound23d}
\end{equation}
with
\begin{equation}
\widetilde{A} = \frac{ 4 \pi a_N | \Lambda | \left[ 2 \varrho^2 - \varrho_0^2(\beta,N,L) \right]}{1- \frac{a_N}{b}} + \text{const. } | \Lambda | (a_N b)^2 \varrho^3.
\label{eq:upperbound23e}
\end{equation}
For the derivation of this result, we assumed $(4 \pi /3) | \Lambda | \varrho^2 a_N b^2 < 1$ and $a_N < b \eta$ with $0 < \eta < 1$. In the next step we will estimate the entropy of the state $\widetilde{\Gamma}_{N,0}^{\mathrm{G}}$ in terms of the entropy of $\Gamma_{N,0}^{\mathrm{G}}$ and compute the final upper bound. 
\subsection{The entropy and the final upper bound}
To relate the entropy of the state $\widetilde{\Gamma}_{N,0}^{\mathrm{G}}$ to the one of $\Gamma_{N,0}^{\mathrm{G}}$, we use \cite[Lemma~2]{RobertFermigas} which we spell out here for the sake of completeness. 

\begin{lemma}
	\label{lem:upperbound4}
	Let $\Gamma$ be a density matrix on some Hilbert space, with eigenvalues $\lambda_{\alpha} \geq 0$. Additionally, let $\lbrace P_{\alpha} \rbrace_{\alpha=1}^{\infty}$ be a family of  one-dimensional orthogonal projections (for which $P_{\alpha} P_{\alpha'} = \delta_{\alpha,\alpha'} P_{\alpha}$ need not necessarily be true) and define $\hat{\Gamma} = \sum_{\alpha=1}^{\infty} \lambda_{\alpha} P_{\alpha}$. Then
	\begin{equation}
	S( \hat{\Gamma} ) \geq S(\Gamma) - \ln\left( \left\Vert \textstyle{ \sum_{\alpha=1}^{\infty} P_{\alpha} } \right\Vert \right).
	\label{eq:upperbound24}
	\end{equation} 
\end{lemma}

Since $ 0 \leq F \leq 1$ we have the following operator inequality
\begin{equation}
\sum_{\alpha=1}^{\infty} \frac{\vert F \Psi_{\alpha} \rangle\langle F \Psi_{\alpha} \vert}{\Vert F \Psi_{\alpha} \Vert^2} \leq \left( \sup_{\alpha \geq 1} \Vert F \Psi_{\alpha} \Vert^{-2} \right) F^2 \leq \left( \sup_{\alpha \geq 1} \Vert F \Psi_{\alpha} \Vert^{-2} \right). \label{eq:upperbound25}
\end{equation} 
Eq.~\eqref{eq:upperbound25} together with Lemmas~\ref{lem:denominators} and~\ref{lem:upperbound4} shows that
\begin{equation}
- T S \left( \tilde{\Gamma}_{N,0}^{\mathrm{G}} \right) \leq -T S \left( \Gamma_{N,0}^{\mathrm{G}} \right) + O \left( \beta^{-1} | \Lambda | \varrho^2 a_N b^2 \right) 
\label{eq:upperbound25b}
\end{equation}
holds as long as $(4 \pi /3) | \Lambda | \varrho^2 a_N b^2 < 1$. 

Having the bound for the entropy at hand, we compute the free energy. With Eqs.~\eqref{eq:freeenergytrialstate}, \eqref{eq:upperbound23d}, \eqref{eq:upperbound23e} and \eqref{eq:upperbound25b} we find
\begin{align}
\tr \left[ H_N \tilde{\Gamma}_{N,0}^{\mathrm{G}} \right] - T S \left( \tilde{\Gamma}_{N,0}^{\mathrm{G}} \right) \leq& F_0(\beta,N,L) + | \Lambda | 4 \pi a_N \left[ 2 \varrho^2 - \varrho_0^2(\beta,N,L) \right] + O\left( | \Lambda | \varrho^2 a_N^2 b^{-1} \right) \label{eq:final1upperbound} \\
&+ O\left( | \Lambda |^2 \varrho^4 a_N^2 b^2 \right) + O\left( | \Lambda | (a_N b )^2 \varrho^3 \right) + O\left( | \Lambda | \beta^{-1} \varrho^2 a_N b^2 \right). \nonumber
\end{align}
To obtain the result we assumed $\frac{4 \pi}{3} | \Lambda | \varrho^2 a_N b^2 < 1$ and $a_N < b \eta$ with $0 < \eta < 1$. Optimization yields $b = (a_N/(N a_N \varrho + \beta^{-1}))^{1/3}$ and the bound
\begin{align}
	F(\beta,\varrho,L) \leq F_0(\beta,N,L) + | \Lambda | 4 \pi a_N \left( 2 \varrho^2 - \varrho_0^2(\beta,N,L) \right) \left( 1 + O\left( a_N (N \varrho)^{1/3} \right) + O \left( \frac{\left( a_N \varrho^{1/3} \right)^{2/3}}{\left( \beta \varrho^{2/3} \right)^{1/3}} \right) \right).
	\label{eq:final2upperbounda}
\end{align}
We note that this bound is uniform in the parameter space $\beta \varrho^{2/3} \gtrsim 1$. When we use $a_N \lesssim L N^{-1}$ and $\beta \varrho^{2/3} \gtrsim 1$, \eqref{eq:final2upperbounda} implies
\begin{align}
F(\beta,N,L) \leq F_0(\beta,N,L) + | \Lambda | 4 \pi a_N \left( 2 \varrho^2 - \varrho_0^2(\beta,N,L) \right) \left( 1 + O\left( N^{-1/3} \right) \right). 
\label{eq:final2upperbound}
\end{align}
This completes the proof of the upper bound.
\section{Proof of the lower bound}
\label{sec:prooflowerbound}
The proof of the lower bound proceeds along similar lines as the proof of the lower bound for the free energy in the thermodynamic limit \cite{Sei2008}. One crucial ingredient in this work is a c-number substitution for momentum modes smaller than some cutoff which allows one to include a condensate. From a technical point of view, the proof in \cite{Sei2008} is written in terms of the interacting Gibbs state of the system and uses the Berezin-Lieb inequality. The main difference between our setting and the one in \cite{Sei2008} is that we also want to make a statement about the 1-pdm of approximate minimizers of the Gibbs free energy functional. To that end, we develop an alternative approach that is based on the Gibbs variational principle and goes hand in hand with the c-number substitution, and therefore also with the approach in \cite{Sei2008}. To prove the statement about the 1-pdm of approximate minimizers of the Gibbs free energy functional, it will be necessary to prove the lower bound for the free energy related to the more general Hamiltonian
\begin{equation}
H_N^{\lambda} = H_N + \lambda \sum_{i=1}^N | \Phi \rangle \langle \Phi |_i,
\label{eq:1pdmcondensate1a}
\end{equation}
where $\Phi(x) = | \Lambda |^{-1/2}$ and the index $i$ indicates that the projection acts on the $i$-th particle. By adding this term we shift the energy of the lowest eigenvalue of $-\Delta$ by $\lambda$. In the following, we will assume that $\lambda \in [0 , (2 \pi/L )^2 \eta ]$ for some $0 < \eta < 1$.

Before presenting the details of the lower bound, we give for the convenience of the reader a short summary of the main ideas of the proof in \cite{Sei2008} in the context of the present setting.
\subsubsection*{Strategy of the proof of the lower bound}
A key ingredient in the proof of the lower bound for the free energy of the interacting system in \cite{Sei2008} is the observation that the interaction energy in \eqref{eq:main1} is, as long as $\beta \sim \beta_{\mathrm{c}}$, much smaller than the free energy of the ideal gas $F_0(\beta,N,L)$ (compare with Remark~4 in Section~\ref{sec:maintheorem}). A naive version of first order perturbation theory is, however, not applicable because the interaction energy of the Gibbs state of the ideal gas is too large (it is even infinite if hard spheres are considered), see also the discussion in the beginning of Section~\ref{sec:proofupperbound}. In the case of the GP limit with Dirichlet boundary conditions there is also a second obstacle, namely that the condensate wave function of the interacting system is not given by the ground state of the Laplacian in the box (as in case of the ideal gas), but rather by the minimizer of the GP energy functional \eqref{eq:GPfunctional}.

The first problem is solved in \cite{Sei2008} with the help of a Dyson Lemma. It allows to replace the singular and short ranged interaction potential $v_N$ by a softer potential with longer range. The price to pay for this replacement is a certain amount of kinetic energy. In the positive temperature setting it is important that only modes with momenta much larger than $\beta^{-1/2}$ are used in the Dyson Lemma, since the other modes are needed to obtain the free energy of the ideal gas in \eqref{eq:main1}.

After the replacement of the interaction potential a rigorous version of first order perturbation theory is applied. It is based on a correlation inequality \cite{RobertCorrelationinequ} that is applicable to fermionic systems and to bosonic system above the critical temperature for BEC of the ideal gas. It allows to replace a general state in the expectation of the interaction energy by the Gibbs state of the related ideal gas and to estimate the  error. An essential ingredient for this method is that the reference state in the perturbative analysis shows an approximate tensor product structure w.r.t. localization in different regions in position space. For a quasi-free state this is true if the off-diagonal of its 1-pdm decays sufficiently fast in position space. In order to overcome this shortcoming, coherent states are used in \cite{Sei2008} to replace creation and annihilation operators of certain low momentum modes by complex numbers. In particular, this allows for the description of a condensate. Coherent states show an exact tensor product structure w.r.t. spatial localization in different regions in space, and therefore fit seamlessly into the above framework. In the case of Dirichlet boundary conditions, this approach also allows us to take into account that the condensate wave function is given by the GP minimizer, which solves the second problem from above.

The statement in Theorem~\ref{thm:periodic} is uniform in the temperature as long as $\beta \gtrsim \beta_{\mathrm{c}}$. If the temperature is sufficiently low the free energy of the ideal gas in \eqref{eq:main1} is much smaller than the interaction energy and the approach from above cannot be expected to work. To extend the proof to this regime, we apply a different technique that uses in an essential way the  zero temperature result in \cite{LiYng1998}.
\subsection{Reduction to integrable potential}
In the proof of the lower bound we will make use of Fock spaces and, in particular, it will be required that the interaction potential has finite Fourier coefficients. As in \cite[Section~2.1]{Sei2008} we are therefore going to replace the potential $v$ by an integrable potential. This is achieved with the following lemma whose proof can be found in \cite[Sec.~2.1]{Sei2008}.
\begin{lemma}
	\label{lem:integrablepotential}
	Let $v : \mathbb{R}_+ \mapsto \mathbb{R}_+ \cup \{ + \infty \}$ have a finite scattering length $a_v$. For any $\epsilon > 0$ and any $\varphi>0$, there exists a function $\widetilde{v}$ with $0 \leq \widetilde{v}(r) \leq v(r)$ for all $r$, such that $\int_0^{\infty} \widetilde{v}(r) r^2 \de r \leq 2 \varphi$, and such that the scattering length $\widetilde{a}$ of $\widetilde{v}$ satisfies
	\begin{equation}
		\widetilde{a} \geq a_v \left( 1 - \sqrt{a_v/\varphi} \right) (1-\epsilon).
	\end{equation}
\end{lemma}
Since $0 \leq \widetilde{v}(r) \leq v(r)$ for all $r$, we can replace $v$ by $\widetilde{v}$ in the Hamiltonian $H_N^{\lambda}$ for a lower bound. The Hamiltonian we obtain by this procedure will be denoted by $\widetilde{H}_N^{\lambda}$. We also define $\widetilde{a}_N$ to be the scattering length of the scaled potential $\widetilde{v}_N$.
\subsection{Fock space}
In the proof of the lower bound it is convenient to give up the restriction on the number of particles and to work in Fock space instead of in $\mathcal{H}_N$. In this section we introduce the necessary notation for this analysis. By $\mu(\lambda)$ we denote the chemical potential of the ideal Bose gas related to the one-particle Hamiltonian $-\Delta + \lambda | \Phi \rangle \langle \Phi |$, leading to an expected number of $N$ particles, and we define $\mu_0 = \mu(0)$. Let $\mathscr{F}$ be the Fock space over $L^2(\Lambda)$. We define the Hamiltonian $\mathbb{H}^{\lambda}$ on $\mathscr{F}$ by
\begin{equation}
	\mathbb{H}^{\lambda} = \mathbb{T}^{\lambda} + \mathbb{V},
\end{equation}
where the kinetic and the potential energy operators are given by
\begin{equation}
	\mathbb{T}^{\lambda} = \sum_p (p^2 - \mu(\lambda) + \lambda \delta_{p,0}) a_p^* a_p, \quad \text{ and } \quad \mathbb{V} = \frac{1}{2 | \Lambda |} \sum_{p,k,\ell} \widehat{v}_N(p) a_{k+p}^* a_{\ell - p}^* a_{k} a_{\ell},
	\label{eq:TandV}
\end{equation}
respectively. Here $\delta_{p,0}$ denotes the Kronecker delta. The Fourier coefficients of $\widetilde{v}_N(|x|) = (L^{-1} N)^2 \widetilde{v}(N |x|/L)$ are denoted by $\widehat{v}_N(p) = \int_{\Lambda} \widetilde{v}_N(|x|) e^{-ipx} \de x$. Under the assumption $R_0 < N/2$, they are given in terms of the Fourier coefficients $\widehat{v}$ of $\widetilde{v}$ by $\widehat{v}_N(p) = L N^{-1} \widehat{v}(Lp/N)$. By construction, the Fourier coefficients of $\widehat{v}_N$ are bounded in absolute value by $| \widehat{v}_N(p) | \leq \widehat{v}_N(0) \leq 8 \pi \varphi L N^{-1}$, where $\varphi$ has been introduced in the previous section. In the following we will denote the grand canonical kinetic energy operator for $\lambda = 0$ by $\mathbb{T}$ and similarly for the full Hamiltonian.
\subsection{Coherent states and the Gibbs variational principle}
In this section we introduce a formalism that allows us to apply a c-number substitution while still keeping information on a given state $\Gamma$ whose free energy we want to investigate. We start by introducing notation for the c-number substitution. Let us pick some $p_{\mathrm{c}} >0$ and decompose the Fock space as $\mathscr{F} \cong \mathscr{F}_< \otimes \mathscr{F}_{>}$, where $\mathscr{F}_{<}$ and $\mathscr{F}_{>}$ denote the Fock spaces of the momentum modes with $|p| < p_{\mathrm{c}}$ and $|p| \geq p_{\mathrm{c}}$, respectively. The trace over $\mathscr{F}_<$ will be denoted by $\tr_<$ and similarly for $\mathscr{F}_{>}$. To keep the notation simple and because we do not expect it to cause confusion, we will denote the traces over $\mathscr{F}$ and $\mathcal{H}_N$ by the same symbol $\tr$. By $M$ we denote the number of momenta $p \in \frac{2 \pi}{L} \mathbb{Z}^3$ with $|p| < p_{\mathrm{c}}$. For a vector $z \in \mathbb{C}^M$ we introduce the coherent state $| z \rangle \in \mathscr{F}_<$ by 
\begin{equation}
	| z \rangle  = \exp \left( \sum_{|p| < p_{\mathrm{c}}} z_p a_p^* - \overline{z_p} a_p \right) | \text{vac} \rangle \equiv U(z) | \text{vac} \rangle,
	\label{eq:U}
\end{equation}
where $| \text{vac} \rangle$ denotes the vacuum in $\mathscr{F}_<$. Coherent states of this kind form an overcomplete basis with $\int_{\mathbb{C}^M} | z \rangle\langle z | \de z = \mathds{1}_{\mathscr{F}_<}$. Here $\de z = \prod_{| p | < p_{\mathrm{c}}}^M \de z_p$ with $z_p = x_p+i y_p$ and $\de z_p = \frac{\de x_p \de y_p}{\pi}$. For every state $\Gamma$ on the Fock space $\mathscr{F}$, we define the operator $\widetilde{\Gamma}_z$ acting on $\mathscr{F}_{>}$ by 
\begin{equation}
	\widetilde{\Gamma}_z = \langle z , \Gamma z \rangle = \tr_< | z \rangle \langle z| \ \Gamma.
\end{equation}
Additionally, we denote
\begin{equation}
\zeta_{\Gamma}(z) = \tr_>[ \widetilde{\Gamma}_z ]. 
\label{eq:measurenu}
\end{equation}
Since $\Gamma$ is a state, $\zeta_{\Gamma}(z) \de z$ is a probability measure on $\mathbb{C}^M$. The entropy of the classical distribution $\zeta_{\Gamma}$ is defined by
\begin{equation}
S( \zeta_{\Gamma} ) = - \int_{\mathbb{C}^M} \ln\left( \zeta_{\Gamma}(z) \right) \zeta_{\Gamma}(z) \de z.
\label{eq:entropyofnu}
\end{equation} 
On the level of the Hamiltonian, we will need the \textit{lower symbol} of $\mathbb{H}^{\lambda}$ which is defined by $\mathbb{H}_{\mathrm{s}}^{\lambda} \coloneqq \langle z, \mathbb{H}^{\lambda} z \rangle$. It is an operator-valued function from $\mathbb{C}^M$ into the unbounded operators on $\mathscr{F}_>$. 
Since $a_p | z \rangle = z_p | z \rangle$, the lower symbol can be obtain from $\mathbb{H}^{\lambda}$ by simply replacing $a_p$ by $z_p$ and $a_p^*$ by $\overline{z_p}$ for all $| p | < p_{\mathrm{c}}$. By $\mathbb{H}^{\lambda,\mathrm{s}}(z)$ we denote the \textit{upper symbol} of the Hamiltonian $\mathbb{H}^{\lambda}$ which is defined by the identity
\begin{equation}
\mathbb{H}^{\lambda} = \int_{\mathbb{C}^M} \mathbb{H}^{\lambda,\mathrm{s}}(z) \ |z \rangle\langle z | \ \de z.
\label{eq:uppersymbol}
\end{equation} 
To compute it, one has to replace $| z_p |^2$ by $| z_p |^2 - 1$ in the lower symbol and similarly with other polynomials in $z_p$, see \cite{Lieb2005}. 

The following Lemma shows that the entropy of a state $\Gamma$ can be bounded from above in terms of the expectation of the entropies of the states 
\begin{equation}
	\Gamma_z = \frac{\widetilde{\Gamma}_z}{\tr_> \widetilde{\Gamma}_z}
\end{equation}
acting on $\mathscr{F}_{>}$ w.r.t. the probability measure $\zeta_{\Gamma}(z) \de z$, plus one additional term that quantifies the entropy of the classical distribution $\zeta_{\Gamma}(z)$.
\begin{lemma}
	\label{lem:entropyz}
	Let $\Gamma$ be a state on $\mathscr{F}$. The entropy of $\Gamma$ is bounded in the following way:
	\begin{equation}
		S(\Gamma) \leq \int_{\mathbb{C}^M}  S\left( \Gamma_z \right) \zeta_{\Gamma}(z) \de z + S(\zeta_{\Gamma}).
		\label{eq:lemmaentropyboundz}
	\end{equation}
\end{lemma}
\begin{proof}
We write the first term on the right-hand side of \eqref{eq:lemmaentropyboundz} as
\begin{equation}
	\int_{\mathbb{C}^M}  S\left( \Gamma_z \right) \zeta_{\Gamma}(z) \de z = \int_{\mathbb{C}^M} S(\widetilde{\Gamma}_z) \de z + \int_{\mathbb{C}^M} \left( \tr_> \widetilde{\Gamma}_z \right) \ln\left( \tr_> \widetilde{\Gamma}_z \right) \de z.
	\label{eq:lemmaentropyboundz1}
\end{equation}	
To prove the result, we need to show that
\begin{equation}
	\int_{\mathbb{C}^M} S(\widetilde{\Gamma}_z) \de z \geq S(\Gamma)
	\label{eq:lemmaentropyboundz1b}
\end{equation}
holds. To that end, we expand $\Gamma = \sum_{\alpha=1}^{\infty} \lambda_{\alpha} | \Psi_{\alpha} \rangle\langle \Psi_{\alpha} |$ which implies $\widetilde{\Gamma}_z =  \sum_{\alpha = 1}^{\infty} \lambda_{\alpha}| \Psi^z_{\alpha} \rangle\langle \Psi^z_{\alpha} |$. In the second equality we denoted $\Psi_{\alpha}^z = \langle z, \Psi_{\alpha} \rangle$. Since $\int_{\mathbb{C}^M} \langle \Psi^z_{\alpha}, \Psi^z_{\alpha} \rangle \de z = 1$, which follows from $\int_{\mathbb{C}^M} | z \rangle\langle z | \de z = \mathds{1}_{\mathscr{F}_{<}}$, Eq.~\eqref{eq:lemmaentropyboundz1b} is equivalent to
\begin{equation}
	- \int_{\mathbb{C}^M} \sum_{\alpha=1}^{\infty} \lambda_{\alpha} \langle \Psi_{\alpha}^z , \ln(\widetilde{\Gamma}_z  \lambda_{\alpha}^{-1}) \Psi_{\alpha}^z \rangle \de z \geq 0.
	\label{eq:lemmaentropyboundz1c}
\end{equation}
We apply Jensen's inequality to show that $\| \Psi_{\alpha}^z \|^{-2}  \langle \Psi_{\alpha}^z , \ln(\widetilde{\Gamma}_z  \lambda_{\alpha}^{-1}) \Psi_{\alpha}^z \rangle \leq \ln( \| \Psi_{\alpha}^z \|^{-2}  \langle \Psi_{\alpha}^z , \widetilde{\Gamma}_z  \lambda_{\alpha}^{-1} \Psi_{\alpha}^z \rangle )$. Hence, the left-hand side of \eqref{eq:lemmaentropyboundz1c} is bounded from below by
\begin{equation}
	- \int_{\mathbb{C}^M} \sum_{\alpha=1}^{\infty} \lambda_{\alpha} \| \Psi_{\alpha}^z \|^{2} \ln\left( \| \Psi_{\alpha}^z \|^{-2} \langle \Psi_{\alpha}^z , \widetilde{\Gamma}_z  \lambda_{\alpha}^{-1} \Psi_{\alpha}^z \rangle \right) \de z. 
	\label{eq:lemmaentropyboundz1d}
\end{equation}
The measure $\lambda_{\alpha} \| \Psi_{\alpha}^z \|^{2} \de z$ is a probability measure with respect to summation over $\alpha \in \mathbb{N}$ and integration over $\mathbb{C}^M$ in $z$. Another application of Jensen's inequality therefore tells us that
\begin{align}
	- \int_{\mathbb{C}^M} \sum_{\alpha=1}^{\infty} \lambda_{\alpha} \| \Psi_{\alpha}^z \|^{2} \ln\left( \| \Psi_{\alpha}^z \|^{-2} \langle \Psi_{\alpha}^z , \widetilde{\Gamma}_z  \lambda_{\alpha}^{-1} \Psi_{\alpha}^z \rangle \right) \de z &\geq - \ln \left( \int_{\mathbb{C}^M}   \sum_{\alpha=1}^{\infty} \langle \Psi_{\alpha}^z , \widetilde{\Gamma}_z \Psi_{\alpha}^z \rangle  \de z \right) \label{eq:lemmaentropyboundz1e} \\
	&= -\ln \left( \int_{\mathbb{C}^M} \tr_> \widetilde{\Gamma}_z \de z \right) = 0.  \nonumber
\end{align}
To come to the last line, we used $\sum_{\alpha=1}^{\infty} | \Psi_{\alpha}^z \rangle\langle \Psi_{\alpha}^z | = \mathds{1}_{\mathscr{F}_>}$. This proves the claim. 
\end{proof}
With the definitions from above and Lemma~\ref{lem:entropyz}, we can derive a lower bound for the Gibbs free energy functional. Let $\Gamma$ be a state on $\mathcal{H}_N \subset \mathscr{F}$.  Eq.~\eqref{eq:uppersymbol} allows us to write the expectation of the energy as 
\begin{equation}
	\tr \left[ H_N^{\lambda} \Gamma \right] = \mu(\lambda) N + \int_{\mathbb{C}^M} \tr \left[ \mathbb{H}^{\lambda,\mathrm{s}}(z) \ |z \rangle \langle z |  \ \Gamma \right] \de z = \mu(\lambda) N + \int_{\mathbb{C}^M} \tr_> \left[\mathbb{H}^{\lambda,\mathrm{s}}(z) \Gamma_z \right] \zeta_{\Gamma}(z) \de z.
\end{equation}
In combination with Lemma~\ref{lem:entropyz}, this implies
\begin{equation}
	\tr \left[ H^{\lambda}_N \Gamma \right] - \tfrac{1}{\beta} S(\Gamma) \geq \mu(\lambda) N + \int_{\mathbb{C}^M} \left\{ \tr_> \left[\mathbb{H}^{\lambda,\mathrm{s}}(z) \Gamma_z \right] - \tfrac{1}{\beta} S(\Gamma_z) \right\} \zeta_{\Gamma}(z) \de z - \tfrac{1}{\beta} S(\zeta_{\Gamma}).
	\label{eq:lemmaentropyboundz5}
\end{equation}
Although the upper symbol naturally appears in the above inequality, it is more convenient to work with the lower symbol instead. Let $\mathbb{N}_{\mathrm{s}} = |z|^2 + \sum_{|p| \geq p_{\mathrm{c}}} a_p^* a_p$ denote the lower symbol of the particle number operator. The difference between the upper and the lower symbol $\Delta \mathbb{H}^{\lambda}(z) = \mathbb{H}^{\lambda}_{\mathrm{s}}(z) - \mathbb{H}^{\lambda,\mathrm{s}}(z)$ can be written as
\begin{align}
\Delta \mathbb{H}^{\lambda}(z) =& \sum_{| p | < p_c} \left( p^2 + \delta_{p,0} \lambda - \mu(\lambda) \right) + \frac{1}{2 | \Lambda |} \bigg[ \widehat{v}_N(0) \left( 2 M \mathbb{N}_s(z) - M^2 \right) \label{eq:lemmaentropyboundz6} \\
&+ 2 \sum_{| \ell | <p_c, |k| \geq p_c} \widehat{v}_N(\ell - k) a_k^{*} a_k + \sum_{| \ell |, |k| <p_c} \widehat{v}_N(\ell - k) \left( 2 |z_k|^2 - 1 \right) \bigg]. \nonumber
\end{align}
The bound $| \widehat{v}_N | \leq 8 \pi \varphi L N^{-1}$ therefore implies
\begin{equation}
	\Delta \mathbb{H}^{\lambda}(z) \leq M \left( p_c^2 - \mu(\lambda) \right) + \lambda + \frac{16 \pi \varphi L}{| \Lambda | N} M \mathbb{N}_{\mathrm{s}}(z)
	\label{eq:lemmaentropyboundz7}
\end{equation}
as well as
\begin{equation}
	\int_{\mathbb{C}^M} \tr \left[ \Delta \mathbb{H}^{\lambda}(z) \Gamma_z \right] \zeta_{\Gamma}(z) \de z \leq M (p_c^2 - \mu(\lambda)) + \lambda + \frac{16 \pi \varphi L}{| \Lambda | N} M \left( M+N \right) \equiv Z^{(1)}.
	\label{eq:lemmaentropyboundz8}
\end{equation}
To obtain the second bound, we used that $\int_{\mathbb{C}^M} \tr_>[ \mathbb{N}_{\mathrm{s}}(z) \Gamma_z ] \zeta_{\Gamma}(z) \de z = N + M$. Eq.~\eqref{eq:lemmaentropyboundz8} allows us to replace the upper by the lower symbol in \eqref{eq:lemmaentropyboundz5} in a controlled way. Before we state the final result, let us introduce the state $\Upsilon^z$ on $\mathscr{F}$ by
\begin{equation}
	\Upsilon^z = | z \rangle \langle z | \otimes \Gamma_z.
	\label{eq:lemmaentropyboundz9}
\end{equation}
We have $\tr [\mathbb{H}^{\lambda}_s(z) \Gamma_z] = \tr [\mathbb{H}^{\lambda} \Upsilon^z]$ and $S(\Gamma_z) = S(\Upsilon^z)$. Putting \eqref{eq:lemmaentropyboundz5} and \eqref{eq:lemmaentropyboundz8} together, we have thus shown that
\begin{equation}
	\tr \left[ H_N^{\lambda} \Gamma \right] - \tfrac{1}{\beta} S(\Gamma) \geq \mu(\lambda) N +  \int_{\mathbb{C}^M} \left\{ \tr \left[ \mathbb{H}^{\lambda} \Upsilon^z \right] - \tfrac{1}{\beta} S\left( \Upsilon^z \right) \right\} \zeta_{\Gamma}(z) \de z - \tfrac{1}{\beta} S(\zeta_{\Gamma}) - Z^{(1)}.
	\label{eq:lemmaentropyboundz11}
\end{equation}
We will later choose the parameters $p_{\mathrm{c}}$ and $\varphi$ such that $Z^{(1)} \ll | \Lambda | a_N \varrho^2$. Eq.~\eqref{eq:lemmaentropyboundz11} is the formula we were looking for. It should be compared to (2.3.9) and (2.3.10) in \cite{Sei2008}, in which a version of the grand canonical Gibbs state of the interacting system appears. In contrast to that, \eqref{eq:lemmaentropyboundz11} allows to use the c-number substitution while still working with a given state $\Gamma$. The Gibbs variational principle applied to $\tr \mathbb{H}^{\lambda} \Upsilon^z - \tfrac{1}{\beta} S\left( \Upsilon^z \right)$ will later allow us to obtain information on an approximate minimizer of the Gibbs free energy functional \eqref{eq:Gibbsfreeenergyfunctional}, see Sec.~\ref{sec:asymptotics1pdm}.
\begin{remark}
\label{rem:linktoRobertspaper1}
In \cite{Sei2008} the additional term 
\begin{equation}
	\mathbb{K} = 4 \pi \widetilde{a} \frac{C}{| \Lambda |} (\mathbb{N} - N)^2
	\label{eq:K}
\end{equation}
is added to the second quantized Hamiltonian before relaxing the restriction on the particle number. Like this one obtains a strong control on the expected number of particles in the system. We do not need this term in our approach because the information that the state $\Gamma$ has exactly $N$ particles is still encoded in the Fock space formalism through the state $\Gamma_z$ and the measure $\zeta_{\Gamma}(z) \de z$.  
\end{remark}

In the remaining part of this section we will go through the proof in \cite{Sei2008}, mention changes due to our approach and collect the necessary results. The following sections will be named like the ones in \cite{Sei2008}.
\subsection{Relative entropy and a-priori bounds}
In this section we derive an a-priori bound for states $\Gamma$ whose free energy is small in an appropriate sense. This bound is the only information we are going to need about the state to prove the lower bound.

For two general states $\Gamma$ and $\Gamma'$ on Fock space we denote by 
\begin{equation}
	S(\Gamma,\Gamma') = \tr \left[ \Gamma \left( \ln\left( \Gamma \right) - \ln\left( \Gamma' \right) \right) \right] 
	\label{eq:relentropy}
\end{equation}
the relative entropy of $\Gamma$ with respect to $\Gamma'$. It is a nonnegative functional that equals zero if and only if $\Gamma = \Gamma'$. Let $\Gamma^0$ be the Gibbs state corresponding to $\mathbb{T}_{\mathrm{s}}(z)$ at inverse temperature $\beta$ on $\mathscr{F}_>$, which is independent of $z$. We emphasize that $\mathbb{T}_{\mathrm{s}}(z)$ is the lower symbol of the grand canonical kinetic energy operator with $\lambda = 0$. Since the interaction potential $\widetilde{v}_N$ and $\lambda$ are nonnegative, we have
\begin{equation}
\tr \left[ \mathbb{H}^{\lambda} \Upsilon^z \right] - \tfrac{1}{\beta} S\left( \Upsilon^z \right) \geq \tr \left[ \mathbb{T} \Upsilon^z \right] - \tfrac{1}{\beta} S\left(\Upsilon^z \right) \geq -\tfrac{1}{\beta} \ln \tr_> \left[ e^{-\beta \mathbb{T}_s(z)} \right] + \tfrac{1}{\beta} S\left(\Gamma_z,\Gamma^0\right). 
\label{eq:aprioribounds1}
\end{equation}
Let us integrate both sides of the above equation with $\zeta_{\Gamma} (z) \de z$ over $\mathbb{C}^M$. The first term on the right-hand side equals
\begin{equation}
	-\tfrac{1}{\beta} \int_{\mathbb{C}^M} \ln \tr_> \left[ e^{-\beta \mathbb{T}_s(z)} \right] \zeta_{\Gamma}(z) \de z =   \sum_{|p|<p_{\mathrm{c}}} \int_{\mathbb{C}^M} \left( p^2 - \mu_0 \right) | z_p |^2 \zeta_{\Gamma}(z) \de z + \tfrac{1}{\beta} \sum_{|p| \geq p_{\mathrm{c}}} \ln\left( 1 - e^{-\beta\left( p^2 - \mu_0 \right)} \right).
	\label{eq:aprioribounds2bb}
\end{equation}
The chemical potential $\mu_0$ is negative because the lowest eigenvalue of $-\Delta$ equals zero. From the Gibbs variational principle we know that
\begin{align}
	\sum_{|p|<p_{\mathrm{c}}} \int_{\mathbb{C}^M} \left( p^2 - \mu_0 \right) | z_p |^2 \zeta_{\Gamma}(z) \de z - \tfrac{1}{\beta} S(\zeta_{\Gamma}) &\geq -\tfrac{1}{\beta} \ln \left( \int_{\mathbb{C}^M} \exp\left( - \beta \sum_{|p|<p_{\mathrm{c}}} \left( p^2 - \mu_0 \right) | z_p |^2 \right) \de z \right) \label{eq:aprioribounds2bc} \\
	&= \tfrac{1}{\beta} \sum_{|p| < p_{\mathrm{c}}} \ln\left( \beta \left( p^2 - \mu_0 \right) \right) \geq \tfrac{1}{\beta} \sum_{|p| < p_{\mathrm{c}}}  \ln\left( 1 - e^{-\beta\left( p^2 - \mu_0 \right)} \right). \nonumber
\end{align} 
To arrive at the last line we used the inequality $x \geq 1 - e^{-x}$.
Eqs. \eqref{eq:aprioribounds1}--\eqref{eq:aprioribounds2bc} therefore imply
\begin{equation}
\int_{\mathbb{C}^M} \left\{ \tr \left[ \mathbb{H}^{\lambda} \Upsilon^z \right] - \tfrac{1}{\beta} S\left(\Upsilon^z \right)	\right\} \zeta_{\Gamma}(z) \de z - \tfrac{1}{\beta} S(\zeta_{\Gamma}) \geq \ \tfrac{1}{\beta} \sum_p \ln\left( 1 - e^{-\beta\left( p^2 - \mu_0 \right)} \right) + \tfrac{1}{\beta} \int_{\mathbb{C}^M} S\left(\Gamma_z,\Gamma^0\right) \zeta_{\Gamma}(z) \de z. 
\label{eq:aprioribounds3}
\end{equation}
Note that we have chosen $\mu_0$ such that the expected number of particles in the grand canonical system equals $N$. Concerning the lower bound, it is sufficient to consider states $\Gamma$ with free energy bounded from above by
\begin{equation}
\tr \left[ H_N^{\lambda} \Gamma \right] - \tfrac{1}{\beta} S(\Gamma) \leq F_0(\beta,N,L,\lambda) + | \Lambda | 16 \pi a_N \varrho^2 .
\label{eq:relentropapriori2}
\end{equation}
Here $F_0(\beta,N,L,\lambda)$ denotes the canonical free energy for the Hamiltonian $H_N^{\lambda}$ with $v = 0$. The actual lower bound we are going to prove will be smaller than the right-hand side of \eqref{eq:relentropapriori2}, that is, the statement will hold independently of this assumption. We use Lemma~\ref{lem:freeenergy} in the Appendix to obtain an upper bound for the canonical free energy in Eq.~\eqref{eq:relentropapriori2} in terms of the grand canonical free energy, that is, \eqref{eq:gcfreeenergy} with $\mu_0$ replaced by $\mu(\lambda)$ in the first term and $p^2 - \mu_0$ replaced by $p^2 + \delta_{p,0} \lambda - \mu(\lambda)$ in the second term. Together with \eqref{eq:lemmaentropyboundz11} and \eqref{eq:relentropapriori2}, this implies
\begin{align}
	\int_{\mathbb{C}^M} S\left(\Gamma_z,\Gamma^0\right) \zeta_{\Gamma}(z) \de z \leq& \ \sum_p \ln\left( 1 - e^{-\beta\left( p^2 - \mu(\lambda) + \delta_{p,0} \lambda \right)} \right) - \sum_p \ln\left( 1 - e^{-\beta\left( p^2 - \mu_0 \right)} \right)  \label{eq:relentropapriori3} \\
	& + \beta | \Lambda | 16 \pi \beta a_N \varrho^2 + \beta Z^{(1)} + \ln(N+1) + 1. \nonumber
\end{align}
To obtain an upper bound on the difference between the two grand canonical potentials in the first line, we write
\begin{equation}
	\sum_{p \in \frac{2 \pi}{L} \mathbb{Z}^3} \ln\left( 1 - e^{-\beta\left( p^2 - \mu(\lambda) + \delta_{p,0} \lambda \right)} \right) \leq \sum_{| p | \geq \frac{2 \pi}{L}} \ln\left( 1 - e^{-\beta\left( p^2 - \mu(\lambda) \right)} \right) \leq \sum_{| p | \geq \frac{2 \pi}{L}} \ln\left( 1 - e^{-\beta\left( p^2 - \mu_0 \right)} \right).
	\label{eq:relentropapriori3b}
\end{equation}
The second estimate follows from $\mu_0 \leq \mu(\lambda)$ which is implied by the monotonicity of the map $\lambda \mapsto \mu(\lambda)$. In combination with \eqref{eq:relentropapriori3} and 
\begin{equation}
	\ln\left( 1 - e^{\beta \mu_0 } \right) \lesssim \ln\left( \frac{1}{-\beta \mu_0} \right) \lesssim \ln(N),
\end{equation}
\eqref{eq:relentropapriori3b} 
proves
\begin{equation}
	\int_{\mathbb{C}^M} S\left(\Gamma_z,\Gamma^0\right) \zeta_{\Gamma}(z) \de z \leq | \Lambda | 16 \pi \beta a_N \varrho^2 + \beta Z^{(1)} + \text{const.} \ln(N).
	\label{eq:a-priori1}
\end{equation}
This is the a-priori bound we were looking for. To compute the interaction energy, we will use \eqref{eq:a-priori1} to replace $\Gamma_z$ by $\Gamma^0$ in a controlled way. In other words, the lower bound represents a rigorous version of first order perturbation theory. 
\begin{remark}
	The interacting free energy corresponding to $H_N^{\lambda}$ depends on $\lambda$ only through the free energy of the ideal gas to leading order. This is because the interaction energy depends, apart from $| \Lambda |$, $a_N$ and $\varrho$ (which are independent of $\lambda$), only on the expected density of the condensate $\varrho_0(\beta,N,L,\lambda)$. It can be checked, however, that $\varrho_0(\beta,N,L,\lambda)$ does not depend on $\lambda$ to leading order if $\lambda \in [0, (2 \pi/L)^2 \eta]$ with $0<\eta<1$. This justifies the use of $\Gamma^0$ in the computation of the interaction energy.
\end{remark}

We also derive a second a-priori bound. It is a simple estimate for the variance of the probability measure $\zeta_{\Gamma}(z) \de z$ which counts the number of particles in the Fock space with momenta smaller than or equal to $p_{\mathrm{c}}$ and reads
\begin{equation}
	\int_{\mathbb{C}^M} |z|^2 \zeta_{\Gamma}(z) \de z \leq N + M.
	\label{eq:relentropapriori4}
\end{equation}
To prove \eqref{eq:relentropapriori4} we use $|z|^2 \leq \tr[ \mathbb{N}_{\mathrm{s}}(z) \Gamma_z]$ and $\int_{\mathbb{C}^M} \tr[ \mathbb{N}_{\mathrm{s}}(z) \Gamma_z ] \zeta_{\Gamma}(z) \de z = N + M$.
\subsection{Replacing vacuum}
\label{sec:replacingvacuum}
In order to prove the lower bound, we have to estimate the kinetic energy and the interaction energy of states of the form $\Upsilon^z = U(z) | \text{vac} \rangle\langle \text{vac} | U(z)^* \otimes \Gamma_z$ with $U(z)$ defined in \eqref{eq:U}, and where $\Gamma_z$ obeys the a-priori bound \eqref{eq:a-priori1}. We find it necessary for this analysis to replace the vacuum in the formula for $\Upsilon^z$ by a more general quasi-free state, which we do in a controlled way in this section. This will become important below when the interaction energy of $\Upsilon^z$ is computed. For this purpose the latter will be replaced by a quasi-free state, whose one-particle density matrix should show rapid off-diagonal decay in order for the localization technique of the relative entropy to be applicable, see \cite[Sections~2.8,~2.13]{Sei2008}. Hence the momentum distribution needs to be sufficiently smooth and cannot vanish identically for the low momentum modes (as it does in the vacuum state). 

We denote by $\Pi$ a particle-number conserving quasi-free state on $\mathscr{F}_<$. It is fully determined by its 1-pdm
\begin{equation}\label{eq:1}
	\pi = \sum_{|p| < p_{\mathrm{c}}} \pi_p | p \rangle \langle p |.
\end{equation} 
Here $| p \rangle$ denotes a plane wave state in $L^2(\Lambda)$ with momentum $p$. We also define $P =  \sum_{|p| < p_{\mathrm{c}}} \pi_p = \trs \pi$. Here and in the following we denote by $\trs[\cdot]$ the trace over the one-particle Hilbert space $\mathcal{H}$. Finally, let us introduce the state
\begin{equation}
	\Upsilon_{\pi}^z = U(z) \Pi U(z)^* \otimes \Gamma_z
\end{equation}
on $\mathscr{F}$. In order to replace $\Upsilon^z$ by $\Upsilon_{\pi}^z$ in a controlled way, we have to estimate the effect of this replacement on the kinetic and the potential energy. Our analysis follows the one in \cite[Section~2.5]{Sei2008} with the only difference that we control the particle number with the measure $\zeta_{\Gamma}(z) \de z$ and not with the help of the operator $\mathbb{K}$, see Remark~\ref{rem:linktoRobertspaper1}. More concretely, we use the identity $\int_{\mathbb{C}^M} \tr_> [ \mathbb{N}_{\mathrm{s}} \Gamma_z ] \zeta_{\Gamma}(z) \de z = N + M$. When we go through the analysis in \cite[Section~2.5]{Sei2008}, we obtain
\begin{equation}
	\int_{\mathbb{C}^M} \tr \left[ \mathbb{V} \left( \Upsilon_{\pi}^z - \Upsilon^z \right) \right] \zeta_{\Gamma}(z) \de z \leq \frac{8 \pi \varphi L}{| \Lambda | N} \left( P^2 + 2 P \left[ N + M \right] \right) \equiv Z^{(2)}. \label{eq:replacingvaccum2}
\end{equation}

We also have to replace $\Upsilon^z$ by $\Upsilon_{\pi}^z$ in the kinetic energy which can be done with the identity
\begin{equation}
	\tr\left[ \mathbb{T}^{\lambda} \Upsilon^z \right] = \tr\left[ \mathbb{T}^{\lambda} \Upsilon_{\pi}^z \right] - \sum_{|p| < p_{\mathrm{c}}}  \left( p^2 + \delta_{p,0} \lambda - \mu(\lambda) \right) \pi_p.
	\label{eq:replacingvaccum3}
\end{equation}
In combination with \eqref{eq:lemmaentropyboundz11}, we obtain
\begin{align}
	\tr \left[ H_N^{\lambda} \Gamma \right] - \tfrac{1}{\beta} S(\Gamma) \geq& \ \mu(\lambda) N + \int_{\mathbb{C}^M} \left\{ \tr \left[ \left( \mathbb{T}^{\lambda} + \mathbb{V} \right) \Upsilon_{\pi}^z \right] - \tfrac{1}{\beta} S(\Upsilon^z)  \right\} \zeta_{\Gamma}(z) \de z - \tfrac{1}{\beta} S(\zeta_{\Gamma})
	\label{eq:replacingvaccum4} \\
	& - \sum_{|p| < p_c} \left( p^2 + \delta_{p,0} \lambda - \mu(\lambda) \right) \pi_p - Z^{(1)} - Z^{(2)} \nonumber
\end{align}
as a lower bound for the free energy of $\Gamma$. 
\subsection{Dyson Lemma and Filling the Holes}
The sections 2.6 (Dyson Lemma) and 2.7 (Filling the Holes) in \cite{Sei2008} remain basically unchanged. To introduce several quantities that are needed later and to mention the necessary changes due to the term $\lambda a_0^* a_0$ in the Hamiltonian, we collect the main result here. The Dyson Lemma \cite[Lemma~2]{Sei2008} is used to replace the singular and short ranged potential $\widetilde{v}_N$ by a softer potential with a longer range at the expense of a certain amount of kinetic energy. To be precise, only the high momentum modes are used for the Dyson Lemma. This is necessary because the low momentum modes are used to obtain the free energy of the ideal Bose gas. The Dyson Lemma naturally leads to an effective interaction potential with a hole around zero. Because it will be necessary for the computation of the interaction energy, this potential is replaced by a slightly different one without a hole.

By $R$ we denote the length scale of the effective potential from the Dyson Lemma satisfying $10 R_0 L /N < R < L/2$. When this potential is replaced by a potential without a hole in the middle, one obtains a potential with a slightly reduced scattering length
\begin{equation}
	a'_N = \widetilde{a}_N (1- \epsilon)(1-\kappa) \left( 1 - \frac{18}{(\pi/4)^3}(4-\pi) \frac{\left( R_0 L/N \right)^3}{(R/10)^3 - \left( R_0 L/N \right)^3} \frac{1}{j(1/10)} \right),
	\label{eq:Dysonholes1a}
\end{equation}
with two parameters $0 < \kappa <1$ and $\epsilon > 0$ that are related to the Dyson Lemma, see \cite[Sec.~2.6]{Sei2008}. The definition of the function $j$ is given by
\begin{equation}
	j(t) = 12 (t+2) [1-t]_+^2.
	\label{eq:Dysonholes1b}
\end{equation}
We apply the Dyson Lemma and the analysis to replace the relevant potential by one without a hole in the same way as in \cite[Sec.~2.7]{Sei2008} and compute the term $\lambda a_0^* a_0$ separately. The final result of this analysis reads
\begin{align}
	\tr \left[ \left( \mathbb{H}^{\lambda} + \mathbb{V} \right) \Upsilon_{\pi}^z \right] - \tfrac{1}{\beta} S(\Upsilon^z) \geq& - \tfrac{1}{\beta} \ln \tr_> e^{-\beta \mathbb{T}_{\mathrm{s}}^{\mathrm{c}}(z)} + \tr\left[ \mathbb{W} \Upsilon_{\pi}^z \right] + \tfrac{1}{\beta} S\left( \Gamma_z, \Gamma^{0,\lambda}_{\mathrm{c}} \right)  \label{eq:Dysonholes1} \\
	&+ \sum_{| p| < p_{\mathrm{c}}} \left[ (1 - \kappa + \kappa') p^2 + \delta_{p,0} \lambda - \mu(\lambda) \right] \pi_p. \nonumber
\end{align}
In the above equation $\mathbb{T}_{\mathrm{s}}^{\mathrm{c}}(z)$ denotes the lower symbol of the operator
\begin{equation}
	\mathbb{T}^{\mathrm{c}} = \sum_p \epsilon(p) a_p^* a_p, \quad \text{ where } \quad \epsilon(p) = \delta_{p,0} \lambda + \kappa' p^2 + (1-\kappa) p^2 \left(1- \nu(sp)^2 \right) - \mu(\lambda)
	\label{eq:Dysonholes2}
\end{equation}
and 
\begin{equation}
	\kappa' = \kappa - \frac{24 \widetilde{a}_N  }{\pi^2} \frac{(4 R_0 L/N)^2}{R^3}. 
	\label{eq:Dysonholes2b}
\end{equation}
We will later choose $\kappa$ and $R$ such that $\widetilde{a}_N (R_0 L/N)^2/R^3 \ll \kappa$. This in particular implies $\kappa' > 0$. The function $\nu : \mathbb{R}^3 \mapsto \mathbb{R}_+$ is chosen such that $\nu(p) = 0$ for $|p| \leq 1$, $\nu(p) = 1$ for $| p | \geq 2$, and $0 \leq \nu(p) \leq 1$ in-between. It is used to implement the fact that only the high momentum modes are used in the Dyson Lemma. The parameter $s$ obeys $s \geq R$ and will later be chosen such that $s \gg R$. We will also choose $\kappa \ll 1$. In combination with $\lambda \leq (2\pi/L)^2 \eta$ with $0 < \eta < 1$, this implies $\epsilon(p) > 0$ for all $p$. The effective interaction potential $\mathbb{W}$ will not be specified here because we will use the same estimate for $\tr\left[ \mathbb{W} \Upsilon_{\pi}^z \right]$ as in \cite{Sei2008}. Its definition can be found in \cite[Sec.~2.7]{Sei2008}. Note that, compared to \cite[(2.7.15)]{Sei2008}, we have the additional term $ \beta^{-1} S( \Gamma_z, \Gamma^{0,\lambda}_{\mathrm{c}} ) $ in our lower bound \eqref{eq:Dysonholes1}. Here $\Gamma^{0,\lambda}_{\mathrm{c}}$ denotes the grand canonical Gibbs state for the kinetic energy operator $\mathbb{T}^{\mathrm{c}}_{\mathrm{s}}(z)$ which is independent of $z$ and depends on $\lambda$ only through the chemical potential $\mu(\lambda)$. The additional term is not important for the lower bound (it is positive and could be dropped), but it will be important for the proof of the asymptotics of the 1-pdm of approximate minimizers of the Gibbs free energy functional in Sec.~\ref{sec:asymptotics1pdm}.  When we insert the above result into \eqref{eq:replacingvaccum4} and argue as in \eqref{eq:aprioribounds2bb} and \eqref{eq:aprioribounds2bc}, we find
\begin{align}
	\tr \left[ H_N^{\lambda} \Gamma \right] - \tfrac{1}{\beta} S(\Gamma) \geq & \ \mu(\lambda) N + \tfrac{1}{\beta} \sum_p \ln\left( 1 - e^{-\beta \epsilon(p)} \right) + \int_{\mathbb{C}^M} \tr \left[ \mathbb{W} \Upsilon_{\pi}^z \right] \zeta_{\Gamma}(z) \de z 
	\label{eq:Dysonholes3} \\
	&+ \tfrac{1}{\beta} \int_{\mathbb{C}^M} S\left( \Gamma_z, \Gamma^{0,\lambda}_{\mathrm{c}} \right) \zeta_{\Gamma}(z) \de z  - (\kappa - \kappa') \sum_{|p| < p_c} p^2 \pi_p - Z^{(1)} - Z^{(2)}. \nonumber
\end{align}
From the first two terms on the right-hand side of \eqref{eq:Dysonholes3}, we will obtain the free energy of the ideal gas. 
\subsection{Localization of Relative Entropy}
\label{sec:localizationofrelativeentropy}
In this section we introduce notation that will be important for the following. The main result from the related section in \cite{Sei2008} will not be stated since we will not explicitly need it. It is used only in parts of the proof in \cite{Sei2008} that we do not have to adjust.

Define the quasi free state $\Omega_{\pi} = \Pi \otimes \Gamma^0$ via its 1-pdm
\begin{equation}
	\omega_{\pi} = \sum_p \frac{1}{e^{\ell(p)}-1} | p \rangle\langle p |
	\label{eq:locrelentropy1}
\end{equation}
with $\ell(p) = \beta(p^2 - \mu_0)$ for $| p | \geq p_{\mathrm{c}}$ and $\ell(p) = \ln(1+\pi_p^{-1})$ for $|p| < p_{\mathrm{c}}$. We recall that the quasi-free state $\Pi$ with 1-pdm $\pi$ has been defined in \eqref{eq:1}. See also the beginning of Section~\ref{sec:replacingvacuum} for the reason why it is introduced. Let $\eta \in \mathcal{C}^{\infty} ([0,\infty), \mathbb{R})$ with $\eta(0) = 1$, $\eta(r) = 0$ for $r \geq 1$ and $\hat{\eta}(p) = \int \eta(|x|) e^{ipx} \de x \geq 0$. For $0 < b \leq L/2$ let $\eta_b(r) = \eta(r/b)$. We also define the quasi-free state $\Omega_b$ via its 1-pdm $\omega_b$ which is given by
\begin{equation}
	\omega_b(x,y) = \omega_{\pi}(x,y) \eta_b(d(x,y)).
	\label{eq:locrelentropy2}
\end{equation}
The densities of the states $\Omega_{\pi}$ and $\Omega_b$ then fulfill
\begin{equation}
\varrho_{\omega} \coloneqq \omega_{\pi}(x,x) = \omega_b(x,x) 
\label{eq:locrelentropy3}
\end{equation}
because $\eta_b(0) = 1$. Finally, let 
\begin{equation}\label{eq:2}
\Omega_b^z = U(z) \Omega_b U(z)^*.
\end{equation}

\subsection{Interaction Energy Part 1 - 3}
The expectation of the effective potential $\mathbb{W}$ in the state $\Upsilon_{\pi}^z$ is estimated as in \cite[Secs.~2.9--2.11]{Sei2008}. The result of the analysis in these sections is the following lower bound:
\begin{align}
\int_{\mathbb{C}^M} \tr\left[ \mathbb{W} \Upsilon_{\pi}^z \right] \zeta_{\Gamma}(z) \de z \geq& \ 4 \pi a_N' | \Lambda | \int_{\mathbb{C}^M} \min \left\{ 2 \varrho^2, \varrho_z^2 + 2 \varrho_z \left( \gamma_b + \varrho_{\omega} \right) + \gamma_b^2 + \varrho_{\omega}^2 \right\} \zeta_{\Gamma}(z) \de z \label{eq:interaction1} \\
&+ \frac{24 \widetilde{a}_N}{R^3} \left\{ \frac{a_N'}{125 \widetilde{a}_N} - \frac{6^3 c R}{s} \right\} \int_{\mathbb{C}^M} \tr \left[ \mathbb{N}\left( \Upsilon_{\pi}^z - \Omega_b^z \right) \right] \zeta_{\Gamma}(z) \de z  \nonumber \\
&-\text{const. } \frac{a_N}{R^6} \left( b^3 | \Lambda | \right)^{1/2} \left( \int_{\mathbb{C}^M} S(\Upsilon_{\pi}^z,\Omega_b^z) \zeta_{\Gamma}(z) \de z \right)^{1/2} \nonumber \\
&- \text{const. } \frac{| \Lambda | a_N \varrho}{R^2 s} \int_{b/s}^{\infty} r^6 |m(r)| \de r - \text{const. } a_N | \Lambda | \varrho^2 \left[ \left(R^3 \varrho \right)^{1/3} + \frac{R}{s} + R p_{\mathrm{c}} \right]. \nonumber
\end{align}
The scattering length $a_N'$ has been defined in \eqref{eq:Dysonholes1a}, $m(r)$ is an explicitly given smooth function that vanishes faster than any power for $|r| \to \infty$, compare with \cite[Sec.~2.10]{Sei2008}, and $c > 0$. Additionally,
\begin{equation}
	\gamma_b = \frac{1}{4 \pi R^3} \int_{\Lambda} \omega_b(x,0) j(d(x,0)/R) \de x,
	\label{eq:interaction2}
\end{equation}
where $j$ is defined in \eqref{eq:Dysonholes1a}. To obtain the result, we started with Eqs.~\cite[(2.11.19--21)]{Sei2008} and the same choice of the parameters $\epsilon$ and $D$ as in \cite{Sei2008}. We also applied Jensen's inequality to the term proportional to the relative entropy and used \eqref{eq:relentropapriori4} to bound $\int_{\mathbb{C}^M} \varrho_z \ \zeta_{\Gamma}(z) \de z \leq | \Lambda|^{-1} (N+M)$ where $\varrho_z = |z|^2/| \Lambda |$ as well as $a_N' \leq \widetilde{a}_N \leq a_N$. We assumed that $p_{\mathrm{c}} \lesssim \varrho^{1/3}$, $\varrho_{\omega} \lesssim \varrho$ and $R \lesssim s$. The bound is valid for any choice of the parameter $0< b \leq L/2$ that has been introduced in the previous section. We will later choose $b$ such that $b p_{\mathrm{c}} \gg 1$ and $\beta b^{-2} \ll 1$.

\subsection{A bound on the number of particles}
The lower bound on the interaction energy contains a term of the form
\begin{equation}
	\frac{24 \widetilde{a}_N}{R^3} \left\{ \frac{a_N'}{125 \widetilde{a}_N} - \frac{6^3 c R}{s}  \right\} \int_{\mathbb{C}^M} \tr \left[ \mathbb{N}\left( \Upsilon_{\pi}^z - \Omega_b^z \right) \right]  \zeta_{\Gamma}(z) \de z.
	\label{eq:boundnumberofparticles1}
\end{equation}
Recall that we will later choose $R \ll s$, that is, the term $\tr \left[ \mathbb{N}\left( \Upsilon_{\pi}^z - \Omega_b^z \right) \right]$ is multiplied by a positive constant. In this section we will first rewrite the integral over the trace on the right-hand side of \eqref{eq:boundnumberofparticles1}. This way it will be apparent that the term in \eqref{eq:boundnumberofparticles1} can be combined with another error term that we will find in Sec.~\ref{sec:finallowerbound}. This term will be of the same form but it will be multiplied by a negative constant that is much smaller than the one in the equation above. Accordingly, we only have to derive a lower bound for the integral on the right-hand side of \eqref{eq:boundnumberofparticles1} to finally estimate the sum of these two terms.

Let us start by rewriting the integral on the right-hand side of \eqref{eq:boundnumberofparticles1}. We note that $\tr [\mathbb{N} \Omega_b^z] = |z|^2 + \tr [\mathbb{N} \Omega_b] = |z|^2 + \tr [\mathbb{N} \Omega_{\pi}]$ as well as $\tr[\mathbb{N} \Upsilon_{\pi}^z] = |z|^2 + \tr[\mathbb{N} \Upsilon_{\pi}]$ and we denote by $\mathbb{N}^> = \sum_{|p| \geq p_{\mathrm{c}}} a_p^* a_p$ the particle number operator on the Fock space $\mathscr{F}_>$. Since $\Upsilon_{\pi} = \Pi \otimes \Gamma_z$ and $\Omega_{\pi} = \Pi \otimes \Gamma^0$ we can write
\begin{equation}
	\tr \left[ \mathbb{N}\left( \Upsilon_{\pi}^z - \Omega_b^z \right) \right] = \tr_> \left[ \mathbb{N}^> \left( \Gamma_z - \Gamma^0 \right) \right].
	\label{eq:boundnumberofparticles1b1}
\end{equation}
We also know that
\begin{equation}
	N+M = \int_{\mathbb{C}^M} \tr_> \left[ \mathbb{N}_s(z) \Gamma_z \right] \zeta_{\Gamma}(z) \de z = \int_{\mathbb{C}^M} |z|^2 \zeta_{\Gamma}(z) \de z + \int_{\mathbb{C}^M} \tr_> \left[ \mathbb{N}^> \Gamma_z \right] \zeta_{\Gamma}(z) \de z.
	\label{eq:boundnumberofparticles1b2}
\end{equation}
In combination, Eqs.~\eqref{eq:boundnumberofparticles1b1} and \eqref{eq:boundnumberofparticles1b2} imply
\begin{equation}
		\int_{\mathbb{C}^M} \tr \left[ \mathbb{N}\left( \Upsilon_{\pi}^z - \Omega_b^z \right) \right] \zeta_{\Gamma}(z) \de z = N + M -  \tr_> \left[ \mathbb{N}^> \Gamma^0 \right] - \int_{\mathbb{C}^M} | z |^2 \zeta_{\Gamma}(z) \de z. \label{eq:boundnumberofparticles1b3c}
\end{equation}
This is the first result we were looking for.

Next, we will derive a lower bound for the right-hand side of \eqref{eq:boundnumberofparticles1b1}. It implies a lower bound on the right-hand side of \eqref{eq:boundnumberofparticles1b3c} that will later allow us to estimate the relevant error term in Sec.~\ref{sec:finallowerbound}.
A bound of this kind was proved in \cite[Sec.~2.12]{Sei2008} in the case of the dilute Bose gas in the thermodynamic limit. In this limit momentum space sums can be replaced by integrals because the relevant errors do not grow proportionally to the volume and are therefore irrelevant. In the GP limit that we consider these error terms have to be quantified, however. To that end, we have to adjust the estimates in Eqs.~(2.12.9)--(2.12.12) in \cite{Sei2008}, which will be done with the help of the following Lemma.
\begin{lemma}
	\label{lem:latticesum1}
	Let $f : [0,\infty) \to \mathbb{R}$ be a nonnegative and monotone decreasing function and choose some $\kappa \geq 0$. Then
	\begin{equation}
		\sum_{p \in \frac{2 \pi}{L} \mathbb{Z}^3 \backslash \{ 0 \}} f\left( |p| \right) \mathds{1}(|p| \geq \kappa) \leq  \left( \frac{L}{2 \pi } \right)^3 \int_{  | p | \geq \left[ \kappa -  \sqrt{3} \frac{2 \pi}{L} \right]_+ } f\left( |p| \right) \left( 1 + \frac{3 \pi}{L | p |} + \frac{6 \pi}{L^2 p^2} \right) \de p \label{eq:boundnumberofparticles1d}
	\end{equation}
	holds.
\end{lemma}
\begin{proof}
	Assume first that $\kappa - \sqrt{3} \frac{2 \pi}{L} \leq 0$. In this case we drop the characteristic function on the left-hand side of \eqref{eq:boundnumberofparticles1d} for an upper bound. Next, we write the sum over $p$ as the sum over those $p$ that are an element of one of the coordinate planes, i.e. with one coordinate $p_i$ equal to zero, plus a sum over all remaining $p$. To estimate the sum over all remaining $p$, we interpret the sum as a lower Riemann sum and find that it is bounded from above by
	\begin{equation}
		\left( \frac{L}{2 \pi} \right)^3 \int_{ \mathbb{R}^3} f\left( |p| \right) \de p. \label{eq:boundnumberofparticles1e}
	\end{equation}
	The sum over those $p$ that are an element of the coordinate planes can be estimated similarly. Here we write the whole sum as a sum over those $p$ that are an element of one of the coordinate axes plus the sum over all remaining $p$. The sum over the remaining $p$ is estimated again by interpreting it as a lower Riemann sum. For one such coordinate plane, we find
	\begin{equation}
		\sum_{ p \in \frac{2 \pi}{L} \mathbb{Z}^2 \cap \{ p_i \neq 0 \text{ for } i=1,2 \} } f\left( |p| \right) \leq \left( \frac{L}{2 \pi} \right)^2  \int_{ \mathbb{R}^2} f\left( |p| \right) \de p.
		\label{eq:boundnumberofparticles1f}
	\end{equation}
	Because there are three coordinate planes, we have three such terms. It remains to estimate the sum over those $p$ that are an element of one of the coordinate axes of $\mathbb{R}^3$. Again by interpreting the sums over the three coordinate axis as lower Riemann sums, we find that they are bounded from above by
	\begin{equation}
		6 \frac{L}{2 \pi} \int_{0}^{\infty} f\left( |p| \right) \de p.
		\label{eq:boundnumberofparticles1g}
	\end{equation}
	In order to write the two-dimensional integral from \eqref{eq:boundnumberofparticles1f} in terms of a three-dimensional integral, we use
	\begin{align}
		\left( \frac{L}{2 \pi} \right)^2 \int_{\mathbb{R}^2} f\left( |p| \right) \de p &= \left( \frac{L}{2 \pi} \right)^3 \int_{\mathbb{R}^3} f\left( |p| \right) \frac{\pi}{| p | L}  \de p. \label{eq:boundnumberofparticles1h} 
	\end{align}
	A similar computation can be done for the term in Eq.~\eqref{eq:boundnumberofparticles1g}. Putting these estimates together proves \eqref{eq:boundnumberofparticles1d} in this case. The bound in the case $\kappa - \sqrt{3} \frac{2 \pi}{L} > 0$ can be obtained similarly. Here we only have to realize that $\kappa - \sqrt{3} \frac{2 \pi}{L}$ is the radius of the largest ball such that the integral over its complement is an upper bound to the relevant three-dimensional lower Riemann sum. This proves the claim.
\end{proof}
To adjust the analysis in \cite{Sei2008} after (2.12.8), we have to find an upper bound for the sum
\begin{equation}
	\sum_{p \in \frac{2 \pi}{L} \mathbb{Z}^3} \mathds{2}\left( | p | \geq p_{\mathrm{c}} \right) \frac{2}{\left( \beta p^2 - \beta \mu_0 \right)^2}. 
	\label{eq:boundnumberofparticles1i}
\end{equation}
An application of Lemma~\ref{lem:latticesum1} tells us that it is bounded from above by
\begin{equation}
	 2 \frac{\mathds{1}\left( p_{\mathrm{c}} = 0 \right)}{\left( \beta \mu_0 \right)^2} + \left( \frac{L}{2 \pi} \right)^3 \int_{  | p | \geq \left[ p_{\mathrm{c}} -  \sqrt{3} \frac{2 \pi}{L} \right]_+ } \frac{2}{\left( \beta p^2 - \beta \mu_0 \right)^2} \left( 1 + \frac{3 \pi}{L p} + \frac{6 \pi}{L^2 p^2} \right)  \de p.
	\label{eq:boundnumberofparticles1j}
\end{equation}
A short computation shows that the expression in the above equation cannot be larger than 
\begin{equation}
	A \coloneqq 2 \frac{\mathds{1}\left( p_{\mathrm{c}} = 0 \right)}{\left( \beta \mu_0 \right)^2} + \text{const. } \frac{L^3}{\beta^{3/2}} \left( \frac{1}{\sqrt{\beta \left[ p_{\mathrm{c}} -  \sqrt{3} \frac{2 \pi}{L} \right]_+^2 - \beta \mu_0 } } + \frac{\beta}{L^2} \left( \frac{1}{\beta \left[ p_{\mathrm{c}} -  \sqrt{3} \frac{2 \pi}{L} \right]_+^2 - \beta \mu_0 } \right)^{3/2} \right).
	\label{eq:boundnumberofparticles1k}
\end{equation}
This bound replaces (2.12.12) in \cite{Sei2008}. Using the above and (2.12.8) in \cite{Sei2008}, we conclude that
\begin{equation}
	\tr_{>}\left[ \mathbb{N}^> \left( \Gamma_z - \Gamma^0 \right) \right] \geq - A^{1/2} S \left( \Gamma_z, \Gamma^0 \right)^{1/2}. 
	\label{eq:boundnumberofparticles1l}
\end{equation}
holds. 

If we combine \eqref{eq:boundnumberofparticles1l} with \eqref{eq:boundnumberofparticles1b1} and \eqref{eq:boundnumberofparticles1b3c} we obtain the bound 
\begin{equation}
	N - \tr\left[ \mathbb{N}^>  \Gamma^0 \right] - \int_{\mathbb{C}^M} | z |^2 \zeta_{\Gamma}(z) \de z  \geq -M -  \left( 2 A \right)^{1/2} \int_{\mathbb{C}^M} S \left( \Gamma_z, \Gamma^0 \right)^{1/2} \zeta_{\Gamma}(z) \de z. \label{eq:boundnumberofparticles1b3}
\end{equation}
An application of Jensen's inequality and the a-priori bound \eqref{eq:a-priori1} therefore imply
\begin{equation}
	N - \tr\left[ \mathbb{N}^>  \Gamma^0 \right] - \int_{\mathbb{C}^M} | z |^2 \zeta_{\Gamma}(z) \de z \geq -M - \text{const. } \left( | \Lambda | \beta a_N \varrho^2 + \beta Z^{(1)} + \ln(N) \right)^{1/2} A^{1/2}. \label{eq:boundnumberofparticles1b4}
\end{equation}
This is the bound we were looking for. It will later be used to bound the relevant error term in \eqref{eq:interaction1}.
\subsection{Relative Entropy, Effect of Cutoff}
\label{sec:relentropycutoff}
In this section we estimate the relative entropy $S( \Upsilon_{\pi}^z,\Omega_b^{z}) = S(\Pi \otimes \Gamma_z, \Omega_b)$, which appears in the lower bound \eqref{eq:interaction1} for the interaction energy, in terms of $S(\Pi \otimes \Gamma_z,\Omega_{\pi}) = S(\Gamma_z,\Gamma^0)$. Since we have an a-priori bound for the integral w.r.t. $\zeta_{\Gamma}(z) \de z$ over the latter expression at hand this will allow us to finalize the lower bound for the interaction energy. Compared to \cite{Sei2008}, we have to adjust how the momentum space sum related to \cite[(2.13.21)]{Sei2008} is estimated. 

We are faced with estimating the term
\begin{equation}
	2 \sum_{p \in \frac{2 \pi}{L} \mathbb{Z}^3} \left( 1 + \beta \left( | p | + | q | \right)^2 \right)^2 \omega^t(p) \left( 1 + 2 \omega^t(p) \right).
	\label{eq:latticesum24b}
\end{equation}
Here 
\begin{equation}
	\omega^t(p) = \frac{1}{e^{\beta \left( p^2 - \mu_0 \right) - D \beta t q^2} - 1}
\end{equation}
with some parameters $D$, $t$ and $q$ that are specified in \cite{Sei2008} and that are chosen such that $\beta ( p^2 - \mu_0 ) - D \beta t q^2 > 0$ for all $p \in \frac{2 \pi}{L} \mathbb{Z}^3$. When we insert the estimate \cite[(2.13.24)]{Sei2008} for $\omega^t$ into \eqref{eq:latticesum24b}, we see that it is bounded from above by a constant times
\begin{equation}
\left( 1 + \left( \beta q^2 \right)^2 \right) e^{- \tau} \sum_{p \in \frac{2 \pi}{L} \mathbb{Z}^3} e^{- \frac{1}{2} \beta p^2} \left( 1 + \left( \beta p^2 \right)^2 \right) \left( 1 + \left( \frac{1}{\tau + \frac{1}{2} \beta p^2} \right)^2 \right),
\label{eq:latticesum26}
\end{equation}
with
\begin{equation}
	\tau = -\beta \mu_0 + \beta p_{\mathrm{c}}^2 \left( \frac{1}{8} - \frac{D}{b^2 p_{\mathrm{c}}^2} \right).
	\label{eq:latticesum25b}
\end{equation}
We will later choose $p_{\mathrm{c}}$ and $b$ such that $b p_{\mathrm{c}} \gg 1$, and hence $\tau > 0$. Next, we bound the summands in \eqref{eq:latticesum26} from above by a monotone function with the same behavior at zero and at infinity. Afterwards we use 
Lemma~\ref{lem:latticesum1} to see
that \eqref{eq:latticesum26} is bounded from above by a constant times
\begin{equation}
	e^{- \tau} \left( 1 + \left( \beta q^2 \right)^2 \right) \left( 1 + \tau^{-2} + \frac{| \Lambda | \left( 1+ \tau^{-1/2} + \beta L^{-2} \tau^{-3/2} \right)}{\beta^{3/2}}   \right).
	\label{eq:latticesum27}
\end{equation}
This is the estimate for the term in \eqref{eq:latticesum24b} we intended to show. The remaining part of the analysis in \cite{Sei2008} can be done similarly. With the a-priori bound \eqref{eq:a-priori1} and the estimate $\widetilde{a}_N \leq a_N$, we finally arrive at
\begin{equation}
\int_{\mathbb{C}^M} S\left( \Upsilon_{\pi}^z, \Omega_{b}^z \right) \zeta_{\Gamma}(z) \de z \lesssim \beta Z^{(1)} + \ln(N) + \frac{\beta^2}{\tau^2 b^4} + | \Lambda | \left( a_N \beta \varrho^2 + \frac{\beta^{1/2} \left( \tau^{-1/2} + \beta L^{-2} \tau^{-3/2} \right)}{b^4} \right).
\label{eq:latticesum28}
\end{equation}
To obtain the result, we used that $\beta b^{-2}$ is small enough and that $b p_{\mathrm{c}}$ is large enough.
\subsection{Final lower bound}
\label{sec:finallowerbound}
We have obtained all necessary estimates to complete the lower bound for the free energy of $\Gamma$. To that end, we collect the estimates from the previous sections, that is, \eqref{eq:Dysonholes3}, \eqref{eq:interaction1}, \eqref{eq:boundnumberofparticles1b3c} and \eqref{eq:latticesum28} and find
\begin{align}
	\tr \left[ H_N^{\lambda} \Gamma \right] - \tfrac{1}{\beta} S(\Gamma) \geq& \ \mu(\lambda) N + \tfrac{1}{\beta} \sum_p \ln\left( 1 - e^{-\beta \epsilon(p)} \right) - Z^{(1)} - Z^{(2)} - Z^{(3)} \label{eq:finallowerbound1} \\
	& + 4 \pi a_N | \Lambda | \int_{\mathbb{C}^M} \min \left\{ 2 \varrho^2, \varrho_z^2 + 2 \varrho_z \left( \gamma_b + \varrho_{\omega} \right) + \varrho_{\omega}^2 + \gamma_b^2 \right\} \zeta_{\Gamma}(z) \de z \nonumber \\
	&+ \text{const. } \frac{\widetilde{a}_N}{R^3} \left( N -  \tr_> \left[ \mathbb{N}^> \Gamma^0 \right] - \int_{\mathbb{C}^M} | z |^2 \zeta_{\Gamma}(z) \de z \right) + \tfrac{1}{\beta} \int_{\mathbb{C}^M} S\left( \Gamma_z, \Gamma^{0,\lambda}_{\mathrm{c}} \right) \zeta_{\Gamma}(z) \de z , \nonumber
\end{align}
with
\begin{align}
	Z^{(3)} \coloneqq& \ \text{const. }  a_N | \Lambda | \left[ \varrho^2 \left( \left( R^3 \varrho \right)^{1/3} + \frac{R}{s} + R p_{\mathrm{c}} + \kappa + \left( \frac{R_0}{R} \right)^3 + \sqrt{\frac{a_N N}{\varphi L}} \right) + \frac{\varrho}{R^2 s} \int_{b/s}^{\infty} r^6 | m(r) | \de r \phantom{ \left( \frac{ \left( \tau^{-3/2} \right) }{b} \right)^{1/2} } \right. \label{eq:finallowerbound2} \\
	&\hspace{3.38cm}\left. \hspace{3cm} + \frac{1}{R^6} \left( b^3 a_N \beta \varrho^2 + \frac{\beta^{1/2} \left( \tau^{-1/2} + \beta L^{-2} \tau^{-3/2} \right) }{b} \right)^{1/2} \right] \nonumber \\
	&\hspace{2cm}+ \text{const. } \frac{a_N}{R^6} \left( b^3 | \Lambda | \right)^{1/2} \left( \beta Z^{(1)} + \ln(N) +  \frac{\beta^2}{\tau^2 b^4} \right)^{1/2} + \text{const. } \left( \frac{a_N}{R} \right)^3 p_{\mathrm{c}}^3 | \Lambda | / \beta. \nonumber
\end{align}
To obtain this result, we used the definition \eqref{eq:Dysonholes1a} of $a_N'$ and $\gamma_b \leq \varrho_{\omega} \leq \varrho$. The first part of the inequality for $\varrho_{\omega}$ follows from the definition of $\gamma_b$ \eqref{eq:interaction2} and the second part from the choice of $\pi_p$ in \cite{Sei2008} after (2.13.15). Using the definition of $\pi_p$ again, we estimated
\begin{equation}
	(\kappa - \kappa') \sum_{|p| < p_{\mathrm{c}}} p^2 \pi_p \leq (\kappa - \kappa') p_{\mathrm{c}}^2 P \lesssim \left( \frac{a_N}{R} \right)^3 p_{\mathrm{c}}^3 | \Lambda | / \beta
	\label{eq:finallowerbound3}
\end{equation} 
as in \cite[Sec.~2.14]{Sei2008}. To replace $\widetilde{a}_N$ by $a_N$ in the term in the second line in \eqref{eq:finallowerbound1} we applied Lemma~\ref{lem:integrablepotential} with the choice $\epsilon = \sqrt{a_N N/(L \varphi)}$. We will later choose $\varphi L/N \gg a_N$. The error terms $Z^{(1)}$ and $Z^{(2)}$ are defined in \eqref{eq:lemmaentropyboundz8} and \eqref{eq:replacingvaccum2}, respectively.

To obtain a bound for the interaction term in \eqref{eq:finallowerbound1}, we write
\begin{align}
  \varrho_z^2 + 2 \varrho_z \left( \gamma_b + \varrho_{\omega} \right) + \varrho_{\omega}^2 + \gamma_b^2 =&	
  \left( \varrho - \varrho^0 \right)^2 + 2 \left( \varrho - \varrho^0 \right) \left( \varrho_{\omega} + \gamma_b \right) + \varrho_{\omega}^2 + \gamma_b^2 + \left( \varrho_z - \varrho + \varrho^0 \right)^2  \label{eq:finallowerbound4b} \\
  &+ 2 \left( \varrho_z - \varrho + \varrho^0 \right) \left( \varrho - \varrho^0 + \varrho_{\omega} + \gamma_b \right), \nonumber
\end{align}
where
\begin{equation}
	\varrho^0 = | \Lambda |^{-1} \tr_> \left[ \mathbb{N}^> \Gamma^0 \right] = | \Lambda |^{-1} \sum_{|p| \geq p_{\mathrm{c}}} \frac{1}{e^{\beta \left( p^2 - \mu_0 \right)} - 1}.
	\label{eq:finallowerbound5}
\end{equation}
The last term in the first line of \eqref{eq:finallowerbound4b} can be dropped for a lower bound. Next we combine the first term in the third line of \eqref{eq:finallowerbound1} with $4 \pi a_N | \Lambda |$ times the term in the second line of \eqref{eq:finallowerbound5} integrated with $\zeta_{\Gamma}(z) \de z$ over $\mathbb{C}^M$. Together, they read
\begin{equation}
	\int_{\mathbb{C}^M} \left( \varrho_z - \varrho + \varrho^0 \right) \left[ 8 \pi a_N | \Lambda | \left( \varrho - \varrho^0 + \varrho_{\omega} + \gamma_b \right) - \text{const. } \frac{| \Lambda | \widetilde{a}_N}{R^3} \right] \zeta_{\Gamma}(z) \de z.
	\label{eq:finallowerbound51}
\end{equation}
We will later choose $R \ll \varrho^{-1/3}$, that is, the term in the second bracket in \eqref{eq:finallowerbound51} is negative and we need a lower bound for 
\begin{equation}
	\int_{\mathbb{C}^M} \left( \varrho - \varrho^0 - \varrho_z \right) \zeta_{\Gamma}(z) \de z = | \Lambda |^{-1} \left( N - \tr_> \left[ \mathbb{N}^> \Gamma^0 \right] - \int_{\mathbb{C}^M} |z|^2 \zeta_{\Gamma}(z) \de z \right).
\end{equation}
Such a bound is provided by \eqref{eq:boundnumberofparticles1b4}. In combination, the results of this paragraph imply
\begin{align}
	&\int_{\mathbb{C}^M} \left[ 4 \pi a_N | \Lambda | \min\left\{ 2 \varrho^2, \left( \varrho_z^2 + 2 \varrho_z \left( \gamma_b + \varrho_{\omega} \right) + \varrho_{\omega}^2 + \gamma_b^2 \right) \right\} + \text{const. } \frac{\widetilde{a}_N}{R^3} \left( N -  \tr_> \left[ \mathbb{N}^> \Gamma^0 \right] -  | z |^2 \right) \right] \zeta_{\Gamma}(z) \de z \nonumber \\
	&\hspace{2cm} \geq 4 \pi a_N | \Lambda | \min\left\{ 2 \varrho^2, \left[ \left( \varrho - \varrho^0 \right)^2 + 2 \left( \varrho - \varrho^0 \right) \left( \varrho_{\omega} + \gamma_b \right) + \varrho_{\omega}^2 + \gamma_b^2 \right] \right\}  \label{eq:finallowerbound5b} \\
	& \hspace{2.4cm} - \text{const. } \frac{a_N}{R^3} \left[ | \Lambda | p_{\mathrm{c}}^3 + \left( | \Lambda | \beta a_N \varrho^2 + \beta Z^{(1)} + \ln(N) \right)^{1/2} A^{1/2} \right] \nonumber
\end{align} 
with $A$ defined in \eqref{eq:boundnumberofparticles1k}. To obtain the result, we also used $M \lesssim | \Lambda | p_{\mathrm{c}}^3$. 

In the following, we assume $p_{\mathrm{c}} \neq 0$. The case where $p_{\mathrm{c}} = 0$ will follow easily from the analysis of this case. Using the definition for $\varrho_0^{\mathrm{gc}}(\beta,N,L)$ in Sec.~\ref{sec:supplementary}, we see that $\varrho - \varrho^0 \geq \varrho_0^{\mathrm{gc}}$ which implies that we obtain a lower bound for the terms in the second line in  \eqref{eq:finallowerbound5b} when we replace $\varrho - \varrho^0$ by $\varrho_0^{\mathrm{gc}}$. In order to derive a lower bound for $\varrho_{\omega}$, we estimate
\begin{align}
	\varrho_{\omega} &\geq | \Lambda |^{-1} \sum_{| p | \geq p_{\mathrm{c}}} \frac{1}{e^{\beta \left( p^2 - \mu_0 \right)} - 1} = \varrho_{\mathrm{th}}^{\mathrm{gc}} - | \Lambda |^{-1} \sum_{0 < | p |  < p_\mathrm{c}} \frac{1}{e^{\beta \left( p^2 - \mu_0 \right) } - 1} \label{eq:finallowerbound6} \\
	&\geq \varrho_{\mathrm{th}}^{\mathrm{gc}} - \int_{\left (2 - \sqrt{3} \right) \frac{2 \pi}{L} < |p| < p_\mathrm{c}} \frac{1}{e^{\beta \left( p^2 - \mu_0 \right) } - 1} \left( 1 + \frac{3 \pi}{L | p |} + \frac{6 \pi}{L^2 p^2} \right) \de p - | \Lambda |^{-1} \sum_{0 < | p |  \leq 2 \frac{2 \pi}{L}} \frac{1}{e^{\beta \left( p^2 - \mu_0 \right) } - 1}. \nonumber
\end{align}
To obtain the last bound, we used Lemma~\ref{lem:latticesum1}. The integral in the second line is not larger than a constant times $p_{\mathrm{c}}/\beta$ and the sum is bounded by a constant times $(\beta L)^{-1}$. Hence,
\begin{equation}
	\varrho_{\omega} \geq \varrho_{\mathrm{th}}^{\mathrm{gc}} - \text{const. } \left( \frac{p_{\mathrm{c}}}{\beta} + \frac{1}{\beta L} \right). 
	\label{eq:finallowerbound7}
\end{equation} 
When we follow the argumentation in \cite[Sec.~2.14]{Sei2008} and invoke Lemma~\ref{lem:latticesum1}, we see that 
\begin{equation}
	\gamma_b \geq \varrho_{\omega}  - \text{const. } \left[ \varrho \left( \frac{R}{b} \right)^2 + \frac{R^2}{\beta^{5/2}}\left( 1 + \frac{\beta}{L^2} \right)  \right].
	\label{eq:finallowerbound8}
\end{equation} 
Eqs.~\eqref{eq:finallowerbound7}, \eqref{eq:finallowerbound8} together imply, that the terms in the second line in  \eqref{eq:finallowerbound5b} are bounded from below by
\begin{equation}
	4 \pi a_N | \Lambda | \left( 2 \varrho^2 - \left( \varrho^{\mathrm{gc}}_0 \right)^2 \right) - \text{const. } a_N | \Lambda | \left[ \varrho^2 \left( \frac{R}{b} \right)^2 + \varrho \left( \frac{p_{\mathrm{c}}}{\beta} + \frac{R^2}{\beta^{5/2}} \left( 1 + \frac{\beta}{L^2} \right) + \frac{1}{\beta L} \right) \right].
	\label{eq:finallowerbound9}
\end{equation}
It remains to replace the grand canonical condensate density $\varrho^{\mathrm{gc}}_0(\beta,N,L)$ by its canonical version $\varrho_0(\beta,N,L)$. This can be achieved with the help of Lemma~\ref{lem:comparisoncondensatedensities} in the Appendix which implies
\begin{equation}
	\left( \varrho^{\mathrm{gc}}_0 \right)^2 \leq \varrho_0^2 + \text{const. } \varrho \left[ \left( \frac{\varrho \ln(N)}{\beta L} \right)^{1/2} + \frac{\ln(N)}{\beta L} \right].
	\label{eq:finallowerbound9b}
\end{equation}
Together with \eqref{eq:finallowerbound9} this implies the result we were looking for. It has been derived under the assumption $p_{\mathrm{c}}>0$. For $p_{\mathrm{c}} = 0$, we have $\varrho^0 = \varrho = \varrho_{\omega}$. Using this and \eqref{eq:finallowerbound8}, we see that the terms in the second line in \eqref{eq:finallowerbound5b} are in this case bounded from below by
\begin{equation}
	4 \pi a_N | \Lambda | \varrho^2 - \text{const. } a_N | \Lambda | \left[ \varrho^2 \left( \frac{R}{b} \right)^2 + \varrho \frac{R^2}{\beta^{5/2}} \left( 1 + \frac{\beta}{L^2} \right) \right].
	\label{eq:finallowerbound10}
\end{equation}
In combination, \eqref{eq:finallowerbound5b} and \eqref{eq:finallowerbound9}--\eqref{eq:finallowerbound10} imply that the term in the second line of \eqref{eq:finallowerbound1} plus the first term in the third line are bounded from below by 
\begin{align}
	4 \pi a_N& | \Lambda | \left( 2 \varrho^2 - \varrho_0(\beta,N,L) \right) - \text{const. } \frac{a_N}{R^3} \left[ | \Lambda | p_{\mathrm{c}}^3 + \left( | \Lambda | \beta a_N \varrho^2 + \beta Z^{(1)} + \ln(N) \right)^{1/2} A^{1/2} \right] \label{eq:finallowerbound11} \\
	&\hspace{1cm} - \text{const. } a_N | \Lambda | \left[ \varrho^2 \left( \frac{R}{b} \right)^2 + \varrho \left( \frac{p_{\mathrm{c}}}{\beta} + \frac{R^2}{\beta^{5/2}} \left( 1 + \frac{\beta}{L^2} \right) + \frac{1}{\beta L} + \left( \frac{\varrho \ln(N)}{\beta L} \right)^{1/2} + \frac{\ln(N)}{\beta L} \right) \right]. \nonumber
\end{align}
We recall that $A$ has been defined in \eqref{eq:boundnumberofparticles1k}. The result has been obtained under the assumption that $R \varrho^{1/3}$ is small enough.
\subsection{The non-interacting free energy}
\label{sec:freefreeenergy}
In this Section we derive a lower bound for the first two terms on the right-hand side of \eqref{eq:finallowerbound1}.
The dispersion relation $\epsilon(p)$ has been defined in \eqref{eq:Dysonholes2}. The following Lemma will be necessary to derive a lower bound for the second term.
\begin{lemma}
	\label{lem:boundchemicalpotential}
	The chemical potential $\mu(\lambda)$ satisfies
	\begin{equation}
	\mu_0 \leq \mu(\lambda) \leq \min\left\{ \left( \frac{2 \pi}{L} \right)^2, \lambda \right\} - \frac{1}{\beta} \ln\left( 1 + \frac{1}{N} \right).
	\label{eq:freefreeenergy1c}
	\end{equation}
\end{lemma}
\begin{proof}
	The lower bound follows from the fact that the map $\lambda \mapsto \mu(\lambda)$ is monotone increasing. The lowest eigenvalue of the operator $p^2 + \delta_{p,0} \lambda$ is given by $\min\left\{ \left( 2 \pi/L \right)^2, \lambda \right\} $. To prove the upper bound, we realize that
	\begin{equation}
	\frac{1}{\exp\left( \beta \left( \min\left\{ \left( \frac{2 \pi}{L} \right)^2, \lambda \right\}  - \mu(\lambda) \right) \right) - 1 } \leq N.
	\label{eq:freefreeenergy1b}
	\end{equation}
	The above statement follows from the fact that we have $N$ particles in the system, and hence the expected number of particles in the condensate cannot exceed $N$. Eq.~\eqref{eq:freefreeenergy1b} is equivalent to the second inequality in \eqref{eq:freefreeenergy1c} and proves the claim.
\end{proof} 
A long as $| p | \leq 1/s$ we have $\epsilon(p) = (1 - \kappa + \kappa') p^2 + \delta_{p,0} \lambda - \mu(\lambda)$ and $\epsilon(p) \geq \kappa' p^2 - \mu(\lambda)$ holds for $| p | > 1/s$. We will later choose the parameters such that $s^2/(L^2 \kappa')$ is much smaller than one. Together with $\lambda \leq (\frac{2 \pi}{L})^2 \eta$ for $0 < \eta < 1$, this in particular implies that $\epsilon(p) > 0$ for $|p | \leq 1/s$ and $\kappa' p^2 - \mu(\lambda) > 0$ for $| p | > 1/s$. In accordance with this decomposition of the momenta, we split the sum in the first line on the right-hand side of \eqref{eq:finallowerbound1} into two parts. The first part is given by
\begin{align}
	\tfrac{1}{\beta} \sum_{| p | \leq 1/s} \ln\left( 1 - e^{-\beta \left( (1- \kappa + \kappa') p^2 + \delta_{p,0} \lambda - \mu(\lambda) \right) } \right) \geq& \tfrac{1}{\beta} \sum_{p} \ln\left( 1 - e^{-\beta \left( p^2 + \delta_{p,0} \lambda - \mu(\lambda) \right) } \right) \label{eq:freefreeenergy2a} \\
	&- \sum_{p} \frac{(\kappa - \kappa') p^2}{\exp\left( \beta \left( (1- \kappa + \kappa') p^2 + \delta_{p,0} \lambda - \mu(\lambda) \right) \right) - 1} \nonumber
\end{align}
To arrive at the right-hand side, we used the concavity of the map $x \mapsto \ln(1 - e^{-x})$. An application of Lemma~\ref{lem:boundchemicalpotential} together with the assumption that $| \kappa - \kappa' | $ is small enough tells us that the absolute value of the term in the second line is bounded from above by a constant times
\begin{align}
	\left| \kappa - \kappa' \right| \left[ \tfrac{1}{\beta} + \sum_{| p | > 2 \frac{2 \pi}{L} } \frac{p^2}{e^{-\beta p^2/4} - 1} \right].
	\label{eq:freefreeenergy3}
\end{align}
The summands times $\beta$ are bounded from above by a constant times $e^{-\beta p^2/8}$. An application of Lemma~\ref{lem:latticesum1} therefore tells us that the sum in \eqref{eq:freefreeenergy3} cannot be larger than a constant times $| \Lambda | / \beta^{5/2}$.

The part of the sum in the first line of \eqref{eq:finallowerbound1} coming from the momenta with $| p | > 1/s$ is bounded from below by
\begin{equation}
	\tfrac{1}{\beta} \sum_{|p| > 1/s} \ln\left( 1 - e^{-\beta \left( \kappa' p^2 - \mu(\lambda) \right) } \right) \geq \tfrac{1}{\beta} \sum_{| p | > 1/s} \ln\left( 1 - e^{- \beta \left( \kappa' p^2/2 + \frac{\kappa'}{2 s^2} - \left( \frac{2\pi}{L} \right)^2 \right)} \right)
	\label{eq:freefreeenergy4}
\end{equation}
As already mentioned in the discussion after Lemma~\ref{lem:boundchemicalpotential}, we will later choose $\kappa'$ and $s$ such that $\kappa' /(2 s^2) - (2\pi/L)^2 \geq 0$. Since $x \mapsto \ln(1-e^{-x})$ is a negative and monotone increasing functions, we can use Lemma~\ref{lem:latticesum1} to show that the right-hand side of \eqref{eq:freefreeenergy4} is bounded from below by a constant times
\begin{equation}
	\frac{| \Lambda |}{\beta^{5/2} \left(\kappa' \right)^{3/2}} \int_{ | p | \geq \left( \beta \kappa' \right)^{1/2} \left[ \frac{1}{s} - \sqrt{3} \frac{2 \pi}{L} \right]_+} \ln\left( 1 - e^{- p^2/2} \right) \left( 1 + \frac{\beta \kappa'}{p^2 L^2} \right) \de p.
	\label{eq:freefreeenergy5}
\end{equation}
We will later choose $\kappa'$ and $s$ such that $\beta \kappa'/s^2 \gg 1$. We also note that the term in \eqref{eq:freefreeenergy5} is an exponentially decaying function of this parameter. Putting the results of this section and the definition of $\kappa'$ in \eqref{eq:Dysonholes2b} together, in particular, \eqref{eq:freefreeenergy2a}--\eqref{eq:freefreeenergy5}, we find 
\begin{align}
	\mu(\lambda) N + \tfrac{1}{\beta} \sum_p \ln\left( 1 - e^{-\beta \epsilon(p)} \right) \geq F_0\left( \beta, N, L, \lambda \right) &- \text{const. } \left\{ \frac{a_N R_0^2}{R^3} \left[ \frac{1}{\beta} + \frac{| \Lambda | }{\beta^{5/2}} \right] + \frac{\ln(N)}{\beta} \right\} \label{eq:freefreeenergy6}  \\
	&- \text{const. } \frac{| \Lambda |}{\beta^{5/2} \left(\kappa' \right)^{3/2}} \exp\left( - \text{const. } \sqrt{\beta \kappa'/s^2} \right) \nonumber
\end{align}
To obtain the result, we also used $s \leq L$ and Lemma~\ref{lem:freeenergy} in the Appendix to replace the grand canonical free energy $F_0^{\mathrm{gc}}\left( \beta, \mu(\lambda), L, \lambda \right)$ by the canonical free energy $F_0(\beta,N,L,\lambda)$. The bound has been derived under the assumptions that $s^2/(L^2 \kappa')$, $s^2/(\beta \kappa')$ and $| \kappa - \kappa' |$ are small enough. 
 \subsection{Choice of Parameters}
 \label{sec:choiceofparameters}
Optimization under the assumptions $a_N = a_{v} L /N$ with fixed $a_{v} > 0$, $\lambda \leq (\frac{2 \pi}{L})^2 \eta$ with fixed $0<\eta<1$ and $\beta \varrho^{2/3} \gtrsim 1$ leads to the same choice of parameters as in \cite[Sec.~2.16]{Sei2008} and implies the lower bound
\begin{align}
	 \tr\left[ H_N^{\lambda} \Gamma \right] - \tfrac{1}{\beta} S(\Gamma) \geq& F_0(\beta,N,L,\lambda) + 4 \pi | \Lambda | a_N \left( 2 \varrho^2 - \varrho_0(\beta,N,L)^2  \right) - L^{-2} N \widetilde{c}_{\ell} (\beta,N,L) \label{eq:choiceofparameters1} \\
	 &+ \tfrac{1}{\beta} \int_{\mathbb{C}^M} S\left( \Gamma_z, \Gamma_{\mathrm{c}}^{0,\lambda} \right) \zeta_{\Gamma}(z) \de z \nonumber 
\end{align}
with
\begin{equation}
	\widetilde{c}_{\ell} (\beta,N,L) \leq C_{\delta} \left( \beta \varrho^{2/3} \right) N^{-4/1209 + \delta} \left( \beta \varrho^{2/3} \right)^{5/403 - \delta}
	\label{eq:choiceofparameters2} 
\end{equation}
and some function $s \mapsto C_{\delta}(s)$ that is uniformly bounded for $s \in [d,\infty)$ with $d>0$. For $\delta \to 0$, the function $C_{\delta}$ blows up. The error term is of lower order as long as $N^{-2/3} (\beta \varrho^{2/3})^{5/2} \ll 1$. The bound is uniform for $0 < a_v \lesssim 1$. The desired uniformity in the temperature will be achieved in the next section. 

Since it will be needed in Sec.~\ref{sec:asymptotics1pdm}, we also state here the choices of $p_{\mathrm{c}}$ and $R$ resulting from the optimization. They are given by
\begin{equation}
	p_{\mathrm{c}} = \begin{cases} \beta^{-1/2} \left( a_N \varrho^{2} \beta^{5/2} \right)^{81/403} & \text{ if } \  \beta | \mu_0| \leq \left( a_N \varrho^{2} \beta^{5/2} \right)^{162/403} \\ 0 & \text{ otherwise. } \end{cases}
	\label{eq:pc} 
\end{equation}
and
\begin{equation}
	R = \varrho^{-1/3} \left( a_N \varrho^{2} \beta^{5/2} \right)^{3/403}.
	\label{eq:R} 
\end{equation}
\subsection{Uniformity in the temperature}
We follow the analysis in \cite[Sec.~2.17]{Sei2008} until equation (2.17.4) and arrive at
\begin{equation}
	H_N \geq \sum_{j=1}^N \ell \left( \sqrt{-\Delta_j} \right) + 4 \pi a_N N \varrho \left( 1 - \epsilon - \text{const. } \left( a^3 \varrho \right)^{1/17} - \text{const. } \frac{1}{\epsilon R s^2 \varrho} \right),
	\label{eq:uniformity1}
\end{equation}
where $\ell(|p|) = p^2 \left( 1 - \kappa/2 - (1-\kappa) \nu(sp)^2 \right)$. To obtain this bound, it has been assumed that $\kappa = ( a_N^3 \varrho )^{1/17}$ and $R = a_N ( a_N^3 \varrho )^{-5/17}$. From this point on we have to adjust the analysis in \cite{Sei2008}. This is necessary because we cannot replace sums by integrals, we have to add the term $\delta_{p,0} \lambda$ to the one-particle Hamiltonian, and we want to obtain the canonical free energy and the canonical condensate density (in the thermodynamic limit the canonical and the grand canonical free energies and condensate densities are the same). 

Denote the first term on the right-hand side of Eq.~\eqref{eq:uniformity1} by $T_N^{\mathrm{c}}$ and let
\begin{equation}
	T_N^{\mathrm{c},\lambda} = T_N^{\mathrm{c}} + \lambda \sum_{i=1}^N | \Phi \rangle \langle \Phi |_i,
	\label{eq:uniformity1b}
\end{equation}
where $\Phi(x) = | \Lambda |^{-1/2}$, as before. For any $N$-particle state $\Gamma$ we have
\begin{align}
	\tr \left[ T^{\mathrm{c},\lambda}_N \Gamma \right] - TS(\Gamma) &= \tr \left[ \left( T^{\mathrm{c},\lambda}_N - \mu(\lambda) N \right) \Gamma \right] - TS(\Gamma) + \mu(\lambda) N \label{eq:uniformity2} \\
	&\geq \tfrac{1}{\beta} \sum_{p} \ln \left( 1- e^{-\beta \left( \ell(|p|) + \delta_{p,0} \lambda - \mu(\lambda) \right)} \right) + \mu(\lambda) N. \nonumber
\end{align}
If we assume that $s^2/(\kappa L^2)$ is small enough, we have $\ell(|p|) + \delta_{p,0} \lambda - \mu(\lambda) > 0$ for all $p \in \frac{2 \pi}{L} \mathbb{Z}^3$. This follows from $\kappa = ( a_N^3 \varrho )^{1/17} \ll 1$ and $\lambda \leq (\frac{2 \pi}{L})^2 \eta$ with $0 < \eta < 1$.  We use $\ell(|p|) = \left( 1 - \kappa/2 \right) p^2$ for $| p | \leq 1/s$ to see that
\begin{align}
	\tfrac{1}{\beta} \sum_{| p | \leq 1/s} \ln \left( 1- e^{-\beta \left( \ell(|p|) + \delta_{p,0} \lambda - \mu(\lambda) \right)} \right) \geq& \tfrac{1}{\beta} \sum_{p} \ln \left( 1- e^{-\beta \left( p^2 + \delta_{p,0} \lambda - \mu(\lambda) \right)} \right) \label{eq:uniformity3} \\
	&- \sum_{p} \frac{\frac{\kappa p^2}{2}}{\exp\left(\beta \left( p^2 + \delta_{p,0} \lambda - \mu(\lambda) \right)\right) - 1}, \nonumber  
\end{align}
as in \eqref{eq:freefreeenergy2a}. The term in the second line of \eqref{eq:uniformity3} can be quantified as a similar term in Sec.~\ref{sec:freefreeenergy}, compare with \eqref{eq:freefreeenergy3}. This is also true for the sum over all momenta with $| p | > 1/s$. Following these arguments and replacing again the grand canonical free energy by the canonical one with Lemma~\ref{lem:freeenergy} in the Appendix, we find that
\begin{align}
		\tr \left[ T^{\lambda}_N \Gamma \right] - TS(\Gamma) \geq& F_0\left( \beta, N, L ,\lambda \right) - \text{const. } \left[ \kappa \left( \frac{ 1 }{\beta} + \frac{| \Lambda |}{\beta^{5/2}} \right) + \frac{\ln(N)}{\beta} \right]  \label{eq:uniformity3b} \\
		&- \text{const. } \frac{| \Lambda |}{\beta^{5/2} \kappa^{3/2}} \exp\left( - \text{const. } \sqrt{\beta \kappa/s^2} \right) \nonumber
\end{align}
holds. To obtain the bound, we assumed that $s^2/(L^2 \kappa)$ and $s^2/(\beta \kappa)$ are small enough. 

In order to obtain the final estimate, we also need to replace the interaction energy in Eq.~\eqref{eq:uniformity1} by the formula we have in Theorem~\ref{thm:periodic}. As above we denote by $\varrho_0(\beta,N,L)$ the expected condensate density of the ideal gas in the canonical ensemble in the case $\lambda = 0$ and we define by $\varrho_{\mathrm{th}} = \varrho - \varrho_0$ the expected density of the thermal cloud. We then have for $\beta \gtrsim \beta_{\mathrm{c}}$
\begin{align}
	4 \pi a_N | \Lambda | \varrho^2 &\geq 4 \pi a_N | \Lambda | \left[ 2 \varrho^2 - \varrho_0^2 \right] - 4 \pi a_N | \Lambda | \left( 2 \varrho_0 \varrho_{\mathrm{th}} + \varrho_{\mathrm{th}}^2 \right) \label{eq:uniformity7} \\
	&\geq 4 \pi a_N | \Lambda | \left[ 2 \varrho^2 - \varrho_0^2 \right] - \text{const. } a_N | \Lambda | \varrho \varrho^{\mathrm{gc}}_{\mathrm{th}}. \nonumber
\end{align}
To come to the second line, we used $\varrho_0 \leq \varrho$ as well as Lemma~\ref{lem:cvsgcexpectations} in the appendix to bound $\varrho_{\mathrm{th}} \lesssim \varrho_{ \mathrm{th} }^{ \mathrm{gc} }$. To see that $\varrho^{\mathrm{gc}}_{\mathrm{th}} \lesssim \beta^{-3/2} + \beta/L$, we write
\begin{equation}
	|\Lambda| \varrho^{\mathrm{gc}}_{ \mathrm{th} } = \sum_{ 0 < |p| < 2 \frac{2 \pi}{L}} \frac{1}{e^{\beta \left( p^2 - \mu_0 \right) } - 1} + \sum_{| p | \geq 2 \frac{2 \pi}{L}} \frac{1}{e^{\beta \left( p^2 - \mu_0 \right) } - 1}. 
	\label{eq:uniformity7b}
\end{equation}
The first term on the right-hand side is bounded by a constant times $L^2/\beta$. To bound the second term, we invoke Lemma~\ref{lem:latticesum1} to see that it is bounded from above by a constant times $| \Lambda | / \beta^{3/2}$. 

In combination, \eqref{eq:uniformity3b} and \eqref{eq:uniformity7} imply:
\begin{align}
	\tr \left[ H^{\lambda}_N \Gamma \right] - TS(\Gamma) \geq& F_0(\beta,N,L,\lambda) + 4 \pi a_N | \Lambda | \left( 2 \varrho^2 - \varrho_0(\beta,N,L)^2 \right) 
	- \text{const. } \left[ \kappa \left( \frac{N}{L^2} + \frac{ L }{\beta^{3/2} } + \frac{| \Lambda |}{\beta^{5/2}} \right) + \frac{\ln(N)}{\beta} \right] \nonumber \\
	& -\text{const. } a_N | \Lambda | \varrho^2 \left[ \epsilon + \left( a_N^3 \varrho \right)^{1/17} + \frac{1}{\epsilon R s^2 \varrho} + \frac{1}{\beta^{3/2} \varrho} \right] \label{eq:uniformity8} \\
	&- \text{const. } \frac{| \Lambda |}{\beta^{5/2} \kappa^{3/2}} \exp\left( - \text{const. } \sqrt{\beta \kappa/s^2} \right). \nonumber
\end{align}
In the derivation of this bound, we assumed that $s^2/(L^2 \kappa)$ and $s^2/(\beta \kappa)$ are small enough. Before we optimize under the assumption $\lambda \leq (\frac{2\pi}{L})^2 \eta$ with fixed $0<\eta<1$ and $a_N = a_v L/N$ with $a_v > 0$ fixed, we insert $\kappa$ and $R$ from above, see the discussion after \eqref{eq:uniformity1}. With some $\delta > 0$ we choose
\begin{equation}
	s^2 = \min \left\{ L^2, \beta \right\} \left( a_N^3 \varrho \right)^{1/17 + \delta} \quad \text{ and } \quad \epsilon^2 = \frac{1}{R s^2 \varrho}.
\label{eq:uniformity9a}
\end{equation}
This fulfills the condition on $s$ and $\kappa$ stated after \eqref{eq:uniformity8} and it assures the smallness of the term in the third line in \eqref{eq:uniformity8}. It implies
\begin{align}
	&\tr \left[ H^{\lambda}_N \Gamma \right] - TS(\Gamma) \geq F_0(\beta,N,L,\lambda) + 4 \pi a_N | \Lambda | \left( 2 \varrho^2 - \varrho_0(\beta,N,L)^2 \right) \label{eq:uniformity9} \\
	&\hspace{2cm}- \text{const. } | \Lambda | a_N \varrho^2 \left[ \left( \frac{1}{\min\{L^2,\beta \} \varrho^{2/3} \left( a_N^3 \varrho \right)^{5/51 + \delta}} \right)^{1/2} + \left( a_N^3 \varrho \right)^{1/17} \left( 1 + \frac{1}{ a_N \varrho^2 \beta^{5/2} } \right) + \frac{1}{\left( \beta \varrho^{2/3} \right)^{3/2}}  \right]. \nonumber
\end{align}
We have to combine the two bounds \eqref{eq:choiceofparameters1} and \eqref{eq:uniformity9} in the same way as in \cite[Sec.~217]{Sei2008} to obtain the optimal rate, that is, we use \eqref{eq:choiceofparameters1} as long as $\beta \varrho^{2/3} \leq N^{7568/103275} \approx N^{0.0733}$ and \eqref{eq:uniformity9} otherwise. This yields the final lower bound 
\begin{equation}
	F(\beta,N,L,\lambda) \geq F_0(\beta,N,L,\lambda) + 4 \pi | \Lambda | a_N \left( 2 \varrho^2 - \varrho_0(\beta,N,L)^2 \right) - L^{-2} N c_{\ell}(N) \label{eq:finallowerbounda}
\end{equation}
 with 
\begin{equation}
	c_{\ell}(N) \leq C_{\delta} \left( \beta \varrho^{2/3} \right) N^{-\alpha}
	\label{eq:finallowerboundb} 
\end{equation}
and some function $s \mapsto C_{\delta}(s)$ that is uniformly bounded on intervals $[d,\infty)$ with $d>0$ and $\alpha = 4/6885 - \delta$. For $\delta \to 0$, the function $C_{\delta}$ blows up. The bound is uniform for $0 < a_v \lesssim 1$. This completes the proof of the lower bound.
\section{Proof of the asymptotics of the one-particle density matrix}
\label{sec:asymptotics1pdm}
In this section we prove the claimed asymptotics for the 1-pdm of approximate minimizers of the Gibbs free energy functional. A crucial input for the analysis in this section are the lower bounds \eqref{eq:choiceofparameters1} and \eqref{eq:uniformity9}. The proof is split into four parts: In the first part we consider the 1-pdm projected onto the subspace of the one-particle Hilbert space with momenta at least $p_{\mathrm{c}}$ and we show that it equals the one of the non-interacting Gibbs state to leading order. In the second step we consider the 1-pdm projected to the orthogonal complement of that subspace and show that also there it is to leading order given by the one of the non-interacting Gibbs state. In the third step we estimate the off-diagonal contributions and in the fourth part we prove the uniformity in the temperature. We highlight that off-diagonal contributions to the 1-pdm have to be estimated because we do not assume that the states under consideration are translation invariant. With this assumption we would obtain a better rate. An important example of a translation invariant state is the interacting Gibbs state \eqref{eq:Gibbsstate}.
\subsection{The one-particle density matrix of the thermal cloud}

Let $\Gamma$ be an approximate minimizers of the Gibbs free energy functional in the sense that 
\begin{equation}
	\mathcal{F}(\Gamma) \leq F_0(\beta,N,L) + 4 \pi a_N | \Lambda | \left( 2 \varrho^2 - \varrho_0(\beta,N,L)^2 \right) ( 1 + \delta(N) ),
	\label{eq:asymptotics1pdm1}
\end{equation}
with $0 \leq \delta(N) = o(1)$ in the considered limit. Together with the lower bound \eqref{eq:choiceofparameters1} with $\lambda = 0$, this implies
\begin{equation}
	\int_{\mathbb{C}^M} S\left ( \Gamma_{z} , \Gamma^{0}_{\mathrm{c}} \right) \zeta_{\Gamma} (z) \de z \leq 4 \pi | \Lambda | \beta a_N \varrho^2 \left( \widetilde{c}_{\ell}(\beta,N,L) + \delta(N) \right).
	\label{eq:asymptotics1pdm4}
\end{equation} 
The state $\Gamma^0_{\mathrm{c}} \equiv \Gamma^{0,\lambda=0}_{\mathrm{c}}$ was defined in Sec.~3.6. The index c refers to the fact that the relevant dispersion relation  is not $p^2 - \mu_0$ but the one we obtained after applying the Dyson Lemma, see \eqref{eq:Dysonholes2}. The goal of this section is to obtain quantitative information on the 1-pdm $\gamma$ of $\Gamma$ from this bound. Let us define 
\begin{equation}
	P = \mathds{1} \left( - \Delta < p_{\mathrm{c}} \right) \quad \text{ and } \quad Q = 1 - P.
	\label{eq:PQ}
\end{equation}
When projected to the high momentum modes, $\gamma$ reads
\begin{equation}
	Q \gamma Q = \int_{\mathbb{C}^M} \gamma_{z} \ \zeta_{\Gamma} (z) \de z,
	\label{eq:asymptotics1pdm5}
\end{equation}
where $\gamma_{z}$ is the 1-pdm of the state $\Gamma_{z}$. Hence, if we denote by $\gamma^0_{\mathrm{c}}$ the 1-pdm of $\Gamma^0_{\mathrm{c}}$, we have
\begin{equation}
	\left\| Q \left( \gamma - \gamma^0_{\mathrm{c}} \right)  Q \right\|_1 \leq \int_{\mathbb{C}^M} \left\| \gamma_{z} - \gamma^0_{\mathrm{c}} \right\|_1 \zeta_{\Gamma} (z) \de z. 
	\label{eq:asymptotics1pdm6}
\end{equation}
In the following, we will derive a bound on the right-hand side.

The starting point of our analysis is \eqref{eq:asymptotics1pdm4}. Since $\Gamma_{\mathrm{c}}^0$ is a quasi-free state, the left-hand side \eqref{eq:asymptotics1pdm4} can be bounded from below in terms of the bosonic relative entropy. For two nonnegative operators $a, b$ with finite trace, it is defined by
\begin{equation}
\mathcal{S}(a,b) = \sum_{i,j} \left| \langle \psi_i, \varphi_j \rangle \right|^2 \left( \sigma(\gamma_i) - \sigma(\eta_j) - \sigma'(\eta_j)(\gamma_i-\eta_j)  \right),
\label{eq:asymptotics1pdm7}
\end{equation}
where $\sigma(x) = x \ln(x) - (1+x) \ln (1+x)$ and $\{ \lambda_i, \psi_i \}$ and $\{ \eta_j, \varphi_j \}$ denote the eigenvalues and eigenfunctions of $a$ and $b$, respectively. We also denote by 
\begin{equation}
	\mathcal{S}(a) = -\trs \left[ \sigma\left( a \right) \right]
	\label{eq:asymptotics1pdm8}
\end{equation}
the bosonic entropy of $a$. We then have  
 $\mathcal{S}(\gamma_z) \geq S(\Gamma_z)$, see \cite[2.5.14.5]{Thirring_4}, as well as $S(\Gamma^0_{\mathrm{c}}) = \mathcal{S}(\gamma^0_{\mathrm{c}})$, and therefore conclude that
 \begin{equation}
 	S\left(\Gamma_z,\Gamma^0_{\mathrm{c}} \right) \geq \mathcal{S}\left( \gamma_z, \gamma^0_{\mathrm{c}} \right).
 	\label{eq:asymptotics1pdm8b}
 \end{equation}
In combination with \eqref{eq:asymptotics1pdm4}, this implies
\begin{equation}
\int_{\mathbb{C}^M} \mathcal{S}\left( \gamma_z , \gamma^0_{\mathrm{c}} \right) \zeta_{\Gamma}(z) \de z \lesssim \beta | \Lambda | a_N \varrho^2 \left( \widetilde{c}_{\ell}(\beta,N,L) + \delta(N) \right).
\label{eq:asymptotics1pdm9}
\end{equation}
In order to obtain quantitative information from Eq.~\eqref{eq:asymptotics1pdm9}, we need the following Lemma which quantifies the coercivity of the bosonic relative entropy. It is an improved version of \cite[Lemma~4.1]{me}.
\begin{lemma}
	\label{lem:relativeentropy}
	There exists a constant $C>0$ such that for any two nonnegative trace-class operators $a$ and $b$ we have
	\begin{equation}
	\mathcal{S}\left( a, b \right) \geq C \frac{\Vert a - b \Vert_1^2}{\left\Vert 1+b \right\Vert \trs\left[ a + b \right]}.
	\label{eq:asymptoticslemma2}
	\end{equation}
\end{lemma}
\begin{proof}[Proof of Lemma~\ref{lem:relativeentropy}]
	Let $f(x,y) = \sigma(x) - \sigma(y) - \sigma'(y)(x-y)$. In the proof of Lemma~4.1 in \cite{me} it has been shown that there is a $C>0$ such that
	\begin{equation}
		f(x,y) \geq C \frac{(x-y)^2}{(1+y)(x+y)}.
		\label{eq:asymptoticslemma5}
	\end{equation}
	We write $(x-y)^2 = \left( \sqrt{x} - \sqrt{y} \right)^2 \left( \sqrt{x} + \sqrt{y} \right)^2$ which allows us to bound the right-hand side from below by
	\begin{equation}
		\frac{(x-y)^2}{(1+y)(x+y)} \geq \frac{\left( \sqrt{x} - \sqrt{y} \right)^2}{1+y} \geq \frac{\left( \sqrt{x} - \sqrt{y} \right)^2}{1+y_{\mathrm{max}}}.
		\label{eq:asymptoticslemma6}
	\end{equation}
	To obtain the final estimate, we assumed that $y \in [0,y_{\mathrm{max}}]$. In combination with \eqref{eq:asymptotics1pdm7}, this proves
	\begin{equation}
		\mathcal{S}(a,b) \geq C \frac{ \left\Vert a^{1/2} - b^{1/2} \right\Vert_2^2}{\Vert 1 + b \Vert}.
		\label{eq:asymptoticslemma1a}
	\end{equation}
	Next, we write the difference between the two density matrices as
	\begin{equation}
	a - b = \tfrac{1}{2} \left( a^{1/2} - b^{1/2} \right)\left( a^{1/2} + b^{1/2} \right) + \tfrac{1}{2} \left( a^{1/2} + b^{1/2} \right) \left( a^{1/2} - b^{1/2} \right)
	\label{eq:asymptoticslemma3}
	\end{equation}
	and estimate their trace norm difference by
	\begin{equation}
	\left\Vert a - b \right\Vert_1 \leq \left\Vert a^{1/2} - b^{1/2} \right\Vert_2  \left( \left\Vert a^{1/2} \right\Vert_2 + \left\Vert b^{1/2} \right\Vert_2 \right) = \left\Vert a^{1/2} - b^{1/2} \right\Vert_2  \left( \left\Vert a \right\Vert_1^{1/2} + \left\Vert b \right\Vert_1^{1/2} \right). 
	\label{eq:asymptoticslemma4}
	\end{equation}
	Here $\Vert \cdot \Vert_2$ denotes the Hilbert-Schmidt norm. Together with \eqref{eq:asymptoticslemma1a} and $\sqrt{x} + \sqrt{y} \leq \sqrt{2} \sqrt{x+y}$ for $x,y \geq 0$, this proves the claim.
\end{proof}
\begin{remark}
	Since the operator norm of $b$ appears in the denominator on the right-hand side of \eqref{eq:asymptoticslemma2}, Lemma~\eqref{lem:relativeentropy} is useful only in the case where the largest eigenvalue of $b$ is not too large. In particular, $b$ should not have a condensate.
\end{remark}
\begin{remark}
	The  Lemma above is an improvement w.r.t. \cite[Lemma~4.1]{me} since it gives a lower bound for the bosonic relative entropy directly in terms of the trace norm difference of the two density matrices. In \cite{me} a considerable amount of further analysis was necessary, during which additional  error terms are accumulated, to obtain such a bound. For this reason one obtains better rates of convergence from given bounds for the bosonic relative entropy when using Lemma~\ref{lem:relativeentropy} compared to what one would obtain with \cite[Lemma~4.1]{me}. Given the size of the remainder in \eqref{eq:choiceofparameters1}, an improved rate of convergence is not of particular relevance  for the analysis here, however. Bounds for the trace norm difference of two one-particle density matrices in terms of the relative entropy of the related states have recently also been proven in \cite[Theorem~6.1]{LewinNamRougerie2018}.
\end{remark}
Before we apply Lemma~\ref{lem:relativeentropy}, we will show that $\trs \gamma_z \leq N$ holds. To that end, we write
\begin{equation}
	\trs \gamma_z = \frac{\tr \left[ \mathbb{N}^> | z \rangle\langle z | \ \Gamma \right]}{\tr \widetilde{\Gamma}_z},
	\label{eq:intermezzo1}
\end{equation}
where $\mathbb{N}^> = \sum_{|p| \geq p_{\mathrm{c}}} a_p^* a_p$, as before. Let us denote by $P_N$ the projection onto the $N$-particle sector of the Fock space $\mathscr{F}$. It is sufficient to show that
\begin{equation}
	P_N \mathbb{N}^> | z \rangle\langle z | P_N \leq N P_N  | z \rangle\langle z | P_N
	\label{eq:intermezzo2}
\end{equation}
holds. With $[\mathbb{N}^>, | z \rangle\langle z | ] = 0 = [\mathbb{N}^>, P_N ]$ we check that $[P_N \mathbb{N}^> P_N, P_N | z \rangle\langle z | P_N ] = 0$. But this implies the claim. With this information at hand, we apply Lemma~\ref{lem:relativeentropy} to the left-hand side of \eqref{eq:asymptotics1pdm9} and additionally use $\trs \gamma_{\mathrm{c}}^0 \lesssim N$ as well as $\Vert 1 + \gamma^0_{\mathrm{c}} \Vert \lesssim L^2/\beta$. We find that
\begin{equation}
	\left\| Q \left( \gamma - \gamma^0_{\mathrm{c}} \right) Q \right\|_1 \lesssim N \left( \widetilde{c}_{\ell}(\beta,N,L) + \delta(N) \right)^{1/2}
	\label{eq:asymptotics1pdm11}
\end{equation}
holds. It remains to replace $\gamma^0_{\mathrm{c}}$ by the canonical 1-pdm of the ideal Bose gas. 

To that end, we first replace the dispersion relation $\epsilon(p)$ \eqref{eq:Dysonholes2} in the definition of $\gamma^0_{\mathrm{c}}$ by $p^2 - \mu_0$. This can be done with an analysis that is very similar to the one carried out between \eqref{eq:freefreeenergy2a} and \eqref{eq:freefreeenergy5} in Sec.~\ref{sec:freefreeenergy} and yields
\begin{equation}
	\left\| Q \left( \gamma^0_{\mathrm{c}} - \gamma_0^{\mathrm{gc}} \right) Q \right\|_1 \lesssim | \kappa - \kappa' |  \frac{| \Lambda | }{\beta^{3/2}} +  \frac{| \Lambda |}{\left( \beta \kappa' \right)^{3/2}} \exp\left( - \text{const. } \sqrt{\kappa' \beta/s^2} \right). \label{eq:asymptotics1pdm11b}
\end{equation}
In order to replace $\gamma_0^{\mathrm{gc}}$ by its canonical analogue $\gamma_0$, we invoke Lemma~\ref{lem:comparisoncondensatedensities} in the Appendix to show that
\begin{align}
	\left\Vert Q \left( \gamma_0 - \gamma_0^{\mathrm{gc}} \right) Q \right\Vert_1 &\leq \left\Vert \mathds{1}\left( - \Delta \neq 0 \right) \left( \gamma_0 - \gamma_0^{\mathrm{gc}} \right) \mathds{1}\left( - \Delta \neq 0 \right) \right\Vert_1 \label{eq:asymptotics1pdm11c} \\
	&\lesssim  \left( \trs \left[ \widetilde{\gamma}_{0}^{\mathrm{gc}} \left( 1 + \widetilde{\gamma}_{0}^{\mathrm{gc}} \right) \right] \ln N \right)^{1/2} + \left( 1 + \left\Vert \widetilde{\gamma}_{0}^{\mathrm{gc}} \right\Vert \right) \ln N \nonumber
\end{align}
holds, where $\widetilde{\gamma}_{0}^{\mathrm{gc}} =  \mathds{1}\left( - \Delta \neq 0 \right) \gamma_0^{\mathrm{gc}}$. With $\Vert \widetilde{\gamma}_{0}^{\mathrm{gc}} \Vert \lesssim L^2/\beta$ one checks that the term in the second line of \eqref{eq:asymptotics1pdm11c} is bounded from above by a constant times $N^{5/6} \ln(N)^{1/2}$ uniformly in $\beta \varrho^{2/3} \gtrsim 1$. In combination  \eqref{eq:asymptotics1pdm11}--\eqref{eq:asymptotics1pdm11c} imply
\begin{equation}
	\left\| Q \left( \gamma - \gamma_0 \right) Q \right\|_1 \lesssim N \left( \widetilde{c}_{\ell}(\beta,N,L) + \delta(N) \right)^{1/2}.
	\label{eq:asymptotics1pdm11d}
\end{equation}
The bound yields the desired result as long as $N^{-2/3} ( \beta \varrho^{2/3} )^{5/2} \ll 1$.
\subsection{The one-particle density matrix of the condensate}
\label{sec:expectednumberofparticlesincondensate}
In order to investigate $P \gamma P$ and, in particular, to show the existence of a BEC, we apply a Griffith argument. 
From Eq.~\eqref{eq:finallowerbounda}, we know that
\begin{equation}
	\mathcal{F}(\Gamma) + \lambda   \langle \Phi | \gamma | \Phi \rangle =  \tr\left[ H_N^{\lambda} \Gamma \right] - \tfrac{1}{\beta} S(\Gamma) \geq F_0(\beta,N,L,\lambda) - 4 \pi a_N | \Lambda | \left( 2 \varrho^2 - \varrho_0(\beta,N,L)^2 \right) \left( 1 - c_{\ell}(N) \right),
	\label{eq:1pdmcondensate3}
\end{equation}
where the Hamiltonian $H_N^{\lambda}$ was defined in \eqref{eq:1pdmcondensate1a} and $0 \leq \lambda \leq (\frac{2 \pi}{L})^2 \eta$ with some fixed $0 < \eta <1$. Together with \eqref{eq:asymptotics1pdm1} this implies
\begin{align}
	\langle \Phi | \gamma | \Phi \rangle &\geq \frac{F_0(\beta,N,L,\lambda) - F_0(\beta,N,L,0)}{\lambda} - \frac{| \Lambda | a_N \varrho^2}{\lambda} \left( c_{\ell}(N) + \delta(N) \right) \label{eq:1pdmcondensate4} \\
	&= 	\langle \Phi | \gamma_{0} | \Phi \rangle + \frac{1}{2} \frac{\partial^2 F_0(\beta,N,L,\lambda)}{\partial \lambda^2} \bigg|_{\lambda = \widetilde{\lambda}} \ \lambda - \frac{| \Lambda | a_N \varrho^2}{\lambda} \left( c_{\ell}(N) + \delta(N) \right) \nonumber
\end{align}
for some $0 \leq \widetilde{\lambda} \leq \lambda$. As above we denoted by $\Phi$ the constant function with value $|\Lambda|^{-1/2}$ on the torus. The second derivative in the above equation is nothing but $(-\beta)$ times the variance $\textbf{Var}_{\lambda}(n_0) = \langle n_0^2 \rangle_{\lambda} - \langle n_0 \rangle_{\lambda}^2$ of the occupation of the $p=0$ orbital. Here and in the following, we denote by $\langle \cdot \rangle_{\lambda}$ the expectation in the canonical ensemble with the energy of the $p=0$ orbital shifted by $\lambda$. We also recall that $n_p = a_p^* a_p$. In order to bound the above variance, we need the following Lemma:
\begin{lemma}
	\label{lem:variance}
	Assume that $0 \leq \lambda < ( \frac{2 \pi}{L} )^2 \eta$ with $0 < \eta < 1$. We then have
	\begin{equation}
		\var_{\lambda}(n_0) \lesssim \frac{L^4}{\beta^2}.
		\label{eq:lemvariance}
	\end{equation}
\end{lemma}
\begin{proof}
	The canonical Gibbs state has exactly $N$ particles. This allows us to conclude that the particle number fluctuations of the condensate and those of the thermal cloud are equal: $\textbf{Var}_{\lambda}\left( n_0 \right) = \textbf{Var}_{\lambda} \left(\sum_{p \neq 0} n_p \right)$. 
	From \cite[Theorem~(ii)]{Suto} we know that the correlation inequality $ \langle n_p n_q \rangle_{\lambda} - \langle n_p \rangle_{\lambda} \langle n_q \rangle_{\lambda} < 0$ holds for the canonical Gibbs state of the ideal gas if $p \neq q$. Using this result, we estimate
	\begin{align}
	\textbf{Var}_{\lambda} \left( \sum_{p \neq 0} n_p  \right) = \sum_{p,q \neq 0, p \neq q} \left( \langle n_p n_q \rangle_{\lambda} - \langle n_p \rangle_{\lambda} \langle n_q \rangle_{\lambda} \right) + \sum_{p \neq 0} \left( \langle n_p^2 \rangle_{\lambda} - \langle n_p \rangle_{\lambda}^2 \right) \leq \sum_{p \neq 0} \langle n_p^2 \rangle_{\lambda} \lesssim \sum_{p \neq 0}  \langle n_p^2 \rangle_{\lambda,\mathrm{gc}}. \label{eq:1pdmcondensate5} 
	\end{align} 
	By $\langle \cdot \rangle_{\lambda,\mathrm{gc}}$ we denote the expectation w.r.t. to the grand canonical ensemble. The last inequality follows from Lemma~\ref{lem:cvsgcexpectations}. A straightforward computation shows
	\begin{equation}
	\langle n_p^2 \rangle_{\lambda,\mathrm{gc}} = \frac{1+e^{\beta \left( p^2 - \mu(\lambda) \right)}}{\left( e^{\beta \left( p^2 - \mu(\lambda) \right)} - 1 \right)^2}.
	\label{eq:1pdmcondensate5b} 
	\end{equation}
	From Lemma~\ref{lem:boundchemicalpotential} we know that $p^2 - \mu(\lambda) \geq (1-\eta) p^2$ as well as that $p^2 - \mu(\lambda) \leq p^2 - \mu_0$.
	We therefore have
	\begin{equation}
		\sum_{p \neq 0} \frac{1+e^{\beta \left( p^2 - \mu(\lambda) \right)}}{\left( e^{\beta \left( p^2 - \mu(\lambda) \right)} - 1 \right)^2} \leq \sum_{p \neq 0} \frac{1+e^{\beta \left( p^2 - \mu_0 \right)}}{\left( e^{\beta (1-\eta) p^2} - 1 \right)^2} \lesssim \frac{L^4}{(1-\eta)^2 \beta^2} \sum_{p \in \mathbb{Z}^3 \backslash \{ 0 \}} \frac{1}{p^4}. 
	\end{equation}
	Together with \eqref{eq:1pdmcondensate5} and \eqref{eq:1pdmcondensate5b}, this proves the claim.
\end{proof}
We use Lemma~\ref{lem:variance} to bound the second term on the right-hand side of \eqref{eq:1pdmcondensate4}. The choice $\lambda = (\frac{2 \pi}{L})^2/2$ implies together with $a_N \lesssim N/L$ the bound
\begin{equation}
	\langle \Phi | \gamma | \Phi \rangle \geq N_0 - \text{const. }  N \left( c_{\ell}(N) + \delta(N) \right).
	\label{eq:1pdmcondensate10}
\end{equation}
We highlight that the right-hand side of \eqref{eq:1pdmcondensate10} is uniform in $\beta \varrho^{2/3}$. 

Let
\begin{equation}
P_{\Phi} = | \Phi \rangle\langle  \Phi | \quad \text{and} \quad \widetilde{P} = \mathds{1}\left( 0 < -\Delta < p_c^2 \right) = P - P_{\Phi}.
\label{eq:offdiag1}
\end{equation}
Our next goal is to derive a bound for $\Vert P (\gamma - \gamma_0) P \Vert_1$. When we write the trace in terms of the eigenfunctions of the Laplacian and use \eqref{eq:asymptotics1pdm11d} as well as \eqref{eq:1pdmcondensate10}, we see that
\begin{align}
N = \trs \gamma  &= \langle \Phi | \gamma | \Phi \rangle + \trs \widetilde{P} \gamma \widetilde{P} + \trs Q \gamma_0 Q + \trs Q (\gamma - \gamma_0) Q \label{eq:offdiag3} \\
&\geq N_0 - N \left( c_{\ell}(N) + \delta_N \right) + \trs \widetilde{P} \gamma \widetilde{P} + \trs Q \gamma_0 Q - N \left( \widetilde{c}_{\ell}(\beta,N,L) + \delta(N) \right)^{1/2}. \nonumber
\end{align}
With Lemma~\ref{lem:cvsgcexpectations} in the Appendix and Lemma~\ref{lem:latticesum1}, we show that
\begin{equation}
\trs \widetilde{P} \gamma_0 \widetilde{P} \lesssim \frac{L^2}{\beta} + \frac{ L^3 p_c}{\beta}
\label{eq:offdiag4}
\end{equation}
holds. Since $\gamma_0$ is diagonal in the momentum basis this implies
\begin{equation}
\trs Q \gamma_0 Q \geq N_{\mathrm{th}} - \text{const. } \left( \frac{L^2}{\beta} + \frac{L^3 p_c}{\beta} \right).
\label{eq:offdiag5}
\end{equation}
Next, we insert this inequality into \eqref{eq:offdiag3} and find
\begin{equation}
\trs \widetilde{P} \gamma \widetilde{P}  \leq N \left( c_{\ell}(N) + \delta(N) + \left( \widetilde{c}_{\ell}(\beta,N,L) + \delta(N) \right)^{1/2} \right) + \text{const. } \left( \frac{L^2}{\beta} + \frac{L^3 p_c}{\beta} \right). \label{eq:offdiag6}
\end{equation}
When we insert \eqref{eq:offdiag4}, \eqref{eq:offdiag5} and \eqref{eq:offdiag6} in the first line of \eqref{eq:offdiag3} in order to obtain a lower bound for the expression on the right-hand side, we find
\begin{equation}
\langle \Phi | \gamma | \Phi \rangle \leq N_0 + N \left( c_{\ell}(N) + \delta_N + \left( \widetilde{c}_{\ell}(\beta,N,L) + \delta(N) \right)^{1/2} \right) + \text{const. } \left( \frac{L^2}{\beta} + \frac{L^3 p_c}{\beta} \right). \label{eq:offdiag7}
\end{equation}
Together with the lower bound on the same quantity \eqref{eq:1pdmcondensate10}, this implies the bound
\begin{equation}
\Vert P_{\Phi} (\gamma - \gamma_0) P_{\Phi} \Vert_1 = | \langle \Phi | \gamma - \gamma_0 | \Phi \rangle | \leq N \left( c_{\ell}(N) + \delta_N + \left( \widetilde{c}_{\ell}(\beta,N,L) + \delta(N) \right)^{1/2} \right) + \text{const. } \left( \frac{L^2}{\beta} + \frac{L^3 p_c}{\beta} \right). 
\label{eq:offdiag8}
\end{equation}
To obtain a bound for the term we are interested in, that is, for $\Vert P \left( \gamma - \gamma_0 \right) P \Vert_1$, we write
\begin{equation}
\Vert P \left( \gamma - \gamma_0 \right) P \Vert_1 \leq \Vert P_{\Phi} \left( \gamma - \gamma_0 \right) P_{\Phi} \Vert_1 + 2 \Vert P_{\Phi} \left( \gamma - \gamma_0 \right) \widetilde{P} \Vert_1 + \Vert \widetilde{P} \left( \gamma - \gamma_0 \right) \widetilde{P} \Vert_1.
\label{eq:offdiag9}
\end{equation}
It remains to give a bound on the second term on the right-hand side. To that end, we use $P_{\Phi} \gamma_0 \widetilde{P} = 0$ and estimate
\begin{equation}
\Vert P_{\Phi} \left( \gamma - \gamma_0 \right) \widetilde{P} \Vert_1 \leq \Vert P_{\Phi} \Vert_1 \Vert \gamma^{1/2} \Vert \Vert \gamma^{1/2} \widetilde{P} \Vert \leq N^{1/2} \Vert \widetilde{P} \gamma \widetilde{P} \Vert^{1/2} \leq N^{1/2} \Vert \widetilde{P} \gamma \widetilde{P} \Vert_1^{1/2}.
\label{eq:offdiag10}
\end{equation}
Putting Eqs.~\eqref{eq:offdiag4}, \eqref{eq:offdiag6}, \eqref{eq:offdiag8}, \eqref{eq:offdiag9} and \eqref{eq:offdiag10} together, we finally obtain 
\begin{equation}
\Vert P \left( \gamma - \gamma_0 \right) P \Vert_1 \lesssim N^{1/2} \left( N \left( c_{\ell}(N) + \widetilde{c}_{\ell}(\beta,N,L)^{1/2} + \delta(N)^{1/2} \right) + \frac{N^{2/3}}{\beta \varrho^{2/3}} + \frac{N  \beta^{1/2} p_c}{\left( \beta \varrho^{2/3} \right)^{3/2}} \right)^{1/2}.
\label{eq:offdiag11}
\end{equation}
This is the bound for the low momentum block of the 1-pdm of $\Gamma$ we were looking for. In combination with the choice for $p_{\mathrm{c}}$ in \eqref{eq:pc}, it implies that $\Vert P \left( \gamma - \gamma_0 \right) P \Vert_1 $ is much smaller than $N$ as long as as long as $N^{-2/3} ( \beta \varrho^{2/3} )^{5/2} \ll 1$. It remains to estimate the off-diagonal contributions and to discuss the uniformity in the temperature.
\subsection{The off-diagonal of the one-particle density matrix and the final estimate}
In this section we are going to control the off-diagonal parts of $\gamma$ which will allow us to give the final estimate. Our analysis follows the lines of a similar analysis in \cite[Sec.~4.3]{me}. We write
\begin{equation}
	\Vert \gamma - \gamma_0 \Vert_1 \leq \Vert P \left( \gamma - \gamma_0 \right) P \Vert_1 + 2 \Vert P \left( \gamma - \gamma_0 \right) Q \Vert_1 + \Vert Q \left( \gamma - \gamma_0 \right) Q \Vert_1
	\label{eq:offdiag2}
\end{equation}
and estimate the right-hand side term by term. A bound for the first term on the right-hand side was given in \eqref{eq:offdiag11}. From \eqref{eq:asymptotics1pdm11d} we know that the last term is bounded by $N (\widetilde{c}_{\ell}(\beta,N,L) + \delta(N))^{1/2}$. 

To derive a bound on the second term on the right-hand side of \eqref{eq:offdiag2}, we use $P \gamma_0 Q = 0$ and write
\begin{equation}
	\Vert P \left( \gamma - \gamma_0 \right) Q \Vert_1  = \Vert P \gamma Q \Vert_1 \leq \Vert P_{\Phi} P \gamma Q \Vert_1 + \Vert \left(1 - P_{\Phi} \right) P \gamma Q \Vert_1.
	\label{eq:offdiag12}
\end{equation}
We estimate the first term on the right-hand side by
\begin{equation}
	\Vert P_{\Phi} P \gamma Q \Vert_1 \leq \Vert P \gamma Q \Vert \leq \Vert \gamma^{1/2} \Vert \ \Vert \gamma^{1/2} Q \Vert \leq N^{1/2} \Vert Q \gamma Q \Vert^{1/2} \leq N^{1/2} \left( \Vert Q (\gamma - \gamma_0) Q \Vert^{1/2} + \Vert Q \gamma_0 Q \Vert \right)^{1/2}. \label{eq:offdiag13}
\end{equation}
The first term in the bracket on the right-hand side of \eqref{eq:offdiag13} can be estimated by its trace norm which can be bounded with the help of \eqref{eq:asymptotics1pdm11d}. To bound the second term, we invoke Lemma~\ref{lem:cvsgcexpectations} in the appendix to see that it is bounded from above by a constant times $(\beta (p_{\mathrm{c}}^2 - \mu_0) )^{-1}$. Putting these two bounds together, we therefore have
\begin{equation}
	\Vert P_{\Phi} P \gamma Q \Vert_1 \lesssim N \left(  \widetilde{c}_{\ell}(\beta,N,L) + \delta(N) \right)^{1/4} + \left( \frac{N}{ \beta \left( p_{\mathrm{c}}^2 - \mu_0 \right) } \right)^{1/2}.
	\label{eq:offdiag13b}
\end{equation} 
To bound the second term on the right-hand side of \eqref{eq:offdiag12}, we write
\begin{align}
	\left\Vert \left(1 - P_{\Phi} \right) P \gamma Q  \right\Vert_1 \leq \left\Vert \left(1 - P_{\Phi} \right) P \gamma^{1/2} \right\Vert_2 \Vert \gamma^{1/2} Q \Vert_2 &\leq N^{1/2} \left( \trs \left(1 - P_{\Phi} \right) P \gamma P \right)^{1/2} \label{eq:offdiag14} \\
	&\leq N^{1/2} \left( \Vert P (\gamma - \gamma_0) P \Vert_1 + \Vert \widetilde{P} \gamma_0 \widetilde{P} \Vert_1 \right)^{1/2}. \nonumber
\end{align}	
The first term in the bracket on the right-hand side can be bounded with \eqref{eq:offdiag11}, the second with \eqref{eq:offdiag4}. We find
\begin{align}
	\left\Vert \left(1 - P_{\Phi} \right) P \gamma Q  \right\Vert_1 \lesssim N \left( c_{\ell}(N)^{1/4} + \widetilde{c}_{\ell}(\beta,N,L)^{1/8} + \delta(N)^{1/8} \right) &+ N^{3/4} \left( \frac{L^2}{\beta} + \frac{L^3 p_{\mathrm{c}}}{\beta} \right)^{1/4} \label{eq:offdiag14c} \\
	&+ N^{1/2} \left( \frac{L^2}{\beta} + \frac{L^3 p_{\mathrm{c}}}{\beta} \right)^{1/2}. \nonumber
\end{align}
By combining Eqs.~\eqref{eq:offdiag12}, \eqref{eq:offdiag13b} and \eqref{eq:offdiag14c}, we estimate the off-diagonal contribution to the 1-pdm by
\begin{align}
	\Vert P \left( \gamma - \gamma_0 \right) Q \Vert_1 \lesssim& N \left( \left( c_{\ell}(N)^{1/4} + \widetilde{c}_{\ell}(\beta,N,L)^{1/8} + \delta(N) \right)^{1/8} \right) + N^{3/4} \left( \frac{N^{2/3}}{\beta \varrho^{2/3}} + \frac{N \beta^{1/2} p_{\mathrm{c}}}{\left( \beta \varrho^{2/3} \right)^{3/2}} \right)^{1/4} \label{eq:offdiag14b} \\
	&+ N^{1/2} \left( \frac{N^{2/3}}{\beta \varrho^{2/3}} + \frac{N \beta^{1/2} p_{\mathrm{c}}}{\left( \beta \varrho^{2/3} \right)^{3/2}} \right)^{1/2} + \left( \frac{N}{ \beta \left( p_{\mathrm{c}}^2 - \mu_0 \right) } \right)^{1/2}. \nonumber
\end{align}
We now have everything together to state the final bound for $\gamma$. To that end, we combine \eqref{eq:asymptotics1pdm11d}, \eqref{eq:offdiag11} \eqref{eq:offdiag2} and \eqref{eq:offdiag14b}. Inserting also the explicit choice for $p_{\mathrm{c}}$ \eqref{eq:pc}, we find
\begin{equation}
	\Vert \gamma - \gamma_0 \Vert_1 \lesssim N \left( c_{\ell}(N)^{1/4} + \widetilde{c}_{\ell}(\beta,N,L)^{1/8} + \delta(N)^{1/8} \right).
	\label{eq:offdiag16}
\end{equation}
This proves the claimed asymptotics for the 1-pdm as long as $N^{-2/3} ( \beta \varrho^{2/3})^{5/2} \ll 1$. In the next section we discuss the uniformity in $\beta \varrho^{2/3}$. 
\subsection{Uniformity in the temperature}
In order to show the desired uniformity in the temperature, we have to consider the case where $\beta \varrho^{2/3}$ is so large that $\widetilde{c}_{\ell}(\beta,N,L)$ is no longer small, that is, $\beta \varrho^{2/3} \gtrsim N^{4/15}$. In this case we have $\varrho_0(\beta,N,L) \simeq \varrho$, and hence the contribution of the thermal cloud to the 1-pdm of the ideal gas is of lower order. In combination with \eqref{eq:1pdmcondensate10}, this will imply a similar statement for $\gamma$. In particular, it will allow us to conclude that $\Vert \gamma - \gamma_0 \Vert_1 \ll N$ uniformly in $\beta \varrho^{2/3}$.

Let $Q_{\Phi} = 1- P_{\Phi}$. From \eqref{eq:1pdmcondensate10} we know that
\begin{equation}
	N = \tr \gamma =  \langle \Phi | \gamma | \Phi \rangle + \trs Q_{\Phi} \gamma Q_{\Phi} \geq N_0 - N \left( c_{\ell}(N) + \delta(N) \right) + \trs Q_{\Phi} \gamma Q_{\Phi}. 
	\label{eq:offdiag17}
\end{equation}
An application of Lemma~\ref{lem:cvsgcexpectations} in the Appendix and one of Lemma~\ref{lem:latticesum1} tell us that
\begin{equation}
\trs \left[ Q_{\Phi} \gamma_0 \right] \lesssim \trs \left[ Q_{\Phi} \gamma_0^{\mathrm{gc}} \right] \lesssim \frac{ | \Lambda | }{\beta^{3/2}} + \frac{L^2}{\beta}.
\label{eq:asymptotics1pdm11e}
\end{equation}
Together with $N = N_0 + \tr Q_{\Phi} \gamma_0 Q_{\Phi}$ and \eqref{eq:offdiag17}, this bound implies 
\begin{align}
	\trs Q_{\Phi} \gamma Q_{\Phi} &\lesssim N \left( c_{\ell}(N) + \delta(N) \right) + \left( \frac{| \Lambda | }{\beta^{3/2}} + \frac{L^2}{\beta} \right) \quad \text{ as well as } \label{eq:asymptotics1pdm11f} \\
	\left\| Q_{\Phi} (\gamma - \gamma_0) Q_{\Phi} \right\|_1 &\lesssim N \left( c_{\ell}(N) + \delta(N) \right) + \left( \frac{| \Lambda | }{\beta^{3/2}} + \frac{L^2}{\beta} \right). \nonumber
\end{align}
With \eqref{eq:1pdmcondensate10}, \eqref{eq:asymptotics1pdm11e} and
\begin{equation}
	\langle \Phi | \gamma | \Phi \rangle \leq N = N_0 + \trs Q_{\Phi} \gamma_0 Q_{\Phi} \leq N_0 + \text{const. } \left( \frac{| \Lambda | }{\beta^{3/2}} + \frac{L^2}{\beta} \right)
	\label{eq:offdiag18}
\end{equation}
we additionally see that
\begin{equation}
	\left\| P \left( \gamma - \gamma_0 \right) P \right\|_1 = | \langle \Phi | \gamma - \gamma_0 | \Phi \rangle | \lesssim N \left( c_{\ell}(N) + \delta(N) \right) + \left( \frac{| \Lambda | }{\beta^{3/2}} + \frac{L^2}{\beta} \right).
	\label{eq:offdiag19}
\end{equation}
Using $P_{\Phi} \gamma_0 Q_{\Phi} = 0$, the off-diagonal contribution $\Vert P_{\Phi} (\gamma - \gamma_0) Q_{\Phi} \Vert_1$ can be estimated similarly to \eqref{eq:offdiag13} by
\begin{equation}
	\Vert P_{\Phi} \left( \gamma - \gamma_0 \right) Q_{\Phi} \Vert_1 = \Vert P_{\Phi} \gamma Q_{\Phi} \Vert_1 \leq N^{1/2} \left\| Q_{\Phi} \gamma Q_{\Phi} \right\|^{1/2} \lesssim N^{1/2} \left( N \left( c_{\ell}(N) + \delta(N) \right) + \frac{| \Lambda | }{\beta^{3/2}} + \frac{L^2}{\beta} \right)^{1/2}.
	\label{eq:offdiag20}
\end{equation} 
To obtain the second estimate, we additionally used \eqref{eq:asymptotics1pdm11f}. In combination Eqs.~\eqref{eq:asymptotics1pdm11f}, \eqref{eq:offdiag19} and \eqref{eq:offdiag20} imply for $\beta \varrho^{2/3} \gg 1$
\begin{equation}
	\left\| \gamma - \gamma_0 \right\|_1 \lesssim N \left[ \left( c_{\ell}(N) + \delta(N) \right)^{1/2} + \frac{ 1 }{ \left( \beta \varrho^{2/3} \right)^{3/4} } + \frac{N^{-1/6}}{\left( \beta \varrho^{2/3} \right)^{1/2}} \right].
	\label{eq:offdiag21}
\end{equation}
This bound needs to be combined with \eqref{eq:offdiag16} in order to obtain a bound that is uniform in $\beta \varrho^{2/3} \gtrsim 1$. The relevant terms depending on $\beta \varrho^{2/3}$ to consider are $N \widetilde{c}_{\ell}(\beta,N,L)^{1/8}$ in \eqref{eq:offdiag16} and $(\beta \varrho^{2/3})^{-3/4}$ in \eqref{eq:offdiag21}. We use \eqref{eq:offdiag16} as long as $\beta \varrho^{2/3} \leq N^{4/7269}$ and \eqref{eq:offdiag21} otherwise. The largest error term is $N c_{\ell}(N)^{1/4}$ in \eqref{eq:offdiag16}. We therefore have
\begin{equation}
	\left\| \gamma - \gamma_0 \right\|_1 \leq C_{\alpha} \left( \beta \varrho^{2/3} \right) N \left( N^{ -1/6884 + \alpha} + \delta(N)^{1/8} \right)
	\label{eq:offdiag22}
\end{equation}
with $\alpha>0$ and some function $s \mapsto C_{\alpha}(s)$ that is uniformly bounded on intervals $[d,\infty)$ with $d>0$. For $\alpha \to 0$, the function $C_{\alpha}$ blows up. Our bound is uniform in $0<a_v \lesssim 1$. This concludes the proof of the asymptotics of the 1-pdm of approximate minimizers of the Gibbs free energy functional, and therewith also the proof of Theorem~\ref{thm:periodic}. 
\appendix
\section{Some properties of the ideal Bose gas}
\label{sec:appendix1}
In this appendix we collect three Lemmas concerning properties of the ideal Bose gas, which have been proven in \cite[Appendix~A]{me} or follow from a statement there, and which we need in the main text. In particular, they concern the comparison of relevant quantities when computed in the canonical and in the grand canonical ensemble. Although these statements hold more generally, we state them here only for the ideal Bose gas on the torus $\Lambda$.

As in the main text we denote by $F_0(\beta,N,L,\lambda)$ the canonical free energy of the ideal gas related to the Hamiltonian \eqref{eq:1pdmcondensate1a} with $v=0$ and by $F_0^{\mathrm{gc}}(\beta,\mu,L,\lambda)$ its grand canonical analogue. We recall that $\lambda \geq 0$ denotes the shift of the lowest eigenvalue of the Laplacian on $\Lambda$. Similarly, $\langle \cdot \rangle_{N}$ and $\langle \cdot \rangle_{\mathrm{gc},\mu}$ denote the expectations and $N_{0}$ and $N_{0}^{\mathrm{gc}}$ the expected number of particles in the condensate in the two ensembles (for simplicity we have suppressed the $\lambda$-dependence here). The expected number of particles in the grand canonical ensemble is denoted by $\overline{N}(\mu)$ and $\gamma_0$/$\gamma_0^{\mathrm{gc}}$ is the 1-pdm of the canonical/grand canonical ideal gas (which depend on $\lambda$). The following three statements hold.
\begin{lemma}
	\label{lem:freeenergy}
	Assume $\mu$ is such that $\overline{N} = N \in \mathbb{N}$. Then
	\begin{equation}
	F_0(\beta,N,L,\lambda) \geq F_0^{\mathrm{gc}}(\beta,\mu,L,\lambda) \geq F_0(\beta,N,L,\lambda) - \tfrac{1}{\beta} \left( \ln(1+\overline{N}) + 1 \right).
	\label{eq:idealgas15}
	\end{equation}
\end{lemma}
\begin{proof}
	The proof follows from \cite[Corollary~A.1]{me}.
\end{proof}
\begin{lemma}
	\label{lem:cvsgcexpectations}
	Assume $\mu$ is such that $\overline{N} = N \in \mathbb{N}$ and let $f : \mathbb{N}_0 \mapsto \mathbb{R}$ be a nonnegative and nondecreasing function. Then
	\begin{equation}
		\left\langle f\left( a_p^* a_p \right) \right\rangle_N \leq \frac{40}{1.8} \left\langle f \left( a_p^* a_p \right) \right\rangle_{\mathrm{gc},\mu}
	\end{equation}
	holds for all $p \in \frac{2 \pi}{L} \mathbb{Z}^3$.
\end{lemma}
\begin{proof}
	The proof follows the proof of a similar statement for the densities of the system in \cite[Proposition~A.2]{me}. See also Remark~A.1 in the same reference.
\end{proof}
\begin{lemma}
	\label{lem:comparisoncondensatedensities}
	Denote $\widetilde{\gamma}_0 = \mathds{1}(-\Delta \neq 0) \gamma_0$ and the same for the grand canonical 1-pdm and choose $\mu$ such that $\overline{N} = N \in \mathbb{N}$ holds. Then
	\begin{equation}
		\left| N_0 - N_0^{\mathrm{gc}} \right| \leq \left\| \widetilde{\gamma}_0 - \widetilde{\gamma}_0^{\mathrm{gc}} \right\|_1 \leq  \left\| \gamma_0 - \gamma_0^{\mathrm{gc}} \right\|_1 \lesssim \left( \frac{N \ln(N) L^2}{\beta} \right)^{1/2} + \frac{\ln(N) L^2}{\beta}.
	\end{equation}
\end{lemma}
\begin{proof}
	The proof follows from \cite[Lemma~A.2]{me}.
\end{proof}
\textbf{Acknowledgments.} 
It is a pleasure to thank Jakob Yngvason for helpful discussions. Financial support by the European Research Council (ERC) under the European Union's Horizon 2020 research and innovation programme (grant agreement No 694227) is gratefully acknowledged. A. D. acknowledges funding from the European Union’s Horizon 2020 research and innovation programme under the Marie Sklodowska-Curie grant agreement No 836146.

\vspace{0.5cm}

(Andreas Deuchert) Institute of Science and Technology Austria (IST Austria)\\ Am Campus 1, 3400 Klosterneuburg, Austria \newline
\textit{present address}: Institute of Mathematics, University of Zurich, Winterthurerstrasse 190, 8057 Zurich, \newline Switzerland \\ 
E-mail address: \texttt{andreas.deuchert@math.uzh.ch}

(Robert Seiringer) Institute of Science and Technology Austria (IST Austria)\\ Am Campus 1, 3400 Klosterneuburg, Austria\\ E-mail address: \texttt{robert.seiringer@ist.ac.at}


\begin{thebibliography}{49}
\addcontentsline{toc}{section}{References}

\bibitem{WieCor1995} M.H. Anderson, J.R. Ensher, M.R. Matthews, C.E. Wieman, E.A. Cornell, \textit{Observation of bose-einstein condensation in a dilute atomic vapor}, Science \textbf{269}, 198 (1995)

\bibitem{BECprogram} T. Balaban, J. Feldman, H. Kn\"orrer, E. Trubowitz, \textit{Complex Bosonic Many-Body Models: Overview of the Small Field Parabolic Flow}, Ann. Henri Poincar\'e \textbf{18}, 2873 (2017)

\bibitem{BenOlivSchl2015} N. Benedikter, G. de Oliveira, B. Schlein, \textit{Quantitative Derivation of the Gross-Pitaevskii Equation}, Comm. Pure Appl. Math. \textbf{68}, 1399 (2015)

\bibitem{BenPorSchl2015} N. Benedikter, M. Porta, B. Schlein, \textit{Effective Evolution Equations from Quantum Dynamics}, Springer, Berlin (2016) 

\bibitem{BogGP} C. Boccato, C. Brennecke, S. Cenatiempo, B. Schlein, \textit{Bogoliubov Theory in the Gross-Pitaevskii Limit}, Acta Mathematica \textbf{222}, 219 (2019)

\bibitem{Kett1995} K.B. Davis, M.-O. Mewes, M.R. Andrews, N.J. van Druten, D.S. Durfee, D.M. Kurn, W. Ketterle, \textit{Bose-Einstein Condensation in a Gas of Sodium Atoms}, 
Phys. Rev. Lett. \textbf{75}, 3969 (1995)

\bibitem{me} A. Deuchert, R. Seiringer, J. Yngvason, \textit{Bose-Einstein Condensation for a Dilute Trapped Gas at Positive Temperature}, Commun. Math. Phys. \textbf{368}, 723 (2019) 

\bibitem{Dyson} F.J. Dyson, \textit{Ground-State Energy of a Hard-Sphere Gas}, Phys. Rev. \textbf{106}, 20 (1957)

\bibitem{ErdSchlYau2009} L. Erd\H os, B. Schlein, H.-T. Yau, \textit{Rigorous derivation of the Gross-Pitaevskii equation with a large interaction potential}, J. Amer. Math. Soc. \textbf{22}, 1099 (2009)

\bibitem{ErdSchlYau2010} L. Erd\H os, B. Schlein, H.-T. Yau, \textit{Derivation of the Gross-Pitaevskii equation for the dynamics of Bose-Einstein condensate}, Ann. of Math. \textbf{172}, 291 (2010)

\bibitem{FournSol2019} S. Fournais, J.P. Solovej, \textit{The energy of dilute Bose gases}, arXiv:1904.06164 [math-ph] (2019) 

\bibitem{FKSS2017} J. Fr\"ohlich, A. Knowles, B. Schlein, V. Sohinger, 
\textit{Gibbs Measures of Nonlinear Schr\"odinger Equations as Limits of Many-Body Quantum States in Dimensions 
	$d\leq 3$}, Commun. Math. Phys. \textbf{356}, 883 (2017)

\bibitem{FKSS2020} J. Fr\"ohlich, A. Knowles, B. Schlein, V. Sohinger, \textit{The mean-field limit of quantum Bose gases at positive temperature}, arXiv:2001.01546 [math-ph] (2020).

\bibitem{Gauntetal2013} A.L. Gaunt, T.F. Schmidutz, I. Gotlibovych, R.P. Smith, Z. Hadzibabic, \textit{Bose-Einstein Condensation of Atoms in a Uniform Potential}, Phys. Rev. Lett. \textbf{110}, 200406 (2013)

\bibitem{Jastrow} R. Jastrow, \textit{Many-Body Problem with Strong Forces}, Phys. Rev. \textbf{98}, 1479 (1955)

\bibitem{LeeHuangYang1957} T.D. Lee, K. Huang, C.N. Yang, \textit{Eigenvalues and eigenfunctions of a bose system of hard spheres and its low-temperature properties}, Physical Review, \textbf{106}, 1135 (1957)

\bibitem{LewinNamRougerie2015} M. Lewin, P.T. Nam, N. Rougerie, \textit{Derivation of nonlinear Gibbs measures from many-body quantum mechanics}, J. \'Ec. polytech. Math. \textbf{2}, 65 (2015)

\bibitem{LewinNamRougerie2017} M. Lewin, P.T. Nam, N. Rougerie, \textit{Gibbs measures based on 1D (an)harmonic oscillators as mean-field limits}, Journal of Mathematical Physics \textbf{59}, 041901 (2018)

\bibitem{LewinNamRougerie2018} M. Lewin, P.T. Nam, N. Rougerie, \textit{Classical field theory limit of 2D many-body quantum Gibbs states}, arXiv:1810.08370 [math.AP] (2018) 

\bibitem{LiSei2002} E.H. Lieb, R. Seiringer, \textit{Proof of Bose-Einstein Condensation for Dilute Trapped Gases}, Phys. Rev. Lett. \textbf{88}, 170409 (2002)

\bibitem{LiSei2006} E.H. Lieb, R. Seiringer, \textit{Derivation of the Gross-Pitaevskii equation for rotating Bose gases}, Commun. Math. Phys. \textbf{264}, 505 (2006)

\bibitem{LiSeiSol2005} E.H. Lieb, R. Seiringer, J. P. Solovej, \textit{Ground State Energy of the Low Density Fermi Gas}, Phys. Rev. A \textbf{71}, 053605 (2005)

\bibitem{Themathematicsofthebosegas} E.H. Lieb, R. Seiringer, J. P. Solovej, J. Yngvason, \textit{The Mathematics of the Bose Gas and its Condensation}, Birkh\"auser, Basel (2005)

\bibitem{RobertGPderivation} E.H. Lieb, R. Seiringer, J. Yngvason, \textit{Bosons in a trap: A rigorous derivation of the Gross-Pitaevskii energy functional}, Phys. Rev. A \textbf{61}, 043602 (2000)

\bibitem{Lieb2005} E.H. Lieb, R. Seiringer, J. Yngvason, \textit{Justification of c-Number Substitutions in Bosonic Hamiltonians}, Phys. Rev. Lett. 94, 080401 (2005).

\bibitem{LiYng1998} E.H. Lieb, J. Yngvason, \textit{Ground State Energy of the Low Density Bose Gas}, Phys. Rev. Lett. \textbf{80}, 2504 (1998)

\bibitem{NamRouSei2016} P.T. Nam, N. Rougerie, R. Seiringer, \textit{Ground states of large bosonic systems: the Gross-Pitaevskii limit revisited}, Anal. PDE \textbf{9}, 459 (2016)

\bibitem{Pickl2015} P. Pickl, \textit{Derivation of the time dependent Gross Pitaevskii equation with external fields}, Rev. Math. Phys. \textbf{27}, 1550003 (2015)

\bibitem{Rou2015} N. Rougerie, \textit{De Finetti theorems, mean-field limits and Bose-Einstein condensation}, arXiv:1506.05263, Lecture notes, (2015)

\bibitem{RobertFermigas} R. Seiringer, \textit{The Thermodynamic Pressure of a Dilute Fermi Gas}, Comm. Math. Phys. \textbf{261}, 729 (2006)

\bibitem{RobertCorrelationinequ} R. Seiringer, \textit{A correlation estimate for quantum many-body systems at positive temperature}, Rev. Math. Phys. \textbf{18}, 233 (2006)

\bibitem{Sei2008} R. Seiringer, \textit{Free Energy of a Dilute Bose Gas: Lower Bound}, Comm. Math. Phys. \textbf{279}, 595 (2008)

\bibitem{Suto} A. S\"ut\H{o}, \textit{Correlation inequalities for noninteracting Bose gases}, J. Phys. A: Math. Gen. \textbf{37}, 3 (2004)

\bibitem{Tamm_etal 2011} N. Tammuz, R.P. Smith, R.L.D. Campbell, S. Beattie,  S. Moullder, \textit{Can an Bose Gas be Saturated?}, Phys. Rev. Lett. \textbf{106}, 230401 (2011)

\bibitem{Thirring_4} W. Thirring, \textit{Quantum Mathematical Physics}, $2^{nd}$ ed., Springer, New York (2002)

\bibitem{YauYin2009} H.-T. Yau, J. Yin, \textit{The Second Order Upper Bound for the Ground Energy of a Bose Gas}, J. Stat. Phys. \textbf{136}, 453 (2009)

\bibitem{Yin2010} J. Yin, \textit{Free Energies of Dilute Bose Gases: Upper Gound}, J. Stat. Phys. \textbf{141}, 683 (2010)

\end{thebibliography}
\end{document}